\documentclass[final]{siamltex}
\usepackage{epsfig}
\usepackage{graphicx}
\usepackage{amsmath}
\usepackage{amssymb}
\usepackage{setspace}
\usepackage{enumerate}
\usepackage{graphicx}
\usepackage{array}
\newcolumntype{C}[1]{>{\centering\let\newline\\\arraybackslash\hspace{0pt}}m{#1}}
\usepackage[ruled,linesnumbered,vlined]{algorithm2e}

\title{Spherical Conformal Parameterization of Genus-0 Point Clouds for Meshing}

\author{Gary Pui-Tung Choi, Kin Tat Ho and Lok Ming Lui}

\begin{document}

\maketitle

\begin{abstract}
Point cloud is the most fundamental representation of 3D geometric objects. Analyzing and processing point cloud surfaces is important in computer graphics and computer vision. However, most of the existing algorithms for surface analysis require connectivity information. Therefore, it is desirable to develop a mesh structure on point clouds. This task can be simplified with the aid of a parameterization. In particular, conformal parameterizations are advantageous in preserving the geometric information of the point cloud data. In this paper, we extend a state-of-the-art spherical conformal parameterization algorithm for genus-0 closed meshes to the case of point clouds, using an improved approximation of the Laplace-Beltrami operator on data points. Then, we propose an iterative scheme called the North-South reiteration for achieving a spherical conformal parameterization. A balancing scheme is introduced to enhance the distribution of the spherical parameterization. High quality triangulations and quadrangulations can then be built on the point clouds with the aid of the parameterizations. Also, the meshes generated are guaranteed to be genus-0 closed meshes. Moreover, using our proposed spherical conformal parameterization, multilevel representations of point clouds can be easily constructed. Experimental results demonstrate the effectiveness of our proposed framework.
\end{abstract}

\begin{keywords}
Mesh generation, Triangulation, Quadrangulation, Spherical conformal parameterization, Surface reconstruction, Point cloud, Multilevel representation
\end{keywords}


\pagestyle{myheadings}
\thispagestyle{plain}
\markboth{Choi, Ho and Lui}{Spherical Conformal Parameterization for Meshing}

\section{Introduction}
\label{intro}
Contemporary scanning technologies enable efficient acquisitions of 3D objects. Using modern 3D scanners, data points are sampled from the surfaces of 3D objects for further analyses and usages. Point clouds are widely applied in computer graphics, vision and many other engineering fields. However, the data points acquired by laser scanners are often complex and unorganized. Moreover, the absence of the connectivity information in point cloud data poses difficulties in understanding the underlying geometry of the 3D objects. This largely hinders the applications of the data. For instance, many applications in 3D printing \cite{Rengier10,Lipson13} and texture mapping \cite{Sander01,Levy01} are built upon mesh structures. With the rapid development of the computer industry, finding a high quality meshing framework for point cloud data is increasingly important.

One possible approach for mesh generation on point clouds is to parameterize a point cloud to a simpler domain with the corresponding genus, such as the unit sphere for genus-0 point clouds. Then, a triangulation or a quadrangulation can be created on the parameter domain instead of the original complicated point cloud. Finally, a mesh structure on the point cloud can be defined with respect to the structure on the parameter domain. The major difficulty of computing parameterizations of point-set surfaces is the extremely limited information they can provide. Most of the existing surface parameterization methods are developed on meshes only. In other words, besides the locations of the point data, a given connectivity is also required as an input. The connectivity information plays an important role in representing the surface structure as well as in approximating continuous operators to minimize certain distortions. As a result, most conventional mesh parameterization approaches fail to work on point clouds. Without the connectivity information, the underlying geometry of the point cloud data become more obscure. Hence, it is more challenging in developing parameterization schemes with good quality for point cloud data.

A good parameterization scheme of point cloud must satisfy certain criteria. In particular, it should retain the geometric information of a point cloud as complete as possible. In our case, one of the ultimate goals is to create a triangulation for a point cloud by finding a Delaunay triangulation on a simpler parameter domain. It is noteworthy that in general, a mesh structure with good quality on the parameter domain does not necessarily imply that the associated mesh structure of the original data points is satisfactory. In other words, meshing the parameter domain may provide meaningless results if the parameterization scheme is arbitrarily chosen. Note that the regularity of the mesh structures is related to the angle structure of the triangles and quadrilaterals. To ensure the regularity of the associated mesh structure on the point cloud, the parameterization should preserve the angle structure of the triangles and quadrilaterals on the parameter domain. This motivates us the use of conformal mappings.

For smooth surfaces, it is well known that the conformal parameterizations preserve angles and hence the local geometry of the surfaces. It is natural to consider the discrete analog of conformal parameterization for point cloud data. Since data points are sampled from real 3D surfaces, we can assume that every point cloud has an underlying geometry. Based on this important assumption, we consider finding conformal parameterizations of genus-0 point clouds. In \cite{Choi15a}, Choi {\it et al.} proposed a fast spherical conformal parameterization algorithm for genus-0 closed surfaces in two steps. In the first step, a Laplace equation is solved on a planar triangular domain and the inverse stereographic projection is applied to obtain an initial spherical parameterization. In the second step, quasi-conformal theories are applied to enhance the conformality of the spherical parameterization. The computation is linear and the conformality distortion of the parameterization is minimal. However, the algorithm is developed on triangular meshes only. In this work, we extend and improve the algorithm for point clouds with spherical topology.

The aforementioned algorithm in \cite{Choi15a} developed on meshes involves solving a Laplace equation. To extend the algorithm for point clouds, we propose a new weight function for enhancing the accuracy of the approximation of the Laplace-Beltrami (LB) operators on point clouds. Using our improved approximation, the LB operator in the mentioned algorithm can be accurately computed on point clouds. Also, we replace a key step of solving for a quasi-conformal map in the mentioned algorithm by an iterative scheme, called the {\it North-South reiteration}, for improving the conformality of the parameterizations. Furthermore, we introduce a balancing scheme for enhancing the distribution of the parameterization results. Experimental results demonstrate the effectiveness of our proposed parameterization algorithm for genus-0 point clouds. Our algorithm achieves global spherical parameterizations with minimal conformality distortions. Furthermore, with the aid of our parameterization schemes, we can easily generate high quality triangulations and quadrangulations on point clouds. The meshes generated are guaranteed to be genus-0 closed meshes. Moreover, multilevel representations of the point clouds can also be easily computed with the aid of our spherical parameterization scheme.

The rest of the paper is organized as follows. The contribution of our work is highlighted in Section \ref{contribution}. In Section \ref{previous}, we review the related previous works on point cloud parameterizations and approximations of differential operators on point clouds. In Section \ref{background}, we introduce the mathematical background of our work. In Section \ref{overview}, we review a spherical conformal parameterization scheme for triangular meshes, which is closely related to our proposed framework for point clouds. In Section \ref{mainsection}, we explain our proposed framework for spherical conformal parameterization and mesh generation of point clouds. In Section \ref{experiment}, we demonstrate the effectiveness of our proposed framework by numerous experiments. The paper is concluded in Section \ref{conclusion}.

\section{Contribution} \label{contribution}
In this work, we propose a framework for meshing using spherical conformal parameterizations of genus-0 point clouds. Our proposed method is advantageous in the following aspects:
\begin{enumerate}[(i)]
 \item We extend and improve the spherical conformal parameterization algorithm on meshes in \cite{Choi15a} for point clouds. An accurate approximation of the Laplace-Beltrami operator is achieved using the moving least square method \cite{Lange05,Liang12,Liang13} together with a new Gaussian-type weight function. A key step of the parameterization algorithm in \cite{Choi15a} for computing a quasi-conformal map is replaced by solving a Laplace equation on the complex plane, followed by an iterative scheme called the North-South (N-S) reiteration. Also, the point distribution of the parameterization is enhanced by a balancing scheme for point clouds.
 \item Our spherical parameterization method is efficient and robust to complex geometric structures. The algorithm completes within a few minutes and can handle highly convoluted point clouds.
 \item Unlike most of the existing approaches, our algorithm specifically minimizes the conformality distortion of the parameterizations. Since the local geometry is preserved under the global spherical conformal parameterizations, we can create an almost-Delaunay triangulation on a point cloud by computing a Delaunay triangulation of its spherical conformal parameterization. The resulting triangulation on the point cloud preserves the regularity of that on the parameterization.
 \item High quality quad meshes can also be generated using our spherical conformal parameterization scheme.
 \item Unlike the conventional approaches for meshing, our method is topology preserving. The meshes produced using our proposed framework are guaranteed to be genus-0 closed meshes. No post-processing is required.
 \item Our method is stable under geometrical and topological noises on the input point clouds.
 \item With the aid of our spherical conformal parameterization scheme, multilevel representations of genus-0 point clouds can be easily constructed.
\end{enumerate}

\begin{table}[t]
    \centering
    \begin{tabular}{ |C{40mm}|C{15mm}|C{17mm}|C{30mm}|C{22mm}| }
    \hline
    Methods & Topology & Parameter domain & Local/global parameterization? & Distortion to be minimized\\ \hline
    Meshless parameterization \cite{Floater00,Floater01} & Disk topology & Plane & Global & / \\ \hline
    Meshless parameterization for Spherical Topology \cite{Hormann03} & Genus-0 & Planes & Local & /  \\ \hline
    Spherical embedding \cite{Zwicker04} & Genus-0 & Sphere & Global & Stretch \\ \hline
    Discrete one-forms \cite{Tewari06} & Genus-1 & Planes & Local & / \\ \hline
    As-rigid-as-possible meshless parameterization \cite{Zhang12}  & Disk topology & Plane & Global & ARAP \\ \hline
    Meshless quadrangulation by global parameterization \cite{Li11} & Unrestricted & Plane & Global & Gradient and principal fields\\ \hline
    \end{tabular}
    \bigbreak
    \caption{Several previous works on meshing point clouds using parameterization.}
    \label{previouswork}
\end{table}

\section{Previous Works} \label{previous}

In this section, we describe some previous works closely related to our work.

Surface parameterization has been extensively studied by different research groups. For surveys on surface parameterization methods, please refer to \cite{Floater02,Floater05,Hormann07,Sheffer06}. In particular, conformal parameterizations have been well established on meshes. For the recent works, Lai {\it et al.} \cite{Lai14} proposed a folding-free spherical conformal mapping scheme by the harmonic energy minimization. In \cite{Aflalo13}, Aflalo {\it et al.} proved theoretical bounds of the conformal factor and proposed a method that minimizes the area distortion and avoids numerical errors of the conformal mapping. In \cite{Choi15a}, Choi {\it et al.} developed a linear algorithm for spherical conformal parameterizations of genus-0 closed meshes.

In the last few decades, numerous studies have been devoted to the parameterization of point cloud data. In \cite{Floater00,Floater01}, Floater and Reimers proposed the meshless parameterization method for unorganized point sets. The point sets are parameterized onto a planar domain by solving a sparse linear system. In \cite{Zwicker04}, Zwicker and Gotsmann presented a parameterization approach for a genus-0 point cloud using a $k$-nearest neighborhood graph of the point cloud, followed by a spherical embedding method for planar graphs. In \cite{Azariadis04,Azariadis05,Azariadis07}, Azariadis and Sapidis introduced the notion of dynamic base surfaces and suggested a parameterization scheme by orthogonally projecting a point cloud onto the dynamic base surface. Guo {\it et al.} \cite{Guo06} computed a global conformal parameterization of point-set surfaces, based on Riemann surface theory and Hodge theory. In \cite{Tewari06}, Tewari {\it et al.} proposed a doubly-periodic global parameterization of point cloud sampled from a closed surface of genus 1 to the plane, with the aid of discrete harmonic one-forms. Wang {\it et al.} \cite{Wang08} suggested a parameterization method for genus-0 cloud data. A point cloud is first mapped onto its circumscribed sphere, then the sphere is mapped onto an octahedron and finally unfolded to a 2D image. In \cite{Zhang12}, Zhang {\it et al.} presented an as-rigid-as-possible parameterization approach for point cloud data. A point cloud with disk topology is mapped onto the plane by a local flattening step and a rigid alignment. In \cite{Liang12}, Liang {\it et al.} constructed spherical conformal mappings of genus-0 point clouds by adapting the harmonic energy minimization algorithm in \cite{Lai14}. Meng {\it et al.} \cite{Meng13} proposed a neural network based method for point cloud parameterization. An adaptive sequential learning algorithm is applied to dynamically adjust the parameterization.

The use of parameterization of point cloud is widespread in computer science and engineering. One of the major applications of point cloud parameterization is mesh generation. Instead of a convoluted point cloud, mesh reconstruction is usually completed on a simpler parameter domain. In \cite{Floater00,Floater01}, Floater and Reimers applied their proposed parameterization scheme for meshing point clouds with disk topology. In \cite{Hormann03}, Hormann and Reimers extended the parameterization method in \cite{Floater01} for surface reconstruction of point clouds with spherical topology. In \cite{Zwicker04}, Zwicker and Gotsmann used their proposed parameterization method for mesh reconstruction of genus-0 point clouds. Tewari {\it et al.} \cite{Tewari06} performed surface reconstruction using their proposed doubly-periodic global parameterization. Li {\it et al.} \cite{Li11} proposed a meshless quadrangulation scheme by global parameterization. The input point cloud is cut to be with disk topology and parameterized onto the plane for meshing. Zhang {\it et al.} \cite{Zhang12} suggested a mesh reconstruction method of point cloud data by meshless denoising and their proposed parameterization scheme. Table \ref{previouswork} compares several previous works on meshing point clouds using parameterizations. The above previous works reflect the importance of parameterization in surface reconstruction of point cloud data.

Finding a conformal parameterization involves solving differential equations. In particular, for conformal parameterizations of point clouds, it is necessary to build a discrete analog of the differential operators on point clouds. Numerous works on approximating differential operators on point cloud have been reported. In \cite{Nayroles92}, Nayroles {\it et al.} described a diffuse approximation method for estimating the derivatives at a given set of points. In \cite{Belkin08}, Belkin and Niyogi established a theoretical foundation for the Laplace-Beltrami operator on point clouds. Belkin {\it et al.} \cite{Belkin09} proposed the PCD Laplace operator for approximating the LB operator using an integral approximation. The moving least square (MLS) method \cite{Shepard68,Lancaster81} is widely used for the approximation. A number of algorithms for the approximation of derivatives are developed based on the MLS method \cite{Liu97,Levin98,Nealen04,Breitkopf05}. In \cite{Lange05}, Lange and Polthier proposed a point set analogue of the Laplace-Beltrami and shape operator using the MLS method. In \cite{Liang12,Liang13}, Liang {\it et al.} approximated the LB operator on point clouds by the MLS method with a special weighting function. In \cite{Lai13}, Lai {\it et al.} presented a local mesh approach for solving PDEs on point clouds. A local mesh structure is constructed at each point using local principal component analysis (PCA). Macdonald {\it et al.} \cite{Macdonald14} computed reaction-diffusion processes on point clouds. In \cite{Lozes14}, Lozes {\it et al.} proposed a method to solve PDEs on point clouds for image processing using partial difference operators on weighted graphs.

\section{Mathematical background} \label{background}

In this section, we introduce some basic mathematical concepts closely related to our work. For more details, readers are referred to \cite{Schoen94,Schoen97,Jost11}.

\subsection{Conformal maps}
An altas of a manifold is said to be \emph{conformal} if all of its transition maps are biholomorphic. A conformal structure is the maximal conformal altas, and a surface with a conformal structure is called a \emph{Riemann surface}. Suppose $\mathcal{M}$, $\mathcal{N}$ are two Riemann surfaces with local coordinate systems $r_1(x^1,x^2)$ and $r_2(x^1, x^2)$, where $r_1,r_2:\mathbb{R}^2 \to \mathbb{R}^3$ are vector-valued functions. The first fundamental forms of $\mathcal{M}$ and $\mathcal{N}$ are respectively defined by
\begin{equation}
ds_{\mathcal{M}}^2 = \sum_{i,j} g_{ij} dx^i dx^j \ \ \text{ and } \ \ ds_{\mathcal{N}}^2 = \sum_{i,j} \tilde{g}_{ij} dx^i dx^j,
\end{equation}
where $\displaystyle g_{ij} = \left\langle \frac{\partial r_1}{\partial x^i}, \frac{\partial r_1}{\partial x^j} \right\rangle$, $\displaystyle \tilde{g}_{ij} =\left\langle \frac{\partial r_2}{\partial x^i}, \frac{\partial r_2}{\partial x^j}\right\rangle$. Consider $f:\mathcal{M} \to \mathcal{N}$. In local coordinate systems, $f$ can be regarded as $f: \mathbb{R}^2 \to \mathbb{R}^2$, with $f(x^1,x^2)  =(f^1(x^1,x^2), f^2(x^1,x^2))$. The pull-back metric $f^*ds_{\mathcal{N}}^2$ defined on $\mathcal{M}$, induced by $f$ and $ds_{\mathcal{N}}^2$, is the metric
\begin{equation}
 f^*ds_{\mathcal{N}}^2 = \sum_{m,n} \left(\sum_{i,j} \tilde{g}_{ij} (f(x^1,x^2)) \frac{\partial f^m}{\partial x^i}\frac{\partial f^n}{\partial x^j}\right) dx^m dx^n.
\end{equation}
$f$ is said to be \emph{conformal} if there exists a positive scalar function $\lambda(x^1,x^2)$, called the {\it conformal factor}, such that $f^*ds_{\mathcal{N}}^2 = \lambda ds_{\mathcal{M}}^2$. An immediate consequence of the above is that every conformal map preserves angles and hence the local geometry of the surface.

\subsection{Harmonic maps}
By the uniformization theorem, every genus-0 closed surface is conformally equivalent to $\mathbb{S}^2$. Hence, it is natural to consider mappings between a genus-0 closed surface and the unit sphere. The {\it Dirichlet energy} for a map $f: {\mathcal{M}} \to \mathbb{S}^2$ is defined as
\begin{equation}\label{eqt:harmonic}
E(f) = \int_{\mathcal{M}} |\nabla f|^2 dv_{\mathcal{M}}.
\end{equation}
In the space of mappings, the critical points of $E(f)$ are called {\it harmonic mappings}. For genus-0 closed surfaces, conformal maps are equivalent to harmonic maps \cite{Jost11}. Hence, the problem of finding a conformal map $f: {\mathcal{M}} \to \mathbb{S}^2$ is equivalent to an energy minimization problem.

\subsection{Quasi-conformal maps}
Quasi-conformal maps are a generalization of conformal maps. Mathematically, $f: \mathbb{C} \to \mathbb{C}$ is a quasi-conformal map if it satisfies the Beltrami equation:
\begin{equation}\label{eqt:beltrami}
\frac{\partial f}{\partial \overline{z}} = \mu(z) \frac{\partial f}{\partial z}
\end{equation}
for some complex-valued function $\mu$ satisfying $||\mu||_{\infty}< 1$ and $\frac{\partial f}{\partial z}$ is non-vanishing almost everywhere. Here, the complex partial derivatives are defined by
\begin{equation}
 \frac{\partial f}{\partial z} := \frac{1}{2} \left(\frac{\partial f}{\partial x}  - i \frac{\partial f}{\partial y} \right)
 \ \ \text{ and } \ \
 \frac{\partial f}{\partial \overline{z}} := \frac{1}{2} \left(\frac{\partial f}{\partial x}  + i \frac{\partial f}{\partial y} \right).
\end{equation}

$\mu$ is called the \emph{Beltrami coefficient} of the quasi-conformal map $f$. Note that the quasi-conformal map $f$ is conformal around a small neighborhood of $p$ if and only if $\mu(p) = 0$, as Equation (\ref{eqt:beltrami}) becomes the Cauchy-Riemann equation in this situation.

Suppose $f:\Omega_1\to \Omega_2$ and $g:\Omega_2\to\Omega_3$ are quasi-conformal maps with the Beltrami coefficients $\mu_f$ and $\mu_g$ respectively. Then, the Beltrami coefficient of the composition map $g\circ f :\Omega_1\to \Omega_3$ is explicitly given by
\begin{equation}\label{eqt:composition_formula}
\mu_{g\circ f} = \frac{\mu_f + \frac{\overline{f_z}}{f_z} (\mu_g\circ f)}{1+\frac{\overline{f_z}}{f_z} \overline{\mu_f}(\mu_g\circ f)}.
\end{equation}

Quasi-conformal maps are also defined between two Riemann surfaces $M$ and $N$. A \emph{Beltrami differential} $\mu(z)\frac{\overline{dz}}{dz}$ on $M$ is an assignment to each chart $(U_{\alpha}, \phi_{\alpha})$ of an $L^{\infty}$ complex-valued function $\mu_{\alpha}$ defined on the local parameter $z_{\alpha}$ such that $\mu_{\alpha}(z_{\alpha}) \frac{ \overline{dz_\alpha}}{dz_\alpha} = \mu_{\beta}(z_{\beta}) \frac{ \overline{dz_{\beta}}}{dz_{\beta}}$ on the domain also covered by another chart $(U_\beta, \psi_\beta)$, where
$\frac{d z_{\beta}}{d z_{\alpha}} = \frac{d}{d z_{\alpha}} \phi_{\alpha \beta}$ and $\phi_{\alpha \beta} = \phi_{\beta} \circ \phi_{\alpha}^{-1}$. An orientation preserving diffeomorphism $f: M \to N$ is called {\it quasi-conformal} associated with $\mu(z)\frac{d\overline{z}}{dz}$ if for any chart $(U_{\alpha}, \phi_{\alpha})$ on $M$ and any chart $(U_\beta, \psi_\beta)$ on $N$, the mapping $f_{\alpha \beta} := \psi_{\beta} \circ f \circ f_{\alpha}^{-1}$ is quasi-conformal associated with $\mu_{\alpha} \frac{ \overline{dz_{\alpha}}}{dz_{\alpha}}$. Readers are referred to \cite{Gardiner00} for more details of quasi-conformal maps.

\subsection{Stereographic projection}
In our work, we frequently make use of the stereographic projection. Mathematically, the {\em stereographic projection} is a conformal map $P_N: \mathbb{S}^2 \to \overline{\mathbb{C}}$ with
\begin{equation}
 P_N(x,y,z) = \frac{x}{1-z} + i \frac{y}{1-z}.
\end{equation}
The {\em inverse stereographic projection} is a conformal map $P_N^{-1}: \overline{\mathbb{C}} \to \mathbb{S}^2$ with
\begin{equation}
 P_N^{-1}(x+iy) = \left(\frac{2x}{1+x^2+y^2}, \frac{2y}{1+x^2+y^2}, \frac{-1+x^2+y^2}{1+x^2+y^2}\right).
\end{equation}
Similarly, we define the {\em south-pole stereographic projection} $P_S: \mathbb{S}^2 \to \overline{\mathbb{C}}$ by
\begin{equation}
 P_S(x,y,z) = \frac{x}{1+z} + i \frac{y}{1+z}.
\end{equation}
The {\em inverse south-pole stereographic projection} is the map $P_S^{-1}: \overline{\mathbb{C}} \to \mathbb{S}^2$ with
\begin{equation}
 P_S^{-1}(x+iy) = \left(\frac{2x}{1+x^2+y^2}, \frac{2y}{1+x^2+y^2}, \frac{1-x^2-y^2}{1+x^2+y^2}\right).
\end{equation}

\subsection{Point cloud and local system}

A {\em point cloud} $P=\{z_1,z_2, \dots, z_n\} \subset \mathbb{R}^3$ is a set of sample points representing a Riemann surface $\mathcal{M}$. Because of the absence of the connectivity information, we construct a local coordinate system for $P$ on each point $z_s$ and approximate the derivatives. To achieve this, We define an atlas $(U_s,\phi_s)$ for each point $z_s$, where $U_s$ is an open cover and $\phi_s$ is the associated local coordinate function. $U_s$ is formed using the collection of all neighboring points of $z_s$, denoted by $\mathcal{N}(z_s)$. Specifically, we apply the {\em $k$-Nearest-Neighbors ($k$-NN)} algorithm to define the neighborhood. The {\em $k$-nearest neighborhood} $\mathcal{N}^{k}(z_s)$ of $z_s$ is a set with the $k$ distinct elements in $P$ (including $z_s$) closest from $z_s$ under the Euclidean 2-norm. In this work, we apply the KD-tree implementation by Lin \cite{Lin} for the computation. We denote $\mathcal{N}^k(z_s) = \{z_s^1,z_s^2, \dots, z_s^k\}$ with $z_s^1=z_s$. Then, one common approach for constructing a local coordinate system is to define the normal vector as the $z$-axis, which is more convenient for further computation. There are various methods to obtain the tangent planes and the normal vectors for point clouds, such as the principal component analysis (PCA) method \cite{Hoppe92}. Using the PCA method for $z_s$, we obtain three vectors $\{e_s^1,e_s^2,e_s^3\}$ which form an orthonormal basis of $\mathbb{R}^3$.

Then, we project $\mathcal{N}^k(z_s)$ to the plane spanned by $\{e_s^1,e_s^2\}$ by $\hat{z}_s^i = z_s^i - \langle z_s^i-z_s,e_s^3 \rangle e_s^3, \ i = 1,2, \dots, k$. Now we have the projection $\hat{\mathcal{N}^k}(z_s)=\{\hat{z}_s^1,\hat{z}_s^2, \dots, \hat{z}_s^k \}$ and the local coordinates $\{(x_s^1,y_s^1),(x_s^2,y_s^2), \dots, (x_s^k,y_s^k)\}$, where $x_s^i = \langle z_s^i-z_s,e_s^1\rangle$ and $y_s^i = \langle z_s^i-z_s,e_s^2\rangle$ for $i = 1,2, \dots, k$. Therefore, we can define $\phi_s : \mathcal{N}_s \to \mathbb{R}^2$ by $\phi_s(z_s^i) = (x_s^i, y_s^i)$. Also, the neighborhood $\mathcal{N}(z_s)$ can be regarded as a graph of its projection $\hat{\mathcal{N}}(z_s)$, that is, $z_s^i = z_s + x_s^i e_s^1 + y_s^i e_s^2 + f_s(x_s^i,y_s^i) e_s^3$.

\section{An overview of the fast spherical conformal parameterization algorithm for triangular meshes} \label{overview}
In this section, we briefly describe the approach in \cite{Choi15a} for computing a spherical conformal parameterization of a genus-0 closed triangular mesh $M$. This approach motivates our proposed parameterization scheme for genus-0 point clouds.

To compute a conformal mapping $f: M \to \mathbb{S}^2$, it suffices to solve Equation (\ref{eqt:harmonic}). This can be achieved by solving the Laplace equation $\Delta^T f = 0$ subject to $\|f\| = 1$, where $\Delta^T f$ is the tangential component of $\Delta f$ on the tangent plane of $\mathbb{S}^2$. This tangential approach was applied by Oberknapp and Polthier in \cite{Oberknapp97}. Note that this problem is nonlinear because of the constraint $\|f\| = 1$. In \cite{Angenent99,Haker00}, Angenent {\it et al.} linearize this problem by solving the equation on the complex plane:
\begin{equation} \label{eqt:laplace}
 \Delta f = 0
\end{equation}
given three boundary constraints $f(a_i) = b_i$, where $a_i, b_i \in \mathbb{C}, i = 1,2,3$ such that the triangle $[a_1,a_2,a_3]$ and the triangle $[b_1,b_2,b_3]$ are with the same angle structures. Note that $\Delta^T f = \Delta f = 0$ since the target domain is now $\mathbb{C}$. As the nonlinear constraint $\|f\|=1$ is removed, the above problem becomes linear and can be solved using the cotangent formula \cite{Pinkall93}.

After solving Equation (\ref{eqt:laplace}), the inverse stereographic projection $P_N^{-1}$ is applied for obtaining a spherical parameterization. However, unlike in the continuous case, the spherical parameterization in the discrete case is with large conformality distortion at the north pole of the sphere due to the discretization and the approximation errors. Hence, Choi {\it et al.} \cite{Choi15a} proposed to apply the south-pole stereographic projection $P_S$ to map the sphere to a planar domain $R \subset \mathbb{C}$. Note that the region with large distortion is the innermost region of $R$ while the outermost region of $R$ is with negligible distortion. Denote the above steps by a map $g: M \to R$. To correct the distortion of $g$, Choi {\it et al.} made use of the quasi-conformal theory.

Let $\mu_{g^{-1}}$ be the Beltrami coefficient of the map $g^{-1}$. Fixing the outermost region on $R$, Choi {\it et al.} \cite{Choi15a} composed the map $g$ with a quasi-conformal map $h: R \to \mathbb{S}^2$ with the associated Beltrami coefficient $\mu_h = \mu_{g^{-1}}$. Let $h = u + iv$ and $\mu_h = \rho + i \tau$. Specifically, by considering the Beltrami Equation (\ref{eqt:beltrami}), each pair of the partial derivatives $v_x, v_y$ and $u_x, u_y$ can be expressed as linear combinations of the other \cite{Lam13},
\begin{equation} \label{eqt:linear}
 \begin{split}
  -v_y &= \alpha_1 u_x + \alpha_2 u_y; \\
  v_x &= \alpha_2 u_x + \alpha_3 u_y,
 \end{split}
  \ \ \text{ and } \ \
 \begin{split}
  -u_y &= \alpha_1 v_x + \alpha_2 v_y; \\
  u_x &= \alpha_2 v_x + \alpha_3 v_y,
 \end{split}
\end{equation}
where $\alpha_1 = \frac{(\rho-1)^2+\tau^2}{1-\rho^2-\tau^2}; \alpha_2 = -\frac{2\tau}{1-\rho^2-\tau^2}; \alpha_3 = \frac{(1+\rho)^2 + \tau^2}{1-\rho^2-\tau^2}$. Since $\nabla \cdot \left( \begin{array}{c}
                            -v_y \\ v_x
                           \end{array} \right) = 0$ and $\nabla \cdot \left( \begin{array}{c}
                            -u_y \\ u_x
                           \end{array} \right) = 0$ , the map $h$ can be constructed by solving the following equations
\begin{equation}\label{eqt:BeltramiPDE}
\nabla \cdot \left(A \left(\begin{array}{c}
u_x\\
u_y \end{array}\right) \right) = 0\ \ \mathrm{and}\ \ \nabla \cdot \left(A \left(\begin{array}{c}
v_x\\
v_y \end{array}\right) \right) = 0
\end{equation}
where
$\displaystyle A = \left( \begin{array}{cc} \alpha_1 & \alpha_2\\
\alpha_2 & \alpha_3 \end{array}\right)$. In the discrete case, the above elliptic PDEs (\ref{eqt:BeltramiPDE}) can be discretized into sparse symmetric positive definite linear systems as described in \cite{Lui13,Choi15a}. In \cite{Jones13}, Jones and Mahadevan derived the system (\ref{eqt:linear}) from the conjugate Beltrami equation $\frac{\partial h}{\partial \overline{z}} = \nu(z) \overline{\left(\frac{\partial h}{\partial z}\right)}$ and proposed an alternative approach for solving the system. Specifically, the authors considered minimizing the following functional
\begin{equation}
 \frac{1}{2} \int_R \left(\alpha_1 |\nabla u|^2 + 2 \alpha_2 \nabla u \cdot \nabla v + \alpha_3 |\nabla v|^2 \right) dS,
\end{equation}
using a Euler-Lagrange variational approach. Despite the different implementations, both of the two abovementioned methods effectively solve for a quasi-conformal map on triangulated meshes. Then, by Equation (\ref{eqt:composition_formula}), the composition map $h \circ g : M \to \mathbb{S}^2$ is with the Beltrami coefficient $\mu_{h \circ g} = 0$ and hence $h \circ g$ is conformal. Readers are referred to \cite{Choi15a} for more details.

Note that a key step above is the computation of the quasi-conformal map $h$ for improving the conformality, which is guaranteed by the composition formula (\ref{eqt:composition_formula}). However, the Beltrami coefficients in the above algorithm are approximated on the triangular faces of a mesh. Hence, the above algorithm cannot be directly applied for point clouds. Moreover, even if we can define the discrete Beltrami coefficients on point clouds, Equation (\ref{eqt:composition_formula}) may not hold anymore. Therefore, we need to replace this key step by a new method suitable for point clouds.

\section{Meshing genus-0 point clouds using spherical conformal parameterization} \label{mainsection}
In this section, we discuss our proposed framework for meshing genus-0 point clouds. The main steps involved include solving a series of Laplace equations on the complex plane for the spherical conformal parameterization of a genus-0 point cloud, and creating a mesh structure with the aid of the global parameterization.

\subsection{Approximation of the Laplace-Beltrami operator}
\begin{table}[t]
\centering
\begin{tabular}{|l|l|}
\hline
Weight & Formula of $w(d)$\\
\hline
Constant weight & $\displaystyle w(d)=1$\\
\hline
Exponential weight & $\displaystyle w(d)=\exp \left(-\frac{d^2}{h^2} \right)$\\
\hline
Inverse of squared distance weight & $\displaystyle w(d)=\frac{1}{d^2 + \epsilon^2}$\\
\hline
Wendland weight \cite{Wendland95,Wendland01} & $\displaystyle w(d)=\left(1-\frac{d}{D}\right)^4\left(\frac{4d}{D}+1\right)$\\
\hline
Special weight \cite{Liang12} & $\displaystyle w(d)=\left\{\begin{array}{ll} 1 & \text{if }d=0\\ \frac{1}{k} & \text{if }d\neq 0\\ \end{array} \right.$ \\
\hline
\end{tabular}
\bigbreak
\caption{Some common weighting functions for the MLS approximation.}
\label{table:weight}
\end{table}

In this subsection, we explain our approximation scheme for the Laplace-Beltrami operator in the Laplace Equation (\ref{eqt:laplace}) on a point cloud $P$ by the moving least-square method. The moving least-square method is widely used for approximation \cite{Liu97,Levin98,Nealen04,Breitkopf05,Lange05,Liang12,Liang13}. In particular, Liang {\it et al.} \cite{Liang12,Liang13} approximated the LB operator on point clouds using the MLS method with a special weight function. Our approximation scheme is built upon the method in \cite{Lange05,Liang12,Liang13}. In this work, we propose a new weight function to achieve a more accurate approximation of the LB operator.

First, we discuss our approximation method for the derivatives on the point cloud $P = \{z_1, z_2, \dots, z_n\}$. To simplify the discussion, we only consider the approximation on the patch $\mathcal{N}(z_s)$ of a point $z_s \in P$. Recall that $\mathcal{N}(z_s)$ can be regarded as a graph of its projection $\hat{\mathcal{N}}(z_s)$, that is, $z_s^i = z_s + x_s^i e_s^1 + y_s^i e_s^2 + f_s(x_s^i,y_s^i) e_s^3$. Denote the derivatives of $f_s$ along the $e_s^1$-direction and the $e_s^2$-direction by $f_{sx}$ and $f_{sy}$ respectively. We select a set of basic functions $\{f_s^1,f_s^2, \dots, f_s^m\}$ as a basis and write $f_s(x,y) \approx\sum_{i=1}^m c_i f_s^i(x,y)$, where $\{c_i\}_{i=1}^m$ are some coefficients to be determined. In our work, we use $\{1, x, y, x^2, xy, y^2\}$ as the basis of the space of all polynomials with second order or below, which means $m = 6$. We add a remark here that $m=6$ is an appropriate choice for our approximation. Since second derivatives are considered in approximating the LB operator, polynomials with at least second order are needed. On the other hand, if we fit a polynomial with third order ($m=10$) or higher, it will be too sensitive to noises and the approximation gets worse. Therefore, $m=6$ is a suitable dimension for our approximation.

In the approximation, we aim to minimize
\begin{equation} \label{eqt:mls}
\sum^{n}_{i=1} w_{i} \left( \sum^{m}_{j=1} c_{j} f_s^{j} (x_{i},y_{i}) - f_s(x_{i},y_{i}) \right)^{2}
\end{equation}
where $w_i = w(\|z_i - z_s\|)$ for some weighting function $w:\mathbb{R} \to \mathbb{R}$. The weight function $w$ significantly affects the accuracy and robustness of the approximation. Hence, one must carefully choose a suitable weight function. Table \ref{table:weight} lists some common weighting functions.

Note that the information provided by the data points near the center point $z_s$ should be more reliable than that of the data points distant from $z_s$. The closer the data points are to $z_i$, the more reliable they are. Hence, it is natural to consider a smooth weight function which concentrates at $z_s$. This motivates us to use of a weight function of the Gaussian type. As a remark, in \cite{Belkin08,Belkin09}, Belkin {\it et al.} used a Gaussian weight function in the form of $\exp\left(-\|x-y\|^2/4t\right)$ for integral approximation. In our MLS approximation, we propose another Gaussian-type weight function:
\begin{equation}
\left\{\begin{array}{ll}
\displaystyle w_s = w(0) = 1 &  \\
\displaystyle w_{i} = w(\|z_i - z_s\|) = \frac{1}{k} \exp \left(-\frac{\sqrt{k}}{h^2} \|z_i - z_s\|^2 \right) & \text{ for all } i \neq s,\\
\end{array} \right.
\end{equation}
where $h$ is the maximum distance from $z_s$ in $\mathcal{N}^k(z_s)$. Numerical experiments are demonstrated in Section \ref{experiment} to support our proposed weight function with the specific factor $\sqrt{k}/h^2$ inside the exponent. It can be observed that our proposed weight results in more accurate approximations of the LB operator on point clouds.

With the proposed weight function, we now solve the minimization problem (\ref{eqt:mls}). Denote $f^{j}_{s,i} = f_s^{j}(x_{i},y_{i})$ and $f_{s,i} = f_s(x_{i},y_{i})$.

Let $\vec{A}=
\left(\begin{matrix}
f^{1}_{s,1} & f^{2}_{s,1} & \cdots & f^{m}_{s,1} \\
f^{1}_{s,2} & f^{2}_{s,2} & \cdots & f^{m}_{s,2} \\
\vdots    & \vdots    & \ddots & \vdots    \\
f^{1}_{s,n} & f^{2}_{s,n} & \cdots    & f^{m}_{s,n}
\end{matrix} \right)$,
$\vec{D}=
\left(\begin{matrix}
w_1 & 0 & \cdots & 0 \\
0 & w_2 & \cdots & 0 \\
\vdots    & \vdots    & \ddots & \vdots    \\
0 & 0 & \cdots    & w_n
\end{matrix} \right)$,
$\vec{c}=
\left(\begin{matrix}
c_1\\
c_2\\
\vdots\\
c_m
\end{matrix} \right)$,
and $\vec{b}=
\left(\begin{matrix}
f_{s,1}\\
f_{s,2}\\
\vdots\\
f_{s,n}
\end{matrix} \right)$. The minimization problem in (\ref{eqt:mls}) can be written as follows:
\begin{equation}
 \min_{c\in \mathbb{R}^n} \left\langle \vec{D} (\vec{Ac} - \vec{b}), \vec{Ac} - \vec{b} \right\rangle.
\end{equation}
We can solve it using the least-square method, namely solving
\begin{equation}
\vec{A}^{T}\vec{D}\vec{A}\vec{c} = \vec{A}^{T}\vec{D}\vec{b}.
\end{equation}

Next, for any function $u$ defined on the neighborhood $\mathcal{N}(z)$, we can approximate it by a combination of $\{f_s^1,f_s^2, \dots, f_s^m\}$:
\begin{equation}
u=f_s(x,y)\approx\sum_{i=1}^m \hat{c}_i f_s^i(x,y).
\end{equation}

Similarly, the coefficients $\hat{c}_i$ can be approximated. Let $\vec{A}=
\left(\begin{matrix}
f^{1}_{s,1} & f^{2}_{s,1} & \cdots & f^{m}_{s,1} \\
f^{1}_{s,2} & f^{2}_{s,2} & \cdots & f^{m}_{s,2} \\
\vdots    & \vdots    & \ddots & \vdots    \\
f^{1}_{s,n} & f^{2}_{s,n} & \cdots & f^{m}_{s,n}
\end{matrix} \right)$,
$\vec{D}=
\left(\begin{matrix}
w_1 & 0 & \cdots & 0 \\
0 & w_2 & \cdots & 0 \\
\vdots & \vdots & \ddots & \vdots \\
0 & 0 & \cdots    & w_n
\end{matrix} \right)$,
$\vec{\hat{c}}=
\left(\begin{matrix}
\hat{c}_1\\
\hat{c}_2\\
\vdots\\
\hat{c}_m
\end{matrix} \right)$, and
$\vec{u}=
\left(\begin{matrix}
u_1\\
u_2\\
\vdots\\
u_n
\end{matrix} \right)$.
We can find the coefficients $\hat{c}_i$ by solving the following least-square problem
\begin{equation}
\vec{A}^{T}\vec{D}\vec{A}\vec{\hat{c}} = \vec{A}^{T}\vec{D}\vec{u}.
\end{equation}

Since we know the explicit formula of the derivatives of each $f_s^i$, we can compute the approximated derivatives of $u$, such as
\begin{equation}\label{ddx}
\frac{\partial u}{\partial x}=\overset{m}{\underset{i=1}{\sum}}\vec{c}_i \frac{\partial f_s^i}{\partial x}  \ \ \text{ and } \ \ \frac{\partial u}{\partial y}=\overset{m}{\underset{i=1}{\sum}}\vec{c}_i \frac{\partial f_s^i}{\partial y}.
\end{equation}

Now, we are ready to introduce the construction of the LB operator of a smooth function $u$ on $\mathcal{N}(z_s)$. For any smooth real-valued function $u$ on the $\mathcal{N}(z)$, the LB operator of $u$ is given by
\begin{equation}\label{laplacian}
\Delta u(z) = \frac{1}{W} \sum^{2}_{i,j=1} \partial_{i} (g^{ij} W \partial_{j} (u(z)) ),
\end{equation}
where $z$ is a point in $\mathcal{N}(z)$, $(g_{ij})$ is the metric of the surface at $z$, $W = \sqrt{det(g_{ij})}$, and $(g^{ij}) = (g_{ij})^{-1}$.

Since $z_s^i=(x_s^i,y_s^i,f_s(x_s^i,y_s^i))$ and $\mathcal{N}(z_s)$ is a graph of $\hat{\mathcal{N}}(z_s)$, we have
\begin{equation}
(g_{ij}) = \left(\begin{matrix}
1+(f_s)^{2}_{x} & (f_s)_{x} (f_s)_{y} \\
(f_s)_{x} (f_s)_{y} & 1+(f_s)^{2}_{y}
\end{matrix} \right) \ \ \text{ and } \ \
(g^{ij}) = \frac{1}{W^{2}} \left(\begin{matrix}
1+(f_s)^{2}_{y} & -(f_s)_{x} (f_s)_{y} \\
-(f_s)_{x} (f_s)_{y} & 1+(f_s)^{2}_{x}
\end{matrix} \right),
\end{equation}
where $W = \sqrt{1+(f_s)^{2}_{x}+(f_s)^{2}_{y}}$.

We use Equation (\ref{ddx}) to calculate the first order partial derivatives of $f_s$. Then, we proceed to compute $\Delta u(z_s)$. Since we have a closed form of $\Delta u$ and the LB operator is a second order differential operator, by differentiating Equation (\ref{laplacian}), we get
\begin{equation}\label{LBoperator}
\Delta u(z_s) = \alpha_{1}\frac{\partial u}{\partial x}(z_s) + \alpha_{2}\frac{\partial u}{\partial y}(z_s) + \alpha_{3}\frac{\partial^{2} u}{\partial x^{2}}(z_s) + \alpha_{4}\frac{\partial^{2} u}{\partial x\partial y}(z_s) + \alpha_{5}\frac{\partial^{2} u}{\partial y^{2}}(z_s)
\end{equation}
where $\alpha_{1}, \alpha_{2}, \alpha_{3}, \alpha_{4}, \alpha_{5}$ are coefficients which depend on partial derivatives of $f_s$. This completes our approximation scheme for the LB operator on point clouds. With this approximation, we are now ready to describe our proposed spherical conformal parameterization algorithm for genus-0 point clouds.

\subsection{Spherical conformal parameterization of genus-0 point clouds}
In this subsection, we introduce our proposed method for the spherical conformal parameterizations of genus-0 point clouds.

Given a point cloud $P$ sampled from a genus-0 closed surface $\mathcal{M}$, our goal is to find a conformal map $\tilde{f}: P \to \mathbb{S}^2$ which effectively resembles the conformal map $f: \mathcal{M} \to \mathbb{S}^2$. By the previous section, we can approximate the LB operator $\Delta$ on $P$. Denote the approximated LB operator on the point cloud by $\Delta_{PC}$. The approximation allows us to solve the Laplace equation (\ref{eqt:laplace}) on point clouds for a map $\phi: P \to \mathbb{C}$. More specifically, we solve the following equation
\begin{equation} \label{eqt:laplacepc}
\Delta_{PC} \phi = 0
\end{equation}
subject to the constraints $\phi(a_i) = b_i$ for $i = 1,2,3$, where $a_i, b_i \in \mathbb{C}$. The choice of the three boundary points $a_1, a_2, a_3$ affects the conformality of the map $\phi$. In the case of triangular meshes, $a_1,a_2,a_3$ are chosen to be the three vertices of the most regular triangle among all triangles on the input mesh \cite{Choi15a}. Here, the regularity of a triangle $[a_1,a_2,a_3]$ is defined by
\begin{equation}
\text{Regularity}[a_1,a_2,a_3] = \left|\alpha - \frac{\pi}{3}\right| + \left|\beta - \frac{\pi}{3}\right| + \left|\gamma - \frac{\pi}{3}\right|,
\end{equation}
where $\alpha, \beta$ and $\gamma$ are the three angles in the triangle $[a_1,a_2,a_3]$. However, in the case of point clouds, we do not have the required connectivity information. Hence, we choose the three points $a_1,a_2,a_3$ in a different way.

Recall that in approximating the LB operator, it is necessary to find the $k$ nearest neighboring data points $z_s^1, z_s^2, \dots, z_s^k$ for each point $z_s$ on the point cloud $P$. We consider forming a triple using $z_s$ and two other neighboring points $z_s^i$ and $z_s^j$, where $i \neq j$. Different combinations of $i$ and $j$ result in different triples $[z_s,z_s^i,z_s^j]$. Then, we propose to choose the three boundary points $a_1, a_2, a_3$ in the constraint of Equation (\ref{eqt:laplacepc}) by considering
\begin{equation}\label{eqt:regularity}
\min_{s,i,j} \text{ Regularity}[z_s,z_s^i,z_s^j]
\end{equation}
among all combinations of $s$, $i$ and $j$.

After solving Equation (\ref{eqt:laplacepc}) with our proposed boundary constraints, we apply the inverse stereographic projection $P_N^{-1}$ on $\phi(P)$ to obtain a spherical point cloud. Recall that the conformality distortion around the north pole is large due to the approximation error in the stereographic projection. Note that the key step in the method in \cite{Choi15a} for correcting the distortion via a composition of quasi-conformal maps does not work for the case of point clouds. Now, we propose a new method to correct the conformality distortion by solely using the LB operator.

We begin with the south-pole stereographic projection $P_S$ to project the spherical point cloud back onto the complex plane. Under the projection, the North pole of the sphere, which corresponds to the outermost region of $\phi(P) \subset \mathbb{C}$, is mapped to the innermost region on the complex plane. It follows that the outermost region is now with very low distortion while the innermost region is with large distortion. We use the outermost low-distortion data points as the boundary constraints and solve the Laplace equation $\psi: (P_S \circ P_N^{-1} \circ \phi)(P) \to \mathbb{C}$ again:
\begin{equation} \label{eqt:laplacepc2}
\Delta_{PC} \psi = 0
\end{equation}
subject to the boundary constraints $\psi(x) = x$ for all data points $x$ in the outermost low-distortion region. The low-distortion boundary constraints provide us with a more accurate result in the inner part of the planar region. Finally, we apply the inverse south-pole stereographic projection $P_S^{-1}$ and obtain a composition map
\begin{equation}
\widetilde{f} = P_S^{-1} \circ \psi \circ P_S \circ P_N^{-1} \circ \phi.
\end{equation}
This step effectively replaces the step in the mesh parameterization algorithm in \cite{Choi15a} which involves computing a quasi-conformal map.

Altogether, by solving Equation (\ref{eqt:laplacepc}) and Equation (\ref{eqt:laplacepc2}) and using a number of projections, we can obtain a conformal map $\widetilde{f}: P \to \mathbb{S}^2$. Note that the method in \cite{Choi15a} is based on certain manipulations of Beltrami coefficients and quasi-conformal maps. In contrast, our new method only involves solving Laplace equations. The equivalence between the two approaches can be explained as follows.

In the first step, the conformality distortion of the spherical parameterization is due to the error in the stereographic projection. Then in the approach in \cite{Choi15a}, the entire initial parameterization result is used in Equation (\ref{eqt:BeltramiPDE}) for computing a quasi-conformal map in order to cancel the distortion. The method is theoretically guaranteed by the composition formula (\ref{eqt:composition_formula}) of quasi-conformal maps. In contrast, in our new approach, we only make use of the most accurate part in the initial parameterization result. More explicitly, we use the southern-most regions as the boundary constraints and compute the remaining part of the spherical parameterization again, with the aid of the LB operator. The replacement of the south-pole step in \cite{Choi15a} by our new south-pole step can be justified by the following theorem.

\begin{theorem}
Let $(S_1,\sigma|dz|^2)$ and $(S_2,\rho|dw|^2)$ be two Riemann surfaces, and $\mu$ is a prescribed Beltrami differential on $S_1$. Then, the map solved by Equation (\ref{eqt:BeltramiPDE}) is a harmonic map between $(S_1,|dz + \mu d\bar{z}|^2)$ and $(S_2,\rho|dw|^2)$. Consequently, solving the Laplace equation (\ref{eqt:laplacepc2}) is equivalent to solving Equation (\ref{eqt:BeltramiPDE}).
\end{theorem}

\begin{proof}
Let $\zeta$ be the coordinates of $S_1$ with respect to the distorted metric $|dz + \mu d\bar{z}|^2$. Then, the harmonic map between $(S_1,|dz + \mu d\bar{z}|^2)$ and $(S_2,\rho|dw|^2)$ is a critical point of the following energy
\begin{equation}
E_{harm}(h) = \int_{S_1} \rho(h(\zeta))(|h_{\zeta}|^2 + |h_{\bar{\zeta}}|^2) dx dy.
\end{equation}
On the other hand, by definition of the Beltrami equation (\ref{eqt:beltrami}), the solution to Equation (\ref{eqt:BeltramiPDE}) is the critical point of the following energy functional
\begin{equation}
E_{QC}(f) = \int_{S_1} \rho(f(z))(|f_{\bar{z}} - \mu f_z|^2) dx dy.
\end{equation}
It is shown in \cite{Lui15} that the above two energy functionals have the same set of critical points. Hence, solving the generalized Laplace equation (\ref{eqt:BeltramiPDE}) for a quasi-conformal map is equivalent to solving the Laplace equation (\ref{eqt:laplacepc2}) for a harmonic map under the distorted metric. Then, the conformality of our approach is again guaranteed by the composition formula (\ref{eqt:composition_formula}) of quasi-conformal maps.
\end{proof}

Therefore, in the continuous case, under suitable boundary conditions in Equation (\ref{eqt:BeltramiPDE}) and Equation (\ref{eqt:laplacepc2}), both of the two methods are theoretically guaranteed for producing a conformal map.

However, in the discrete case, the two methods perform differently. For the case of triangular meshes, the Beltrami coefficients can be accurately computed and the composition formula (\ref{eqt:composition_formula}) of quasi-conformal maps is accurate under the discretization. In this situation, the method in \cite{Choi15a} can be effectively applied. Yet, for the case of point clouds, we only have an approximation of the LB operator but not the Beltrami coefficients, and there is no guarantee about the composition formula (\ref{eqt:composition_formula}) of quasi-conformal maps. Consequently, it is more suitable to use our proposed method as it only involves solving the Laplace Equation (\ref{eqt:laplacepc2}). However, since the accuracy of our proposed method depends on the boundary constraints in solving the Laplace Equation (\ref{eqt:laplacepc2}), the boundary constraints obtained from the initial parameterization result may contain small error and hence slightly affect the result in solving Equation (\ref{eqt:laplacepc2}). Therefore, it is desirable to perform some more iterations for obtaining a more accurate result.

It is noteworthy that in the parameterization algorithm in \cite{Choi15a} for triangular meshes, no further steps are required after the second step. However, because of the abovementioned issue about the boundary constraints in the Laplace Equation (\ref{eqt:laplacepc2}), further iterations are necessary for enhancing the parameterization result. We call them the {\it North-South (N-S) reiterations}. In each N-S reiteration, two Laplace equations are solved again after the north-pole stereographic projection and the south-pole stereographic projection respectively. For solving each Laplace equation, we fix the outermost $r\%$ points on $\mathbb{C}$ to ensure the existence of the solution.

More specifically, in each N-S reiteration, we first project the previous spherical parameterization result onto the complex plane using the north-pole stereographic projection. Next, we compute a harmonic map $\widetilde{\phi}: (P_N \circ \widetilde{f})(P) \to \mathbb{C}$ by solving the Laplace equation
\begin{equation}\label{eqt:laplacepc3}
\Delta_{PC} \widetilde{\phi} = 0
\end{equation}
with the boundary constraints $\widetilde{\phi}(x) = x$ for the outermost $r\%$ of the data points on $\mathbb{C}$. After obtaining $\widetilde{\phi}$, the inverse north-pole stereographic projection is again applied, followed by the south-pole stereographic projection. Then, we compute another harmonic map $\widetilde{\psi}: (P_S \circ P_N^{-1} \circ \widetilde{\phi} \circ P_N \circ \widetilde{f})(P) \to \mathbb{C}$ by solving the Laplace equation
\begin{equation} \label{eqt:laplacepc4}
\Delta_{PC} \widetilde{\psi} = 0
\end{equation}
with the boundary constraints $\widetilde{\psi}(x) = x$ for the outermost $r\%$ of the data points on $\mathbb{C}$. We then define the updated spherical parameterization by the composition map
\begin{equation}
 P_S^{-1} \circ \widetilde{\psi} \circ P_S \circ P_N^{-1} \circ \widetilde{\phi} \circ P_N \circ \widetilde{f}.
\end{equation}
We check whether the above updated parameterization result is close to the previous parameterization result $\widetilde{f}$. If yes, then the parameterization is stable and we complete the algorithm. If no, we apply another N-S reiteration on the updated parameterization point by repeating the procedures and so on. In practice, we choose $r = 10$. Our proposed spherical conformal parameterization scheme is outlined in Algorithm \ref{alg:pc_spherical}.

To explain the motivation of our proposed N-S reiteration scheme, we define the \emph{N-S Dirichlet energy} by
\begin{equation}
\tilde{E}(f) = \frac{1}{2} \left(E(P_N(f)) + E(P_S(f))\right),
\end{equation}
where $E(f)$ is the Dirichlet energy. It follows that the minimum of $\tilde{E}$ is attained if and only if $E(P_N(f))$ and $E(P_S(f))$ are minimized, which implies that $E(f)$ is minimized and $f$ is conformal. Therefore, to find a conformal $f$, we can consider minimizing the N-S Dirichlet energy $\tilde{E}(f)$. More specifically, we aim to minimize both $E(P_N(f))$ and $E(P_S(f))$. Note that these two energies are respectively minimized by solving the Laplace equations (\ref{eqt:laplacepc3}) and (\ref{eqt:laplacepc4}) in our proposed N-S reiteration. Introducing the N-S reiteration for minimizing the energies is advantageous for two reasons. Firstly, it linearizes the computation as we only need to solve linear systems on $\mathbb{C}$. Secondly, it avoids the error induced by the stereographic projection as we consider both the north-pole step and the south-pole step in each reiteration.

Note that theoretically we only need to take away the two points of infinity, and at least 3 points on $\mathbb{C}$ are required to be fixed to guarantee the existence of the solution of the two Laplace equations. However, in terms of the numerical computations, the large matrix equations may be ill-posed if only 3 points are fixed as boundary constraints. Therefore, we consider fixing the outermost $r\%$ of data points in solving Equations (\ref{eqt:laplacepc3}) and (\ref{eqt:laplacepc4}). Nevertheless, fixing these extra points may not affect the accuracy of the solution too much. This can be explained with the aid of the Beltrami holomorphic flow \cite{Gardiner00} as follows.

\begin{theorem}[Beltrami holomorphic flow on $\overline{\mathbb{C}}$ \cite{Gardiner00}]
There is a one-to-one correspondence between the set of quasi-conformal diffeomorphisms of $\overline{\mathbb{C}}$ that fix the points $0, 1 $ and $\infty$ and the set of smooth complex-valued functions $\mu$ on $\overline{\mathbb{C}}$ with $\|\mu\|_{\infty} = k <1$. Furthermore, the solution $f^{\mu}$ to the Beltrami equation (\ref{eqt:beltrami}) depends holomorphically on $\mu$. Let $\{\mu(t)\}$ be a family of Beltrami coefficients depending on a real or complex parameter $t$. Suppose also that $\mu(t)$ can be written in the form
\begin{equation}
\mu(t)(z) = \mu(z) + t \nu(z) + t \epsilon(t)(z)
\end{equation}
for $z \in \mathbb{C}$, with suitable $\mu$ in the unit ball of $C^{\infty}(\mathbb{C}), \nu, \epsilon(t) \in L^{\infty}(\mathbb{C})$ such that $\|\epsilon(t) \|_{\infty} \to 0$ as $t \to 0$. Then, for all $w \in \mathbb{C}$,
\begin{equation}
f^{\mu(t)}(w) = f^{\mu}(w) + t V(f^{\mu}, \nu)(w) + o(|t|)
\end{equation}
locally uniformly on $\mathbb{C}$ as $t\to 0$, where
\begin{equation} \label{eqt:bhf}
V(f^{\mu}, \nu)(w) = - \frac{f^{\mu}(w)(f^{\mu}(w)-1)}{\pi} \int_{\mathbb{C}} \frac{\nu(z)((f^{\mu})_z (z))^2}{f^{\mu}(z)(f^{\mu}(z)-1)(f^{\mu}(z)-f^{\mu}(w))} dx dy.
\end{equation}
\end{theorem}

In our case, as the conformality distortion of the outermost region is negligible, $\nu$ is compactly supported around origin. Hence, it can be deduced from Equation (\ref{eqt:bhf}) in the above theorem that the data points located farther away from the origin will be associated with a smaller flow $V$, as the denominator in the integral becomes larger. Therefore, in each iteration, the outermost points will remain almost unchanged, while the innermost points (which have the largest conformality distortion) will be adjusted and improved. In other words, fixing more outermost points for the numerical stability does not affect the solution much. Numerical experiments are presented in Section \ref{experiment} to verify the convergence of our N-S reiteration scheme. Intuitively, the boundary constraints in Equations (\ref{eqt:laplacepc3}) and (\ref{eqt:laplacepc4}) are adjusted to the positions associated with a conformal map by the iterations. They are observed to eventually stabilize and hence we obtain the desired conformal map by solving Equations (\ref{eqt:laplacepc3}) and (\ref{eqt:laplacepc4}) with these boundary constraints.

Finally, we make an important remark about our proposed spherical conformal parameterization algorithm for genus-0 point clouds. In addition to genus-0 point clouds, our proposed algorithm also efficiently works on genus-0 triangular meshes. Note that for triangular meshes, the LB operator can be easily constructed by computing the cotangent weights on the given mesh structures. Also, solving Laplace equations on the complex plane requires only linear time. Hence, our proposed algorithm can serve as an alternative approach for computing spherical conformal parameterizations of genus-0 closed triangular meshes in linear time.

\begin{algorithm}[h]
\KwIn{A genus-0 point cloud $P$.}
\KwOut{A spherical conformal parameterization $f: P \to \mathbb{S}^2$.}
\BlankLine
Approximate the LB operator on $P$ and denote it by $\Delta_{PC}$\;
Find the most regular triple of points by solving problem (\ref{eqt:regularity})\;
Obtain a map $\phi: P \to \overline{\mathbb{C}}$ by solving the Laplace equation (\ref{eqt:laplacepc})\;
Apply the inverse stereographic projection $P_N^{-1} : \overline{\mathbb{C}} \to \mathbb{S}^2$ on $\phi(P)$\;
Apply the south-pole stereographic projection $P_S :\mathbb{S}^2 \to \overline{\mathbb{C}}$ on $(P_N^{-1} \circ \phi)(P)$\;
Solve the Laplace equation (\ref{eqt:laplacepc2}) for $\psi: (P_S \circ P_N^{-1} \circ \phi)(P) \to \mathbb{C}$\;
Apply the inverse south-pole stereographic projection $P_S^{-1}$ and denote the overall composition of the maps by $f = P_S^{-1} \circ \psi \circ P_S \circ P_N^{-1} \circ \phi$\;
\Repeat{$\text{mean}(\|f(p_i) - \widetilde{f}(p_i)\|^2) < \epsilon$}{
Update $\widetilde{f}$ by $f$\;
Solve the Laplace equation (\ref{eqt:laplacepc3}) for $\widetilde{\phi}: (P_N \circ \widetilde{f})(P) \to \mathbb{C}$\;
Solve the Laplace equation (\ref{eqt:laplacepc4}) for $\widetilde{\psi}: (P_S \circ P_N^{-1} \circ \widetilde{\phi} \circ P_N \circ \widetilde{f})(P) \to \mathbb{C}$\;
Update $f$ by $P_S^{-1} \circ \widetilde{\psi} \circ P_S \circ P_N^{-1} \circ \widetilde{\phi} \circ P_N \circ \widetilde{f}$\;
}
\caption{Our proposed spherical conformal parameterization algorithm.}
\label{alg:pc_spherical}
\end{algorithm}

\subsection{Improving the distribution of the spherical parameterization}
It is obvious that spherical conformal parameterizations are unique only up to M\"obius transformations. Although the conformality does not change under the M\"obius transformations, the distribution of the points on the sphere does. The distribution is crucial for meshing. Hence, it is desirable to obtain an even distribution of the points on the sphere.

In the spherical conformal parameterization algorithm for triangular meshes \cite{Choi15a}, Choi {\it et al.} proved the following theorem:

\begin{theorem}[See \cite{Choi15a}, P.75] \label{invariance}
Let $T_1$ and $T_2$ be two triangles of $\mathbb{C}$. The product of the perimeters of $T_1$ and $P_S(P_N^{-1}(T_2))$ is invariant under arbitrary scaling of $T_1$ and $T_2$.
\end{theorem}
\\
With this theorem, Choi {\it et al.} achieved an even distribution of a spherical parameterization mesh by applying the stereographic projection on the sphere, and then considering the outermost triangle $T$ and the innermost triangle $t$ on the complex plane. They scaled the planar domain by a factor so that $T$ and $t$ are with the same perimeters on the sphere, under the inverse stereographic projection.

In our case, the above idea does not work as we do not have any information about the connectivity of the point clouds. However, we can extend Theorem \ref{invariance} for point clouds by considering two sets of points. The extension is as follows:

\begin{theorem}\label{invariance2}
Let $\{u_i\}_{i=0}^m$ and $\{v_j\}_{j=0}^n$ be two sets of points on $\mathbb{C}$. Then
\[
\begin{split}
&\left(\sum_{i=1}^m \|\lambda u_i-\lambda u_0\|\right) \left(\sum_{j=1}^n \left\|P_S(P_N^{-1}(\lambda v_j))-P_S(P_N^{-1}(\lambda v_0))\right\| \right) \\
= &\left(\sum_{i=1}^m \|u_i-u_0\|\right) \left(\sum_{j=1}^n \left\|P_S(P_N^{-1}(v_j))-P_S(P_N^{-1}(v_0))\right\| \right)\\
\end{split}
\]
for any scaling factor $\lambda \neq 0$. In other words, the product is an invariance under arbitrary scaling.\\
\end{theorem}
\begin{proof}
 We prove the theorem using the approach in \cite{Choi15a}. Note that for any $z = x+iy$, we have
\begin{equation}
 \begin{split}
  P_S(P_N^{-1}(z)) &= P_S(P_N^{-1}(x+iy))\\
 &= \frac{-\frac{2x}{1+x^2+y^2}}{1+\frac{-1+x^2+y^2}{1+x^2+y^2}} + i\frac{\frac{2y}{1+x^2+y^2}}{1+\frac{-1+x^2+y^2}{1+x^2+y^2}}\\
 &= \frac{-x}{x^2+y^2}+ i\frac{y}{x^2+y^2} = \frac{-Re(z)}{|z|^2}+ i\frac{Im(z)}{|z|^2}.\\
 \end{split}
\end{equation}

Hence, for any scaling factor $\lambda \neq 0$, we have
\begin{equation}
 \begin{split}
  &\left(\sum_{i=1}^m \|\lambda u_i-\lambda u_0\|\right) \left(\sum_{j=1}^n \left\|P_S(P_N^{-1}(\lambda v_j))-P_S(P_N^{-1}(\lambda v_0))\right\|\right) \\
=& \left(\sum_{i=1}^m \|\lambda u_i-\lambda u_0\|\right) \left(\sum_{j=1}^n \left\|\frac{-Re(\lambda v_j)}{|\lambda v_j|^2}+ i\frac{Im(\lambda v_j)}{|\lambda v_j|^2}-\frac{-Re(\lambda v_0)}{|\lambda v_0|^2}+ i\frac{Im(\lambda v_0)}{|\lambda v_0|^2}\right\| \right) \\
=& \left(\lambda \sum_{i=1}^m \| u_i- u_0\|\right) \left(\frac{\lambda}{\lambda^2} \sum_{j=1}^n \left\|\frac{-Re(v_j)}{|v_j|^2}+ i\frac{Im(v_j)}{|v_j|^2}-\frac{-Re(v_0)}{|v_0|^2}+ i\frac{Im(v_0)}{|v_0|^2}\right\| \right) \\
= &\left(\sum_{i=1}^m \|u_i-u_0\|\right) \left(\sum_{j=1}^n \left\|P_S(P_N^{-1}(v_j))-P_S(P_N^{-1}(v_0))\right\| \right).\\
\end{split}
\end{equation}

Therefore, the product is an invariance.
\end{proof}
\\
To apply this theorem for obtaining an even distribution of our spherical parameterization result, we propose to use the average distance between the poles on the unit sphere and their $k$-NN neighborhoods. More specifically, suppose $v_N$ and $v_S$ are the northernmost point and the southernmost point on the spherical parameterization result $f(P)$ obtained by Algorithm \ref{alg:pc_spherical} respectively. By the north-pole stereographic projection $P_N$, $v_N$ is mapped to the point $x_N$ on the complex plane. On the other hand, by the south-pole stereographic projection $P_N$, $v_S$ is mapped to the point $x_S$ on the complex plane. Denote the average distances of $x_N$ and $x_S$ to their $k$-NN neighborhood on their corresponding planar domain by $d_N$ and $d_S$ respectively. $d_N$ and $d_S$ are explicitly given by
\begin{equation}
 d_p = \text{mean}(\{|P_N(f(z)) - x_N|: z \in \mathcal{N}^{k}(f^{-1}(v_N))\})
\end{equation}
and
\begin{equation}
 d_s = \text{mean}(\{|P_S(f(z)) - x_S|: z \in \mathcal{N}^{k}(f^{-1}(v_S))\}).
\end{equation}
Then, we scale the whole planar domain $(P_N \circ f)(P)$ by a scaling factor
\begin{equation}
 \lambda = \frac{\sqrt{d_p \times d_s}}{d_p}.
\end{equation}
Now, denote the two updated average distances by $\tilde{d}_p$ and $\tilde{d}_s$. It follows that
\begin{equation}
 \tilde{d}_p = \lambda d_p = \frac{\sqrt{d_p \times d_s}}{d_p} \times d_p = \sqrt{d_p \times d_s}.
\end{equation}
Also, by Theorem \ref{invariance2}, we have
\begin{equation}
\tilde{d}_p \times \tilde{d}_s= d_p \times d_s.
\end{equation}
Therefore,
\begin{equation}
\tilde{d}_s = d_p \times d_s \times \frac{1}{\tilde{d}_p} = \sqrt{d_p \times d_s}.
\end{equation}
In other words, the two updated average distances $d_p$ and $d_s$ defined on the new spherical parameterization result $P_N^{-1}(\lambda(P_N(f(P))))$ are equal. This indicates that the distribution of the points at the two poles is balanced. Hence, Algorithm \ref{alg:pc_spherical} together with the described balancing scheme provide us with a a spherical conformal parameterization with an even distribution. Our balancing scheme is summarized in Algorithm \ref{alg:balance}.

\begin{algorithm}[h]
\KwIn{A spherical conformal parameterization $f:P \to \mathbb{S}^2$.}
\KwOut{A spherical conformal parameterization with improved distribution.}
\BlankLine
Apply the north-pole stereographic projection $P_N$ on $f(P)$\;
Denote the northernmost and the southernmost points of $f(P)$ by $v_N$ and $v_S$. Multiply all points in $P_N(f(P))$ by a scaling factor $\lambda = \frac{\sqrt{d_p \times d_s}}{d_p}$, where $d_p = \text{mean}(\{|P_N(f(z)) - x_N|: z \in \mathcal{N}^{k}(f^{-1}(v_N))\})$ and $d_s = \text{mean}(\{|P_S(f(z)) - x_S|: z \in \mathcal{N}^{k}(f^{-1}(v_S))\})$\;
Apply the inverse north-pole stereographic projection $P_N^{-1}$ on $\lambda(P_N(f(P)))$\;
\caption{Our proposed balancing scheme for better distribution.}
\label{alg:balance}
\end{algorithm}

\subsection{Meshing using spherical conformal parameterization}
In this subsection, we present our meshing framework for genus-0 point clouds. Directly triangulating a point cloud is difficult because of its complicated geometry. However, with the aid of the spherical conformal parameterization of point clouds, the difficulty is significantly alleviated. Instead of triangulating a point cloud, we triangulate the unit sphere obtained by our spherical conformal parameterization algorithm. Algorithms for triangulating a spherical point cloud are well-established. In particular, the spherical Delaunay triangulation method, which computes a Delaunay triangulation on the unit sphere, is the most suitable method for our purpose.

Delaunay triangulations are widely used in computer graphics because of their good triangle quality. More specifically, Delaunay triangulations are advantageous as they maximize the minimum angle in every triangle in the triangulations and hence avoid skinny triangles. With this important property, the triangulations generated by this method are more regular than the common triangulation methods.

By applying the spherical Delaunay triangulation method on the spherical conformal parameterization of a genus-0 point cloud, we obtain a nice triangulation on the spherical point cloud. Since the points on the original point cloud and those obtained by the spherical conformal parameterization have a 1-1 correspondence, the triangulation on the spherical point cloud naturally induces a triangulation on the original point cloud. It is noteworthy that since the parameterization is conformal, the angles of the new triangulation on the original point cloud are well preserved. In other words, the regularity of the triangulation defined on the original point cloud closely resembles that of the spherical Delaunay triangulation. Moreover, the meshing result is guaranteed to be a genus-0 closed triangular mesh because of the spherical Delaunay method. This completes our goal of meshing a genus-0 point cloud. Our meshing framework is described in Algorithm \ref{alg:meshing}.

\begin{algorithm}[h]
\KwIn{A genus-0 point cloud $P$.}
\KwOut{A triangular mesh $M = (P,T)$ where $T$ is a triangulation of $P$.}
\BlankLine
Apply Algorithm \ref{alg:pc_spherical} and Algorithm \ref{alg:balance} to obtain a spherical conformal parameterization $f: P \to \mathbb{S}^2$\;
Compute a triangulation $T$ on $f(P)$ using the spherical Delaunay algorithm\;
Use $T$ to form a triangular mesh $M = (P,T)$\;
\caption{Our proposed meshing framework for genus-0 point clouds.}
\label{alg:meshing}
\end{algorithm}

Before ending this section, we make an important remark about a possible extension of our proposed framework. In fact, our proposed parameterization and meshing scheme can be possibly extended for point clouds with disk topology. In this case, we can first extend the double covering technique \cite{Gu03,Choi15c} to turn a point cloud with disk topology into a genus-0 point cloud. More specifically, given a point cloud $P$ of a simply-connected open surface $\mathcal{M}$, we turn $P$ into a point cloud $\widetilde{P}$ with spherical topology and approximate the derivatives on it by the following steps.

\begin{enumerate}[Step 1: ]
 \item Approximate the derivatives on $P$ using the $k$-NN algorithm and the moving least square method.
 \item Duplicate $P$ and denote the copy of it by $P'$.
 \item Define the derivatives on $P'$ using the results in Step 1, with reversed orientations.
 \item Identify the boundary points of $P$, $P'$ and obtain a genus-0 point cloud $\widetilde{P}$.
 \item Create the LB operator for $\widetilde{P}$ using the derivatives on $P$ and $P'$.
\end{enumerate}

Then, we can apply the abovementioned spherical conformal parameterization algorithm on $\widetilde{P}$ to obtain a spherical point cloud. After that, by applying the stereographic projection on the southern hemisphere, we obtain a planar conformal parameterization of $P$. Finally, we can easily compute a Delaunay triangulation on the planar parameter domain. Since both the parameterization algorithm and the stereographic projection produce conformal results, this triangulation on the planar domain accurately induces a regular mesh structure on $P$. This completes the task of meshing a point cloud with disk topology.

\section{Experimental results} \label{experiment}
In this section, we demonstrate the effectiveness of our proposed framework for meshing genus-0 point clouds using spherical conformal parameterization. In the following, we assess the performance of our proposed framework in different aspects. The datasets used in the experiments are freely adapted from the AIM@SHAPE Shape Repository \cite{aim@shape}, the Stanford 3D Scanning Repository \cite{stanford} and the RGB-D Scenes Dataset v.2 \cite{Lai14b}. The mentioned algorithms are implemented in MATLAB. The sparse linear systems for the Laplace equations are solved using the built-in backslash operator (\textbackslash) in MATLAB. All experiments are performed on a PC with an Intel(R) Core(TM) i5-3470 CPU @3.20 GHz processor and 8.00 GB RAM.

\subsection{The performance of our approximation of the Laplace-Beltrami operator}
\begin{figure}[t]
\centering
\includegraphics[width=0.33\textwidth]{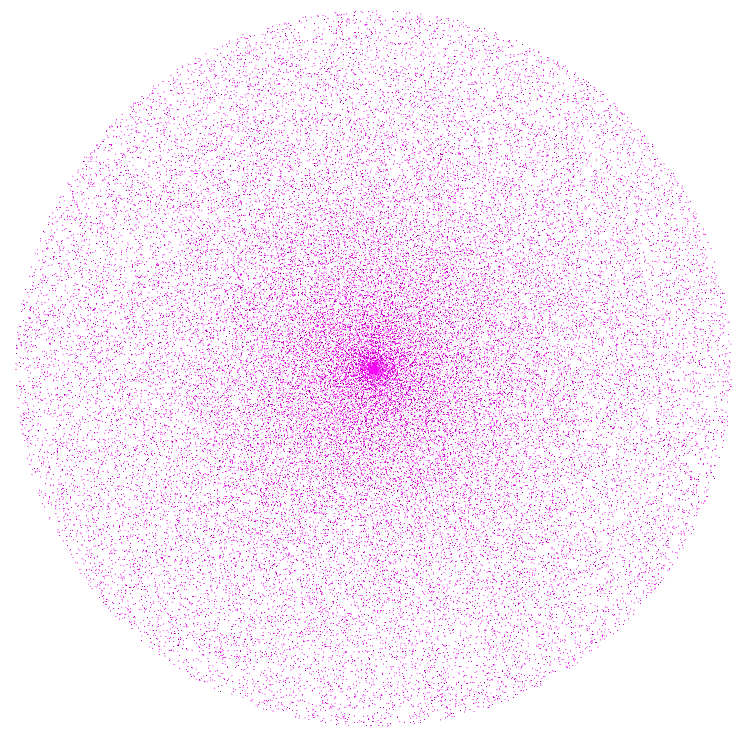}
\includegraphics[width=0.45\textwidth]{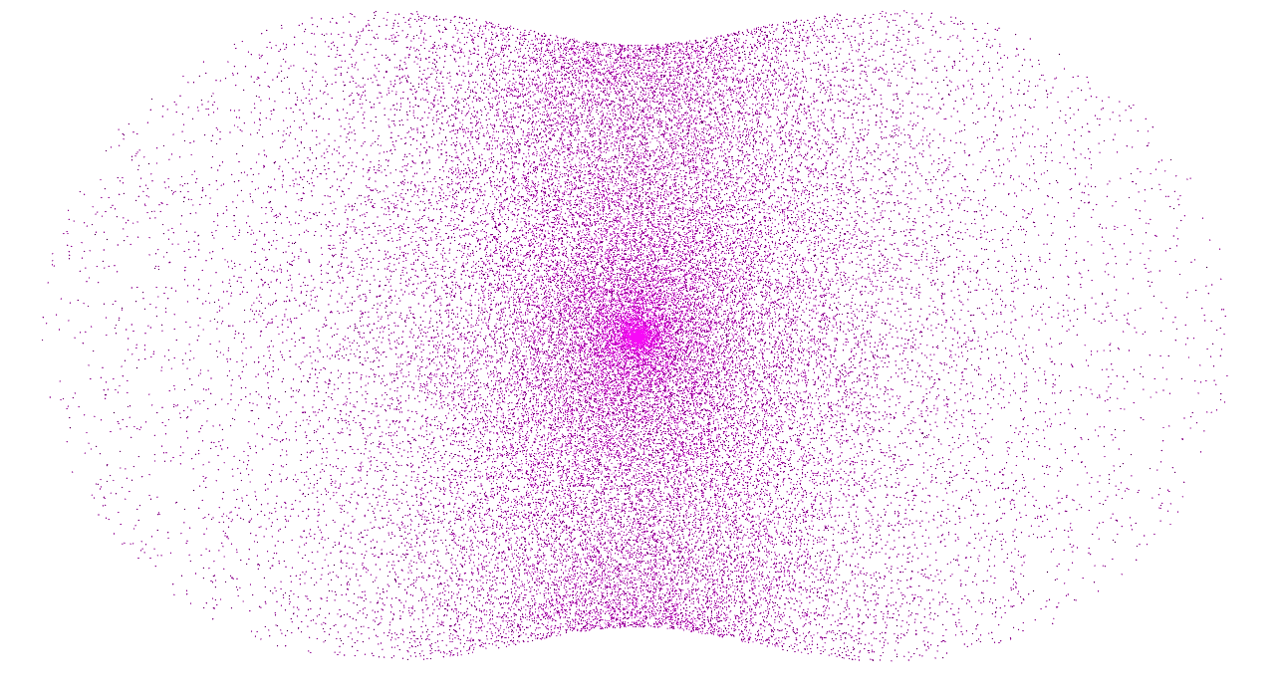}
\includegraphics[width=0.33\textwidth]{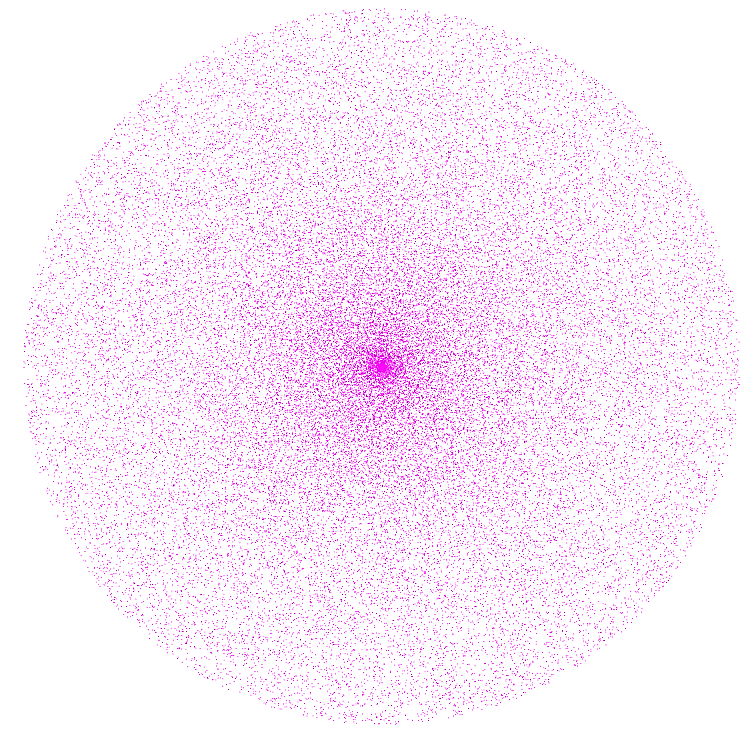}
\includegraphics[width=0.45\textwidth]{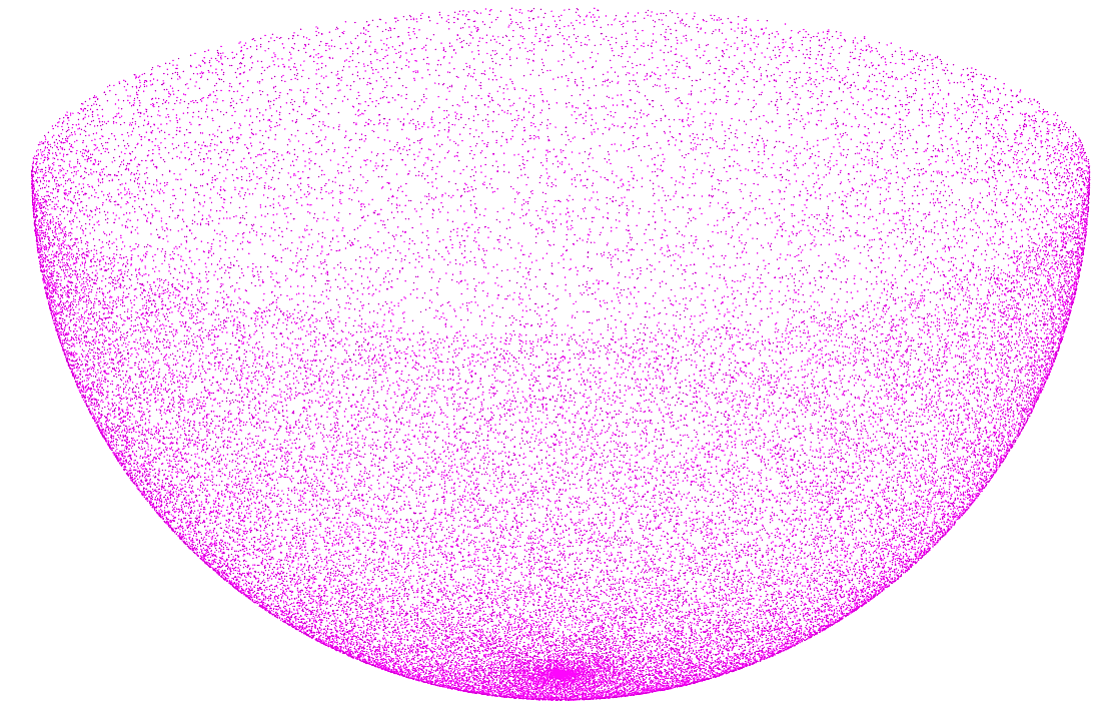}
\caption{Two experiments for assessing the approximation accuracies of the Laplace-Beltrami operator. In each experiment, we generate a point cloud on the unit disk and transform it using a conformal map with an explicit formula. We then approximate the LB operator on the transformed point cloud and solve the Laplace equation back onto the unit disk. Top: the first experiment. Bottom: the second experiment.}
\label{fig:example}
\end{figure}

\begin{table}
\centering
\begin{tabular}{|l|C{25mm}|C{25mm}|}
\hline
Method & maximum position error & average position error \\ \hline
Local mesh method \cite{Lai13} & 1.3427 & 0.0179 \\ \hline
MLS with the Wendland weight in \cite{Wendland95,Wendland01} & 3.3074 & 0.1696\\ \hline
MLS with the Gaussian weight in \cite{Belkin08,Belkin09} & 0.5697 & 0.0114\\ \hline
MLS with the special weight in \cite{Liang12} & 0.0427 & 0.0006\\ \hline
MLS with our proposed weight & 0.0245 & 0.0004\\ \hline
\end{tabular}
\\\bigskip
\begin{tabular}{|l|C{25mm}|C{25mm}|}
\hline
Method & maximum position error & average position error \\ \hline
Local mesh method \cite{Lai13} & 1.5148 & 0.0271 \\ \hline
MLS with the Wendland weight in \cite{Wendland95,Wendland01} &  2.0082  & 0.0803 \\ \hline
MLS with the Gaussian weight in \cite{Belkin08,Belkin09} & 1.5460 & 0.0925\\ \hline
MLS with the special weight in \cite{Liang12} & 0.0110  &  0.0001\\ \hline
MLS with our proposed weight &  0.0103  & 0.0002 \\ \hline
\end{tabular}
\bigbreak
\caption{The approximation error in the two experiments. Top: the first experiment. Bottom: the second experiment.}
\label{table:example}
\end{table}

In this work, we apply the MLS method with a new weight function for approximating the Laplace-Beltrami operator. It is natural to ask whether our proposed weight function produces better results. It is also necessary to compare other approximation approaches such as the local mesh method to justify our choice. In this subsection, we compare the numerical accuracy of the local mesh method and the MLS method with several weighting functions for approximating the LB operator on point clouds. More specifically, we compare the performance of the following methods:
\begin{enumerate}
\item The local mesh method \cite{Lai13},
\item The MLS method with the Wendland weight in \cite{Wendland95,Wendland01},
\item The MLS method with the Gaussian weight in \cite{Belkin08,Belkin09},
\item The MLS method with the special weight in \cite{Liang12}, and
\item The MLS method with our proposed weight function.
\end{enumerate}

Experiments are carried out for assessing the numerical accuracies of the abovementioned approaches. Figure \ref{fig:example} shows the setups in two of the experiments. In each experiment, we first generate a point cloud on the unit disk. This serves as the ground truth in our analysis. Then, we transform the point cloud using a conformal map with an explicit formula. We apply the mentioned approximation schemes for approximating the LB operator on the transformed point cloud. Then, we solve the Laplace equation with the circular boundary constraints on the original unit disk. Theoretically, the result obtained by the disk harmonic map should be exactly the same as the original point cloud, as the transformation is given by a conformal map with an explicit formula. In other words, the ideal position error between the disk harmonic map and the original point cloud should be 0. By measuring the maximum and average position error between the pairs of points, we can evaluate the accuracy of the aforementioned approximation schemes for approximating the LB operator.

Table \ref{table:example} illustrates the approximation error of different approaches in the two experiments. It is noteworthy that in both experiments, the MLS method with our proposed weight function produces approximations which are much more accurate those produced by the local mesh method and the MLS method with the Wendland weight \cite{Wendland95,Wendland01} and the Gaussian weight in \cite{Belkin08,Belkin09}. With similar and negligible average position errors, our proposed scheme reduces the maximum position errors by about 25\% on average when compared with the MLS method with the special weight \cite{Liang12}. The comparisons reflect the advantage of our proposed method for approximating the LB operator.

\subsection{Performance of our proposed spherical conformal parameterization}
\begin{figure}[t]
 \centering
 \includegraphics[width=0.31\textwidth]{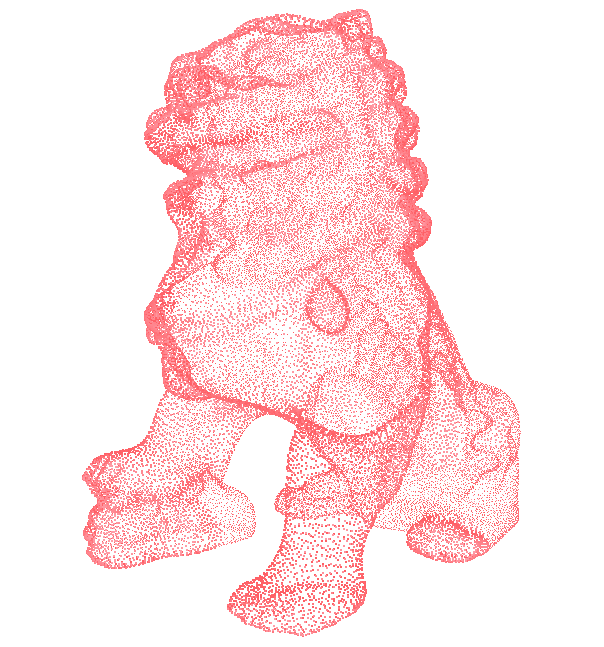}
 \includegraphics[width=0.35\textwidth]{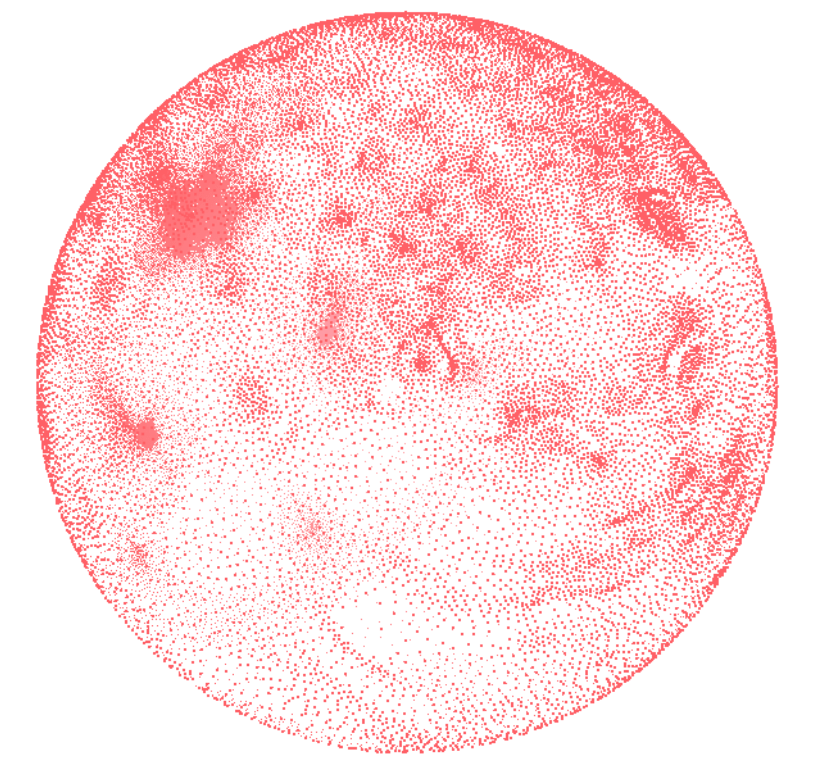}
 \includegraphics[width=0.26\textwidth]{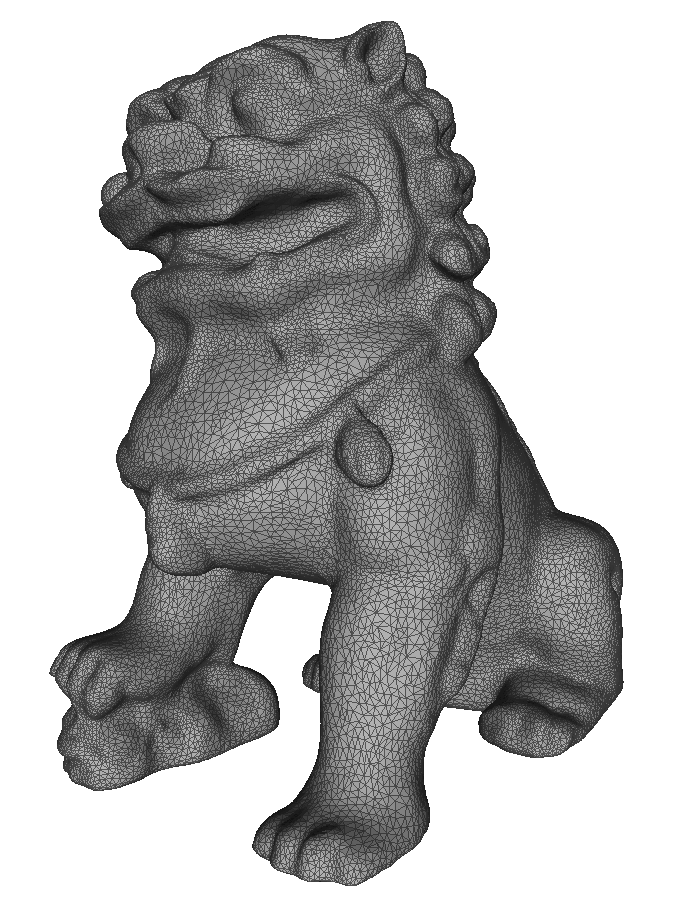}
 \newline
 \includegraphics[width=0.32\textwidth]{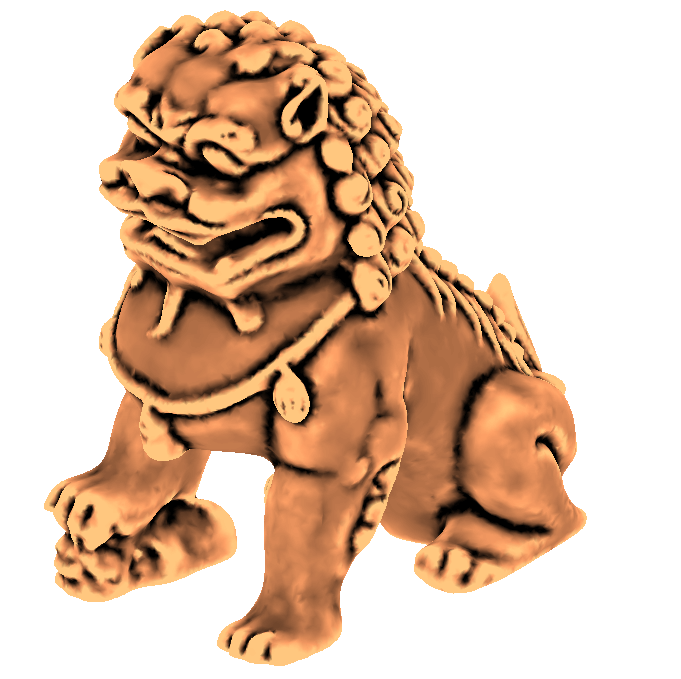}
 \includegraphics[width=0.32\textwidth]{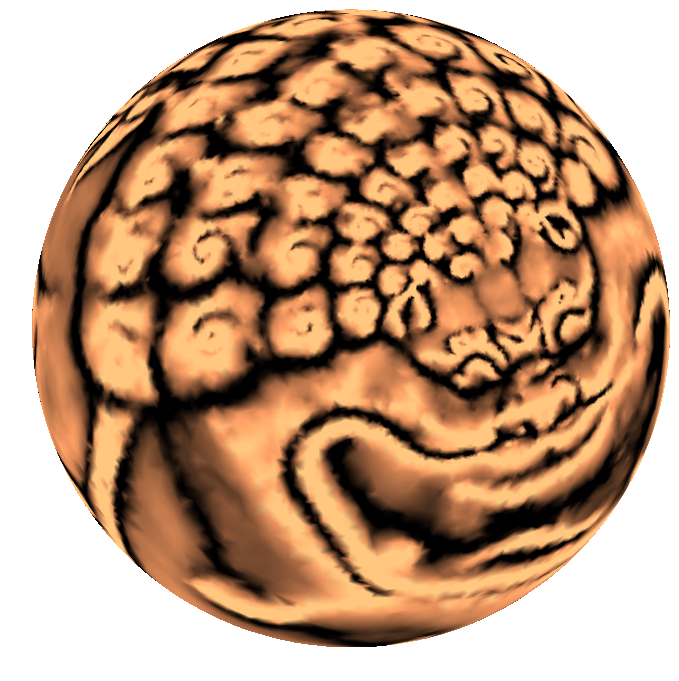}
 \includegraphics[width=0.32\textwidth]{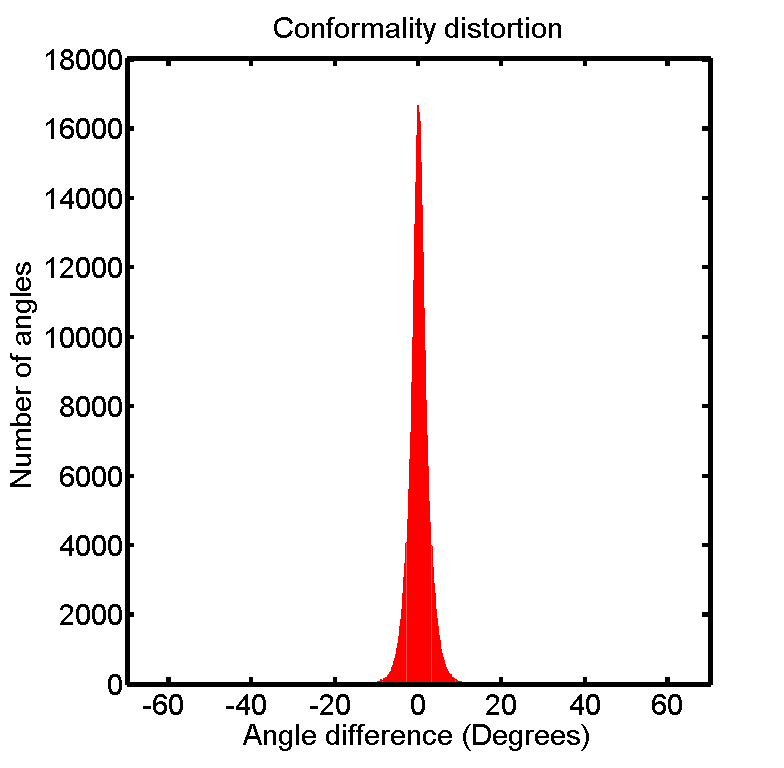}
 \caption{Parameterizing a lion point cloud. Top left: A lion point cloud. Top middle: The spherical conformal parameterization of the lion point cloud. Top right:Middle left: A triangulation created by our method. Bottom left and middle: The triangulated point cloud and the spherical parameterization result colored with the approximated mean curvature at each vertex. Bottom right: The conformality distortion of the parameterization based on the triangulation. }
 \label{fig:lion}
\end{figure}

\begin{figure}[t]
 \centering
 \includegraphics[width=0.32\textwidth]{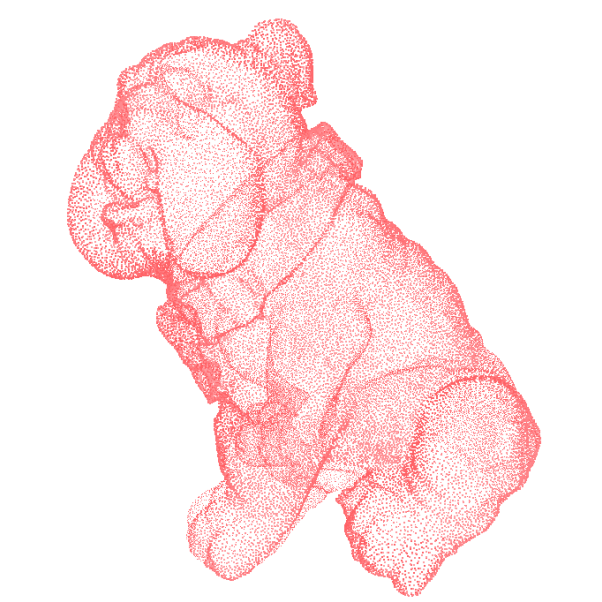}
 \includegraphics[width=0.32\textwidth]{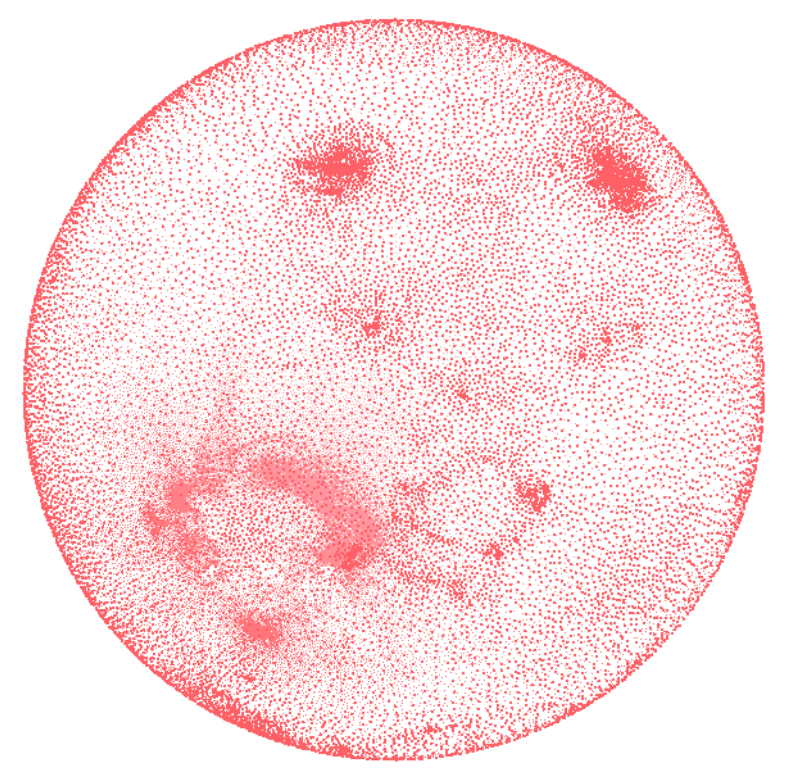}
 \includegraphics[width=0.29\textwidth]{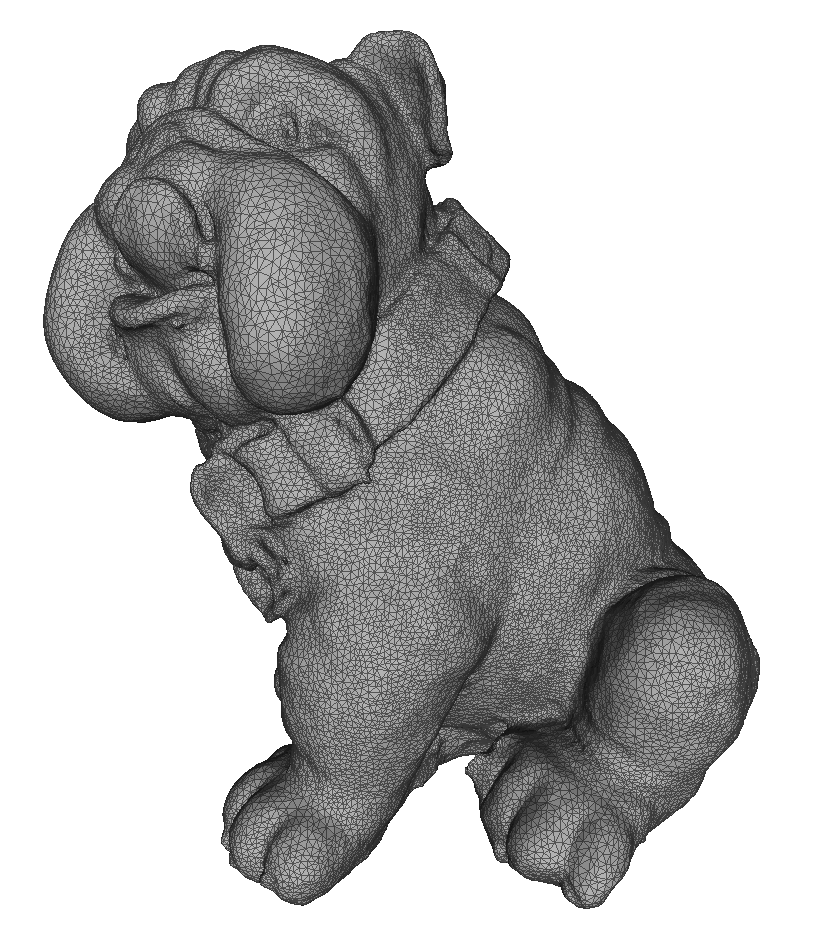}
 \newline
 \includegraphics[width=0.31\textwidth]{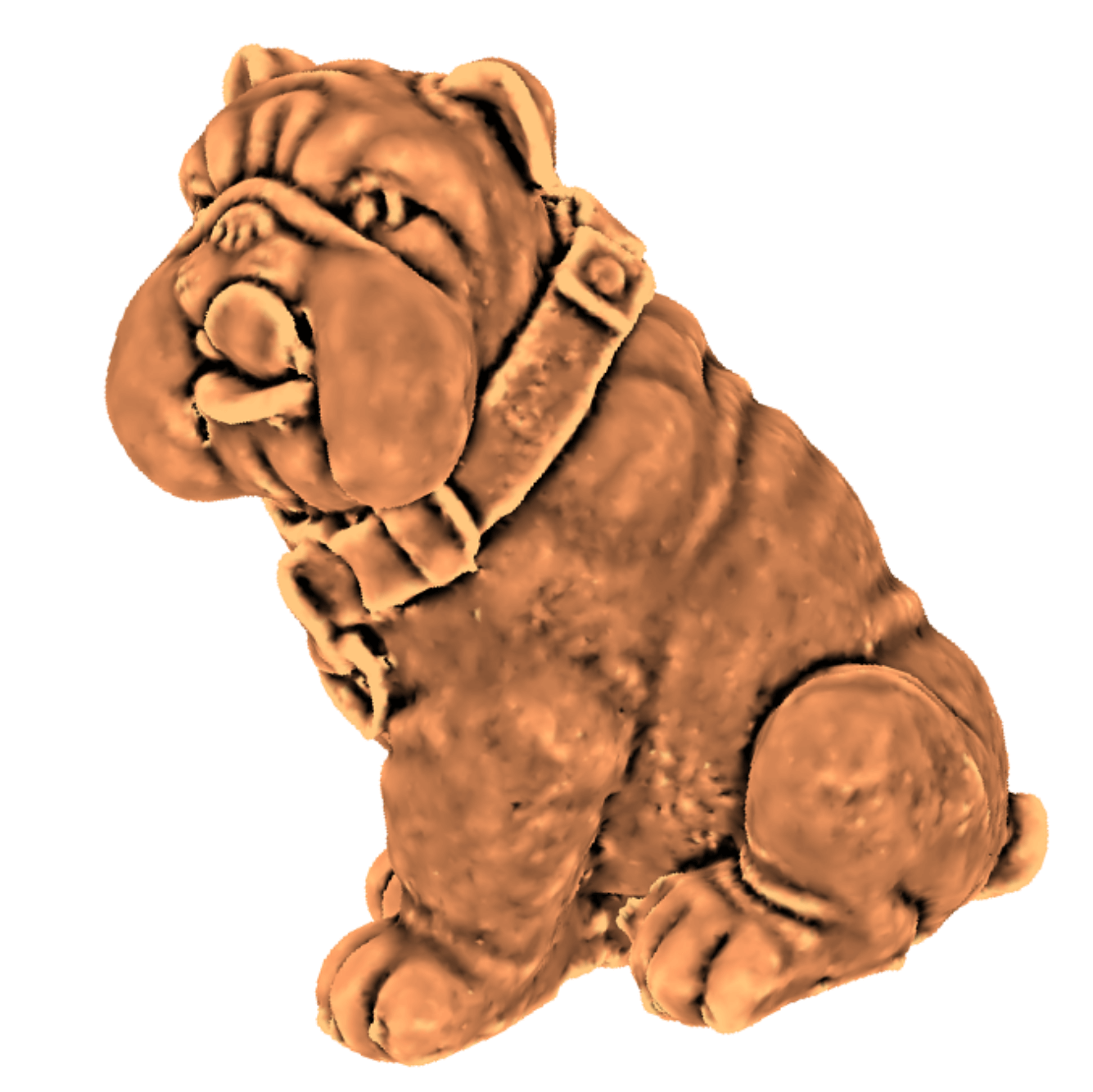}
 \includegraphics[width=0.31\textwidth]{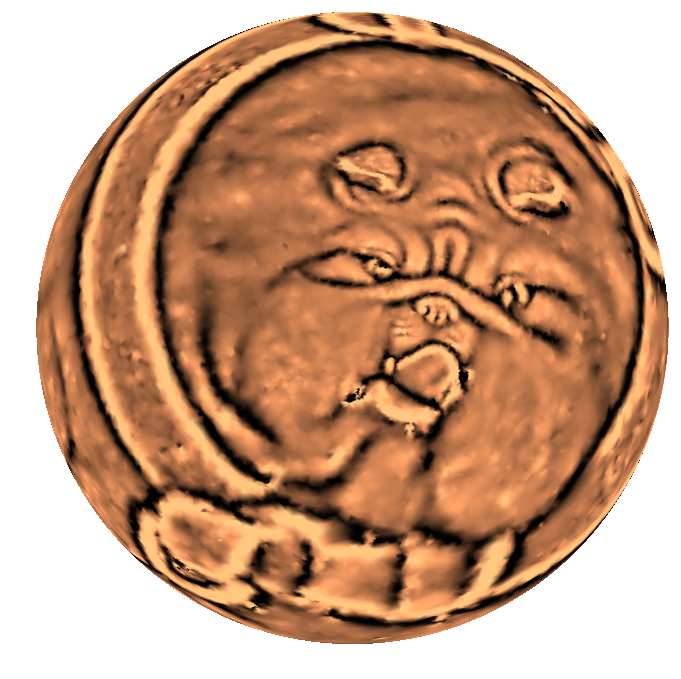}
 \includegraphics[width=0.33\textwidth]{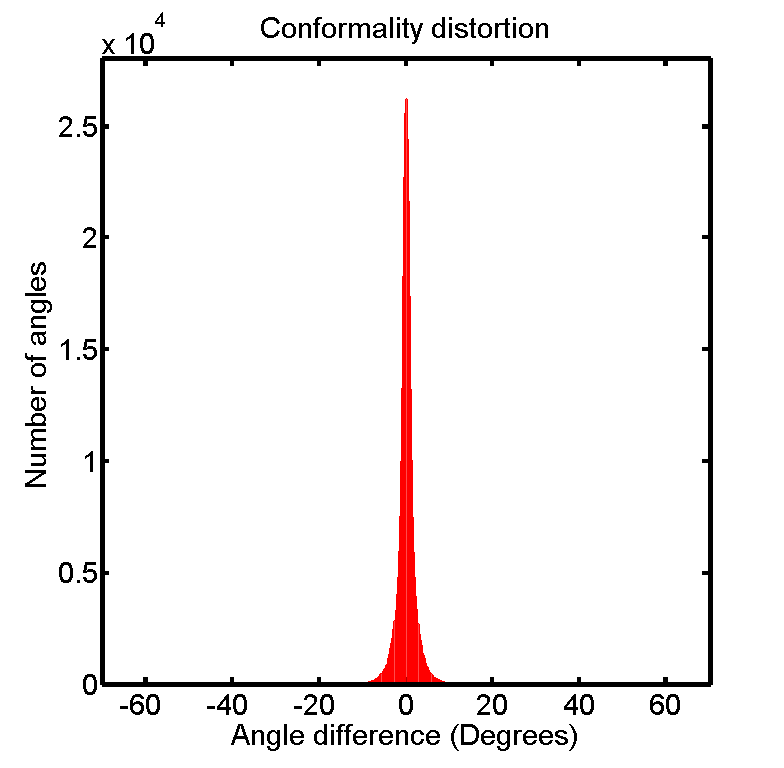}
 \caption{Parameterizing a bulldog point cloud. Top left: A bulldog point cloud. Top middle: The spherical conformal parameterization of the bulldog point cloud. Top right: A triangulation created by our method. Bottom left and middle: The triangulated point cloud and the spherical parameterization result colored with the approximated mean curvature at each vertex. Bottom right: The conformality distortion of the parameterization based on the triangulation. }
 \label{fig:bulldog}
\end{figure}

\begin{figure}[t]
 \centering
 \includegraphics[width=0.35\textwidth]{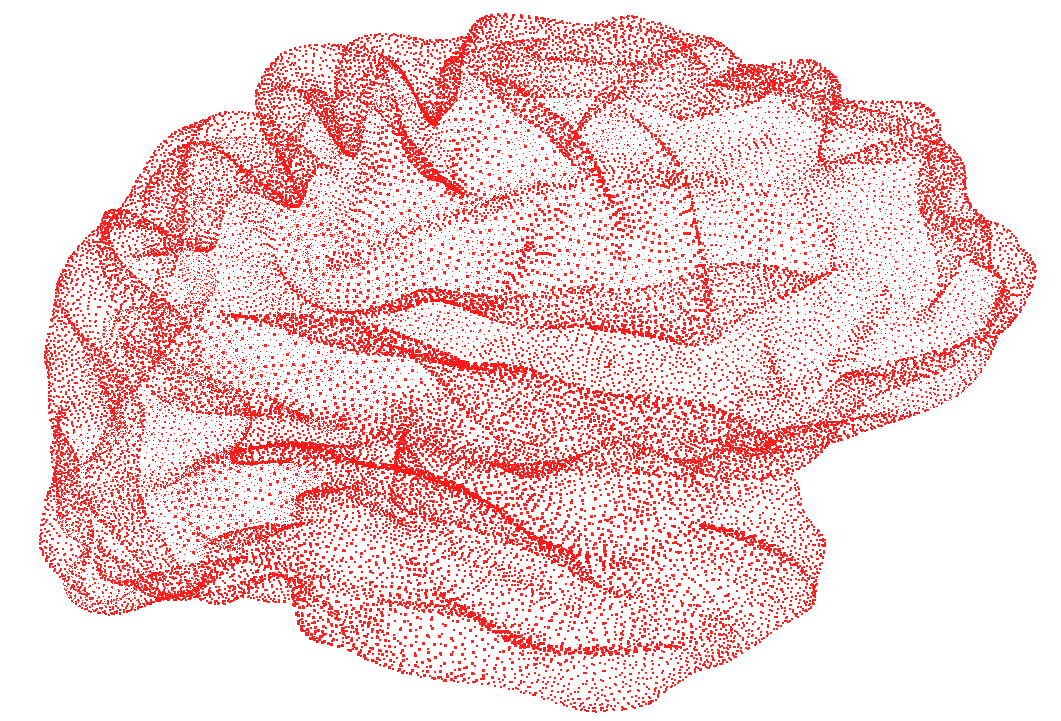}
 \includegraphics[width=0.28\textwidth]{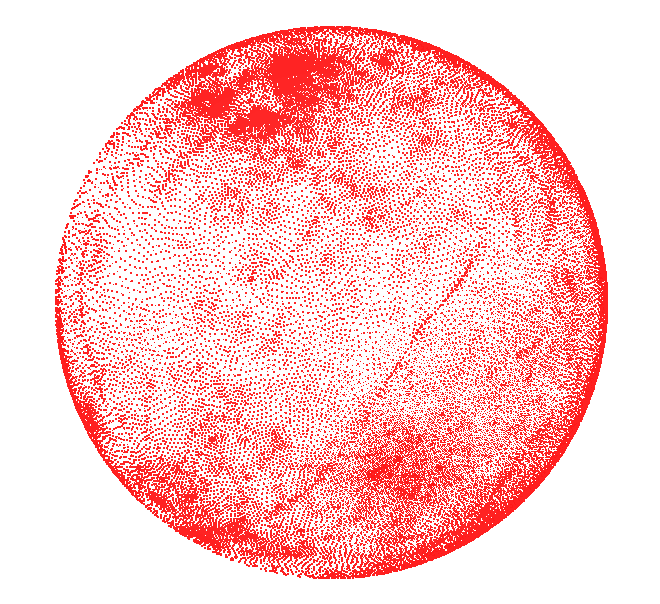}
 \includegraphics[width=0.35\textwidth]{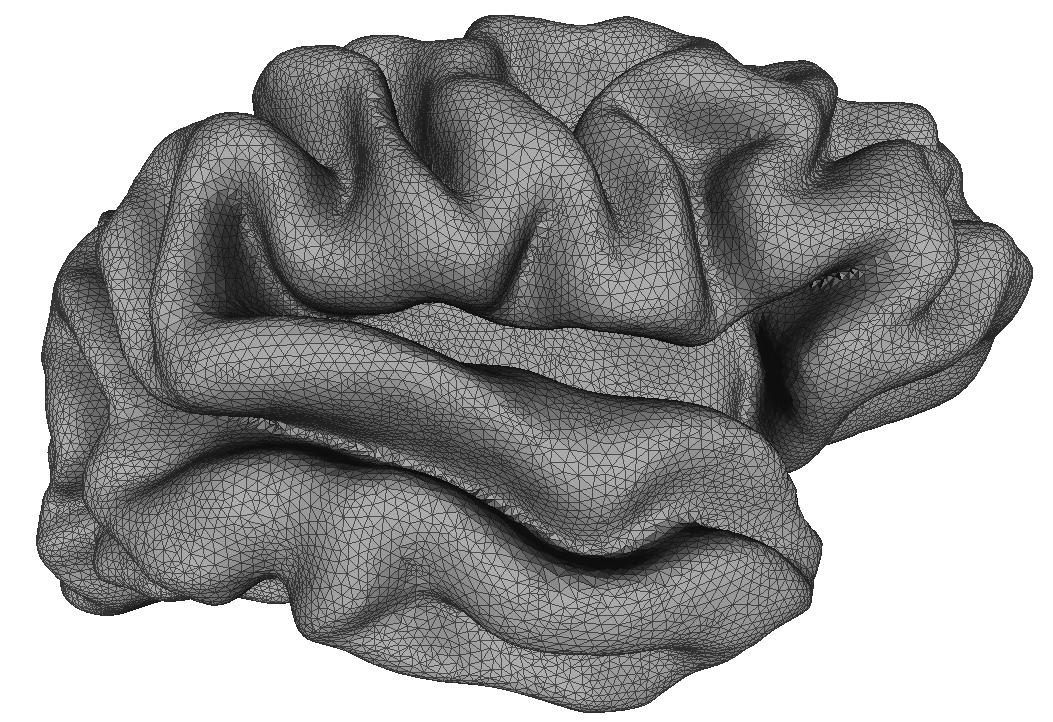}
 \newline
 \includegraphics[width=0.32\textwidth]{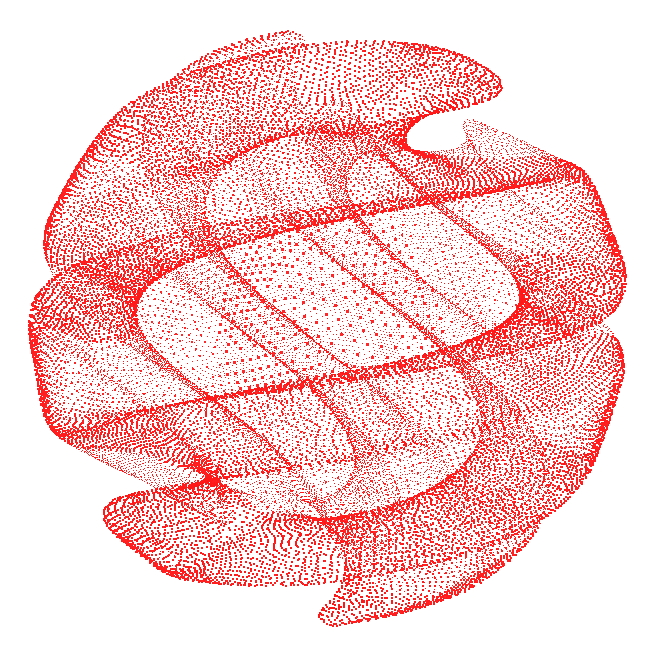}
 \includegraphics[width=0.3\textwidth]{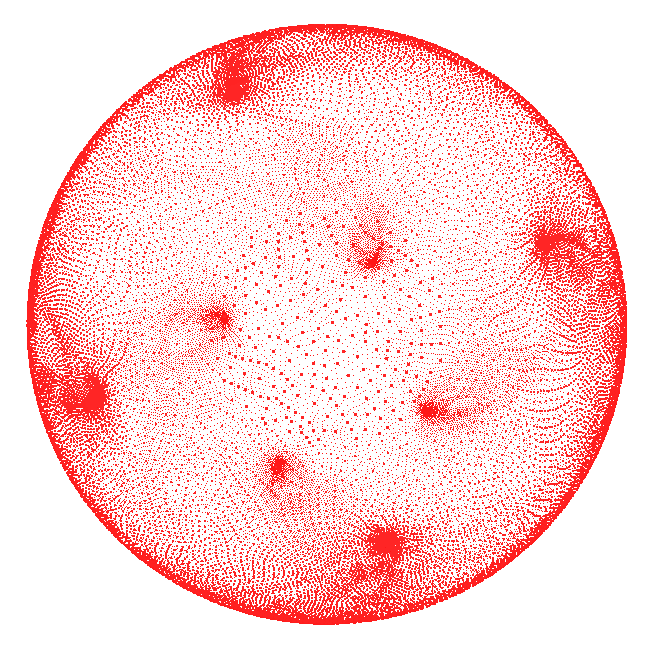}
 \includegraphics[width=0.32\textwidth]{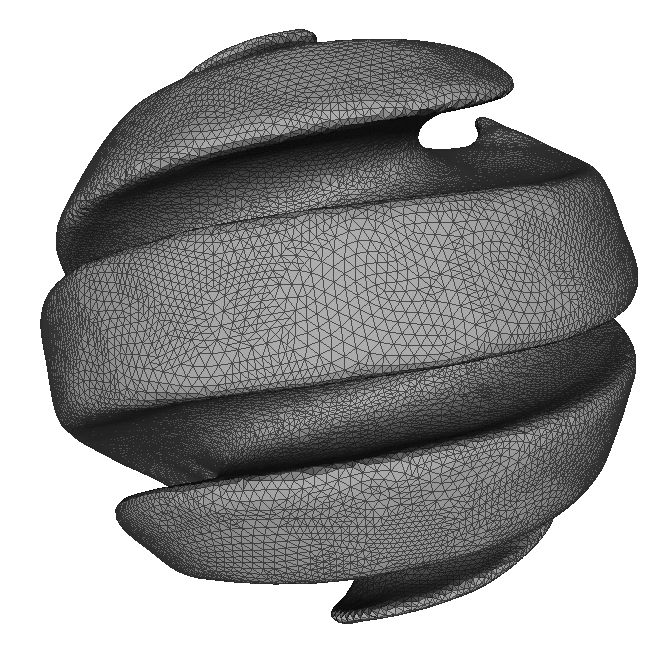}
 \caption{Parameterizing two convoluted point clouds of a human brain and a spiral. Left: The input point clouds. Middle: The spherical conformal parameterizations obtained by our proposed algorithm. Right: The triangulations created by our method.}
 \label{fig:convoluted}
\end{figure}

After demonstrating the advantage of our approximation scheme for the LB operator, we investigate the performance of our proposed spherical conformal parameterization algorithm for genus-0 point clouds. Figure \ref{fig:lion} and Figure \ref{fig:bulldog} show the results of parameterizing a lion point cloud and a bulldog point cloud using our proposed parameterization method. Two more convoluted examples are shown in Figure \ref{fig:convoluted}. The experiments demonstrate the effectiveness of our proposed algorithm for convoluted point cloud data.

\begin{figure}[t]
 \centering
\includegraphics[width=0.32\textwidth]{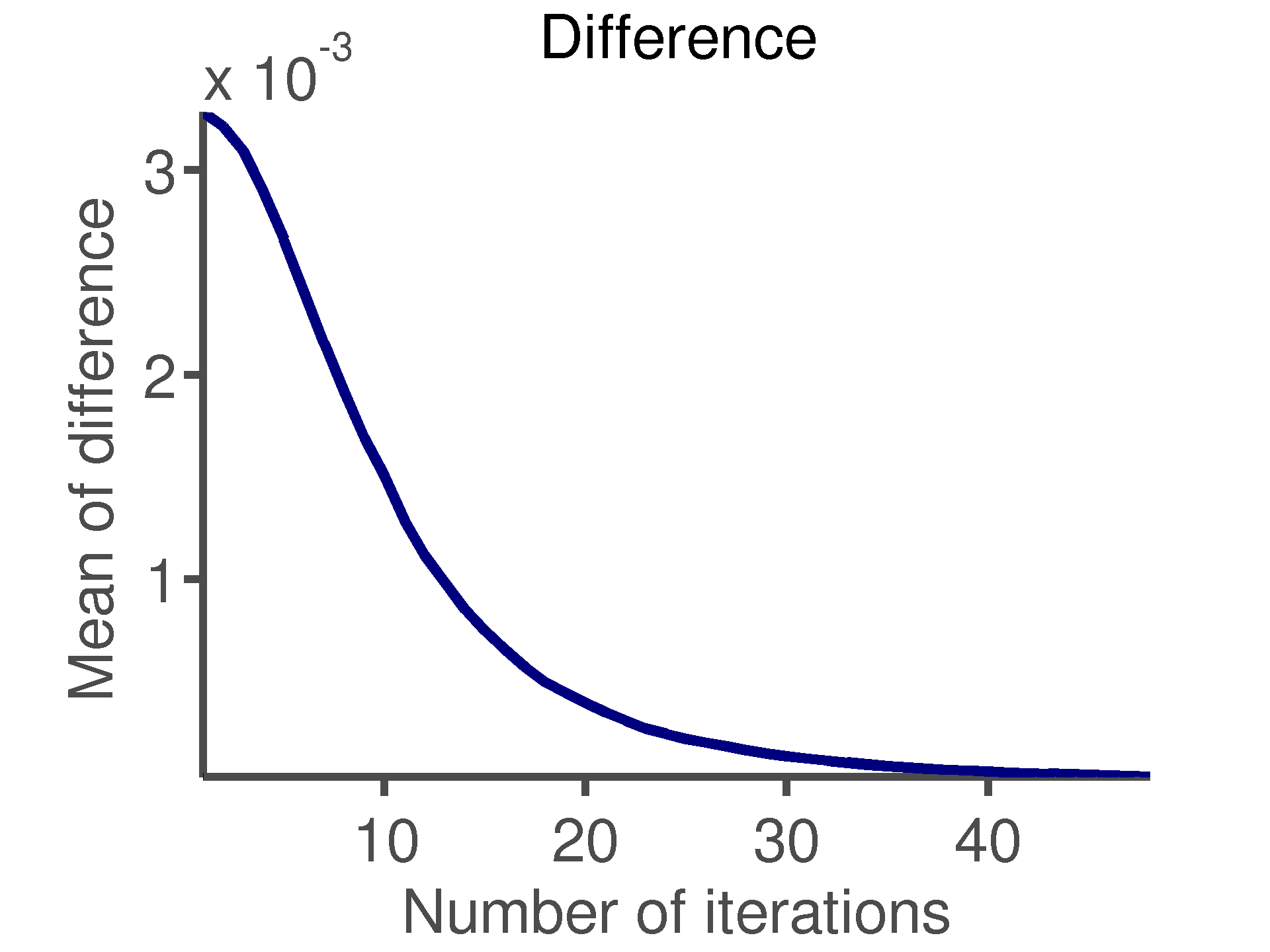}
 \includegraphics[width=0.32\textwidth]{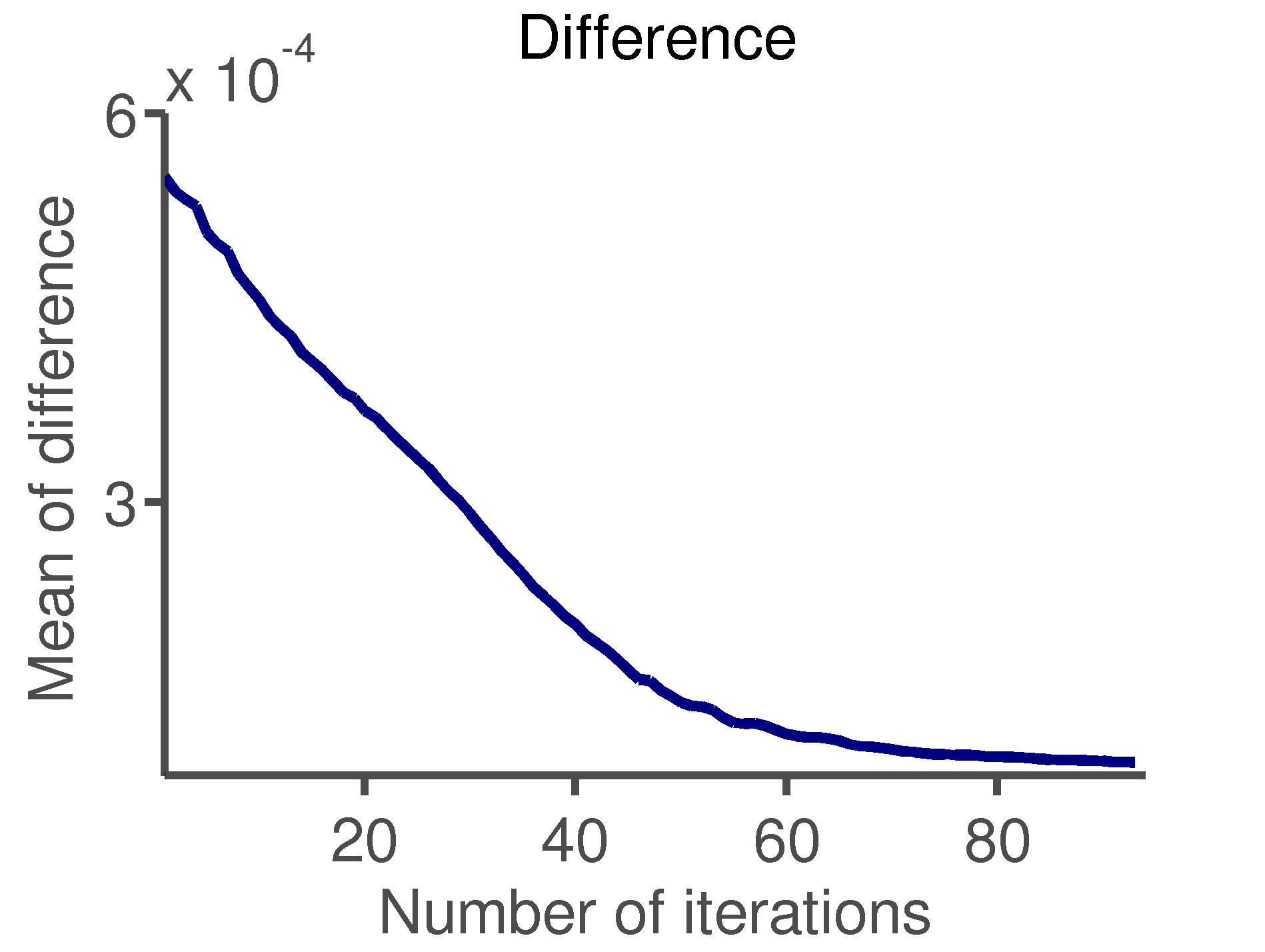}
 \includegraphics[width=0.32\textwidth]{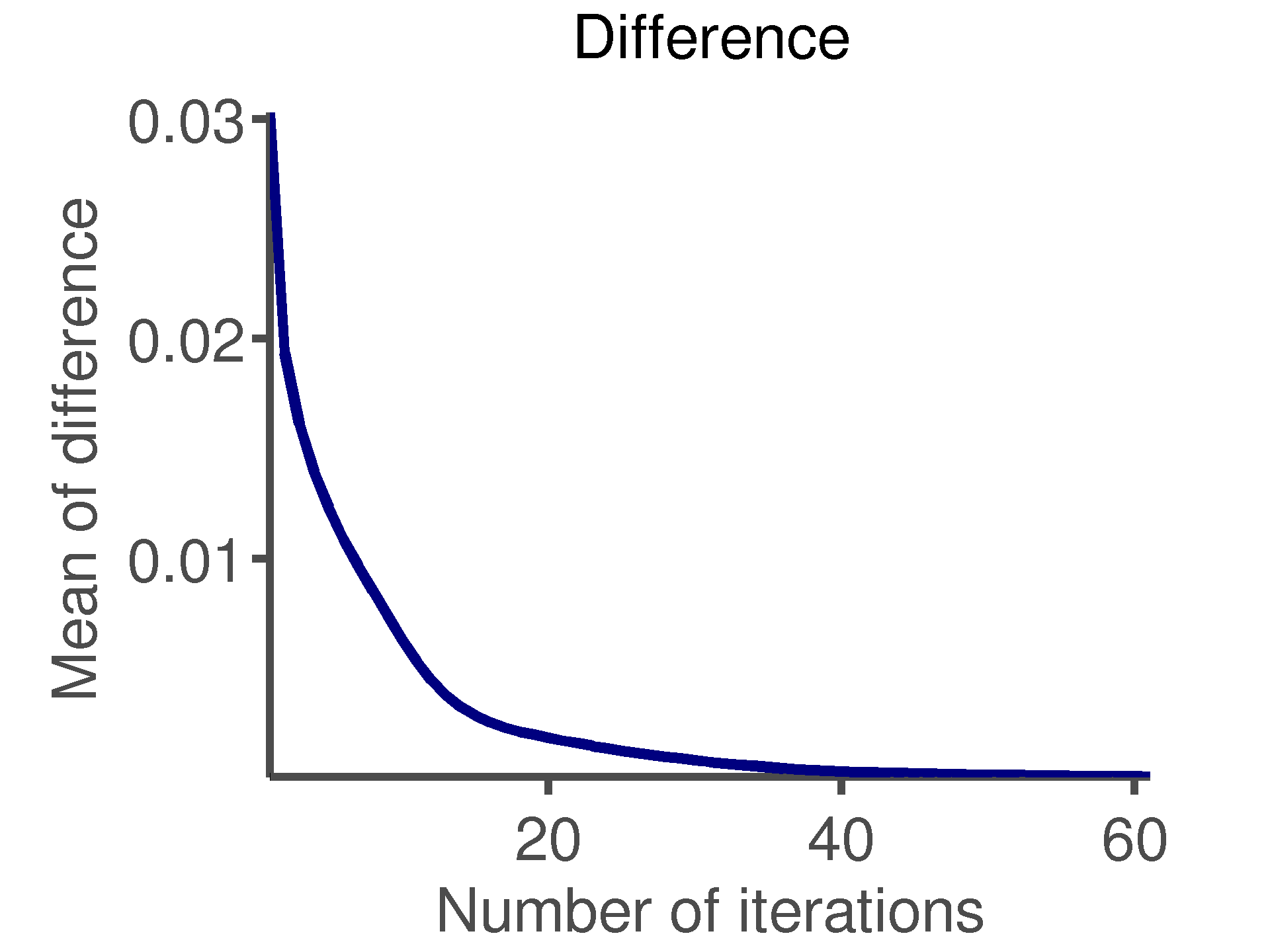}
 \caption{Several plots of the difference $\text{mean}(\|f(p_i) - \widetilde{f}(p_i)\|^2)$ with the number of iterations in our North-South reiteration scheme. Left: Cereal box. Middle: Hippocampus. Right: Bulldog. The plots demonstrate the convergence of our proposed algorithm.}
 \label{fig:energy}
\end{figure}

\begin{table}
    \begin{center}
    \begin{tabular}{ |l|C{12mm}|c|C{20mm}|C{20mm}|C{20mm}| }
    \hline
     Point clouds & No. of points & Performance & Our proposed method & Spherical embedding \cite{Zwicker04} & Global conformal map \cite{Liang12} \\ \hline

		  &	  & Time (s) 		& 8.1768  & 23.7988 & 16.7058	\\
      Soda Can    & 6838  & Mean($|\delta|$)	& 0.5902  & 4.0431  & 2.3352	\\
		  &	  & SD($|\delta|$)	& 0.8007  & 5.2731  & 1.9803	\\ \hline
		
		  &	  & Time (s) 		& 13.0919 & 37.4124 & 18.7151	\\
      Hippocampus & 10242 & Mean($|\delta|$)	& 1.2855  & 14.3072 & 1.3062	\\
		  &	  & SD($|\delta|$)	& 1.4701  & 19.6461 & 1.5100	\\ \hline
			
		  &	  & Time (s) 		& 30.7785 & 87.0887 & 40.4214	\\
      Max Planck  & 21530 & Mean($|\delta|$)	& 0.7326  & 8.6058  & 1.0792	\\
		  &	  & SD($|\delta|$)	& 1.0803  & 14.0857 & 1.5756	\\ \hline
			
		  &	  & Time (s) 		& 50.7390 & 132.1765& 		\\
      Cereal Box  & 33061 & Mean($|\delta|$)	& 0.6523  & 12.3573 & Fail	\\
		  &	  & SD($|\delta|$)	& 0.9165  & 14.0440 & 		\\ \hline
		
		  &	  & Time (s) 		& 114.7057& 291.8312& 122.8818	\\
      Spiral 	  & 48271 & Mean($|\delta|$)	& 0.8580  & 16.4704 & 0.9658	\\
		  &	  & SD($|\delta|$)	& 1.3280  & 22.5073 & 1.3135	\\ \hline
			
		  &	  & Time (s) 		& 115.5802& 198.2285& 126.7621	\\
      Brain 	  & 48487 & Mean($|\delta|$)	& 1.4266  & 35.2629 & 2.2495	\\
		  &	  & SD($|\delta|$)	& 2.9093  & 35.7986 & 2.7430	\\ \hline
			
		  &	  & Time (s) 		& 88.9297 & 206.9920& 113.4447	\\
      Bulldog	  & 49797 & Mean($|\delta|$)	& 1.5432  & 16.2010 & 1.8700	\\
		  &	  & SD($|\delta|$)	& 2.9183  & 21.1544 & 3.1891	\\ \hline
			
		  &	  & Time (s) 		& 95.8935 & 212.5685& 136.3296	\\
      Chinese Lion& 50002 & Mean($|\delta|$)	& 1.8474  & 19.1579 & 2.4907	\\
		  &	  & SD($|\delta|$)	& 1.9286  & 22.7259 & 2.6207	\\ \hline
			
		  &	  & Time (s) 		& 198.6064& 360.7178& 227.0290	\\
      Bimba 	  & 74764 & Mean($|\delta|$)	& 0.6227  & 18.0340 & 0.6379	\\
		  &	  & SD($|\delta|$)	& 0.8129  & 20.6272 & 0.7975	\\ \hline
			
		  &	  & Time (s) 		& 427.7658& 731.8661& 560.6077	\\
      Igea 	  & 134345& Mean($|\delta|$)	& 0.7076  & 5.0853  & 3.8293	\\
		  &	  & SD($|\delta|$)	& 1.4273  & 8.2623  & 2.9703	\\ \hline
			
		  &	  & Time (s) 		& 676.4106& 995.7537& 		\\
      Armadillo   & 172974& Mean($|\delta|$)	& 1.4167  & 23.2354 & Fail	\\
		  &	  & SD($|\delta|$)	& 1.6855  & 23.9892 & 		\\ \hline
			
		  &	  & Time (s) 		&1305.9013&1484.7682& 1642.9208	\\
      Lion Vase   & 256094& Mean($|\delta|$)	& 2.0920  & 17.8501 & 3.6696 	\\
		  &	  & SD($|\delta|$)	& 4.1052  & 21.9588 & 5.8502	\\ \hline

    \end{tabular}
    \end{center}
    \bigbreak
    \caption{Performances of three spherical conformal parameterization methods for genus-0 point clouds. To quantitatively evaluate the conformality of the parameterization, we build a mesh structure on the spherical parameterization using the spherical Delaunay method and then create an induced mesh structure on the original point cloud. The conformality distortion (denoted by $\delta$) of the parameterization is assessed using the angle difference (in degrees) between an angle on a meshed point cloud and the mapped angle on the meshed spherical parameterization result.}
    \label{table:performance}
\end{table}

Moreover, with the aid of the spherical conformal parameterization, we can create a Delaunay triangulation on the spherical parameterization result by the spherical Delaunay algorithm and define an induced triangulation on the input point cloud. Using the mesh structures, we can measure the angle differences of the two meshed point clouds and hence effectively evaluate the conformality of our parameterization scheme. In particular, we define the conformality distortion of the parameterization by the angular distortion between the two meshes. The angle difference provides an accurate measurement of the conformality distortion of the parameterizations. It can be easily observed in Figure \ref{fig:lion} and Figure \ref{fig:bulldog} that the histograms of the angle differences highly concentrate at 0. Besides, for better visualizations of the spherical conformal parameterization results, we color the surfaces by the approximated mean curvature on the source surfaces. It can be observed from the colored figures that the local geometries of the point clouds are well preserved under our proposed spherical conformal parameterization algorithm. Figure \ref{fig:energy} shows the difference plots of several experiments. The plots demonstrate the convergence of our North-South reiteration scheme.

We compare our proposed spherical conformal parameterization method for genus-0 point clouds with the spherical embedding method proposed by Zwicker and Gotsman \cite{Zwicker04} and the global conformal map \cite{Liang12}. In our experiment, $k = 25$ nearest neighbors of every point are used for approximating the LB operator in Algorithm \ref{alg:pc_spherical}. The stopping threshold for the N-S reiteration is set to be $\epsilon = 0.0001$. Table \ref{table:performance} summarizes the computational time and the conformality distortion of the three schemes. It is noteworthy that our proposed method produces spherical conformal parameterizations with better conformality. The better conformality obtained by our method is attributed to the south-pole step in our algorithm, which is conceptually equivalent to a composition of quasi-conformal maps. With this specific step, our method can hence further reduce the conformality distortion of the parameterization. Moreover, our method is more efficient than the algorithms in \cite{Zwicker04} and \cite{Liang12}. The above results indicate that our parameterization algorithm preserves the local geometry of the point clouds very well.

\subsection{Performance of our meshing scheme}
\begin{figure}[t]
 \centering
 \includegraphics[width=0.33\textwidth]{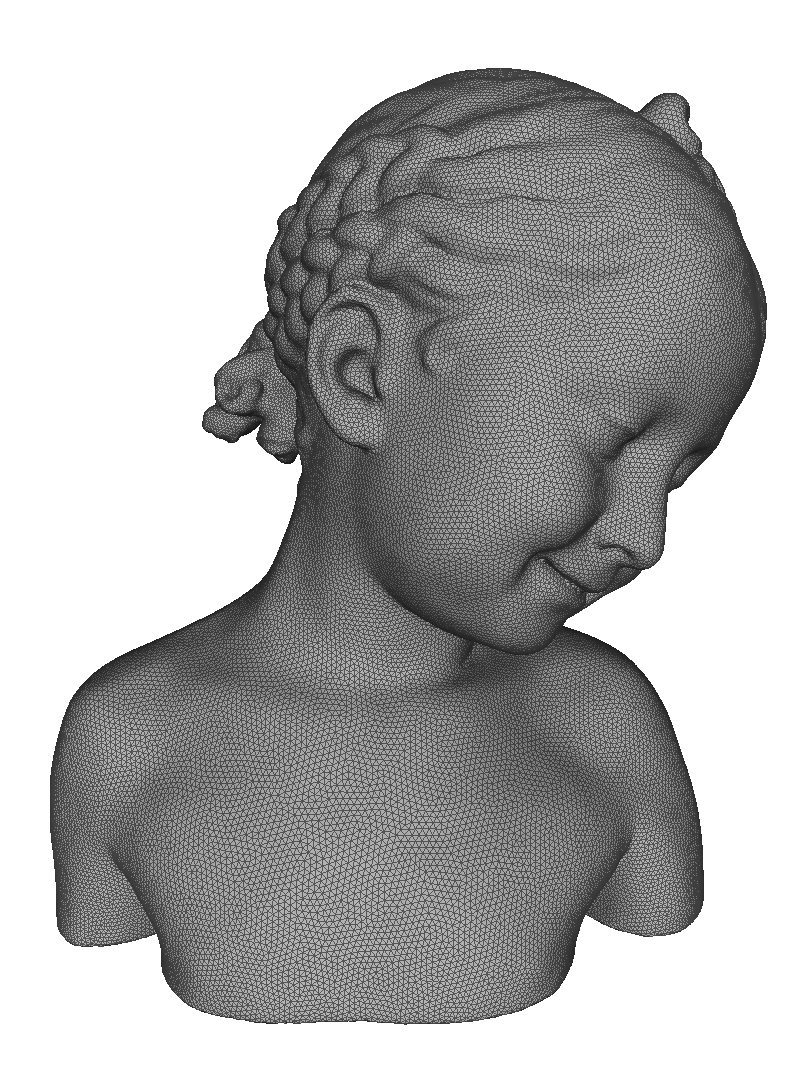}
 \includegraphics[width=0.31\textwidth]{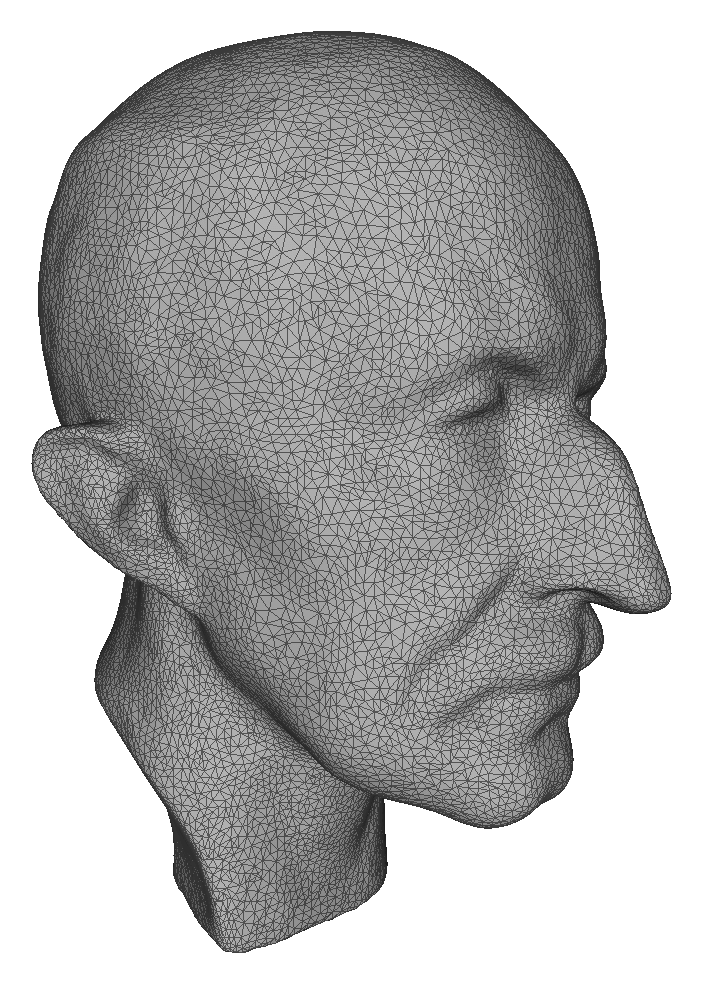}
 \includegraphics[width=0.33\textwidth]{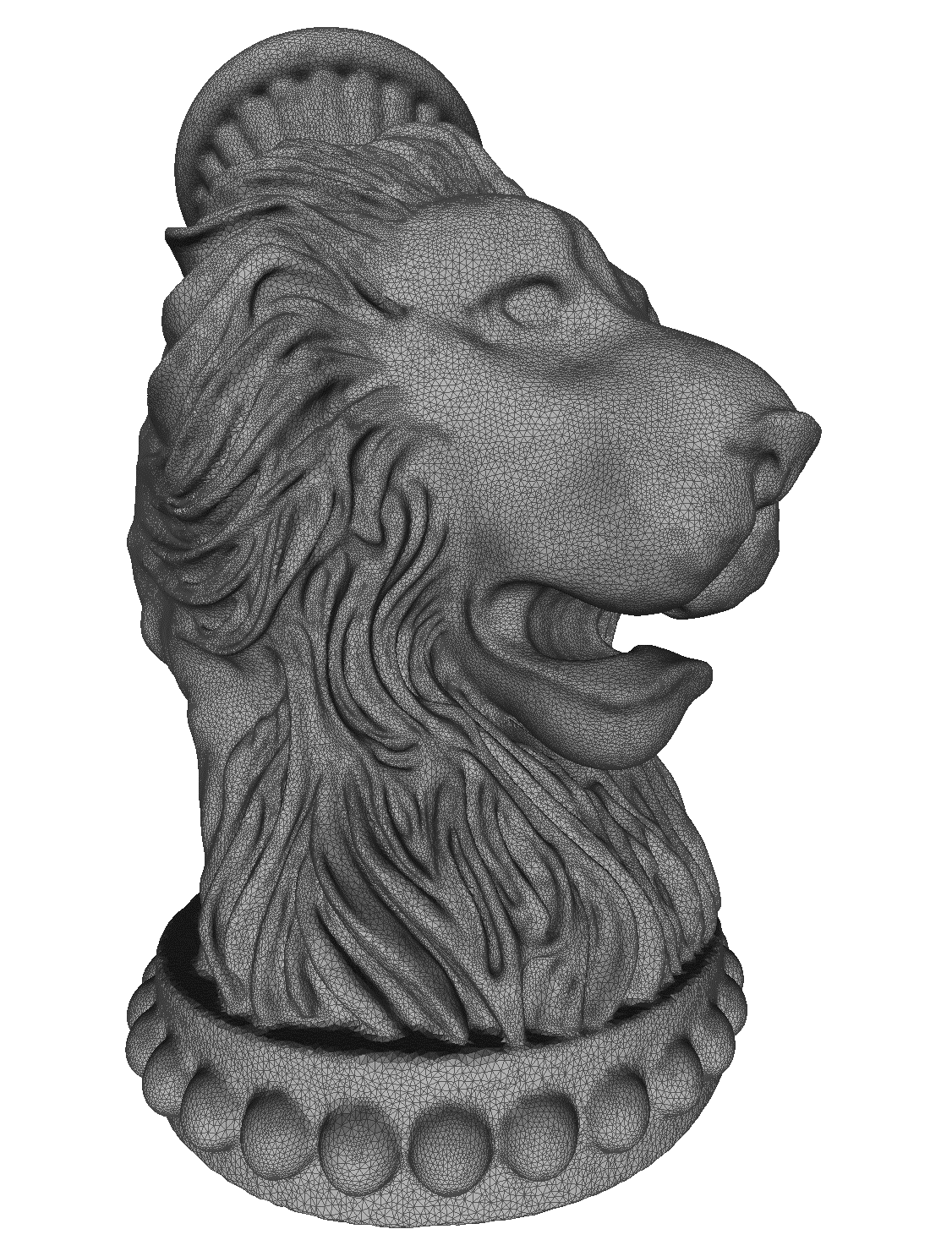}
 \includegraphics[height=3.2cm]{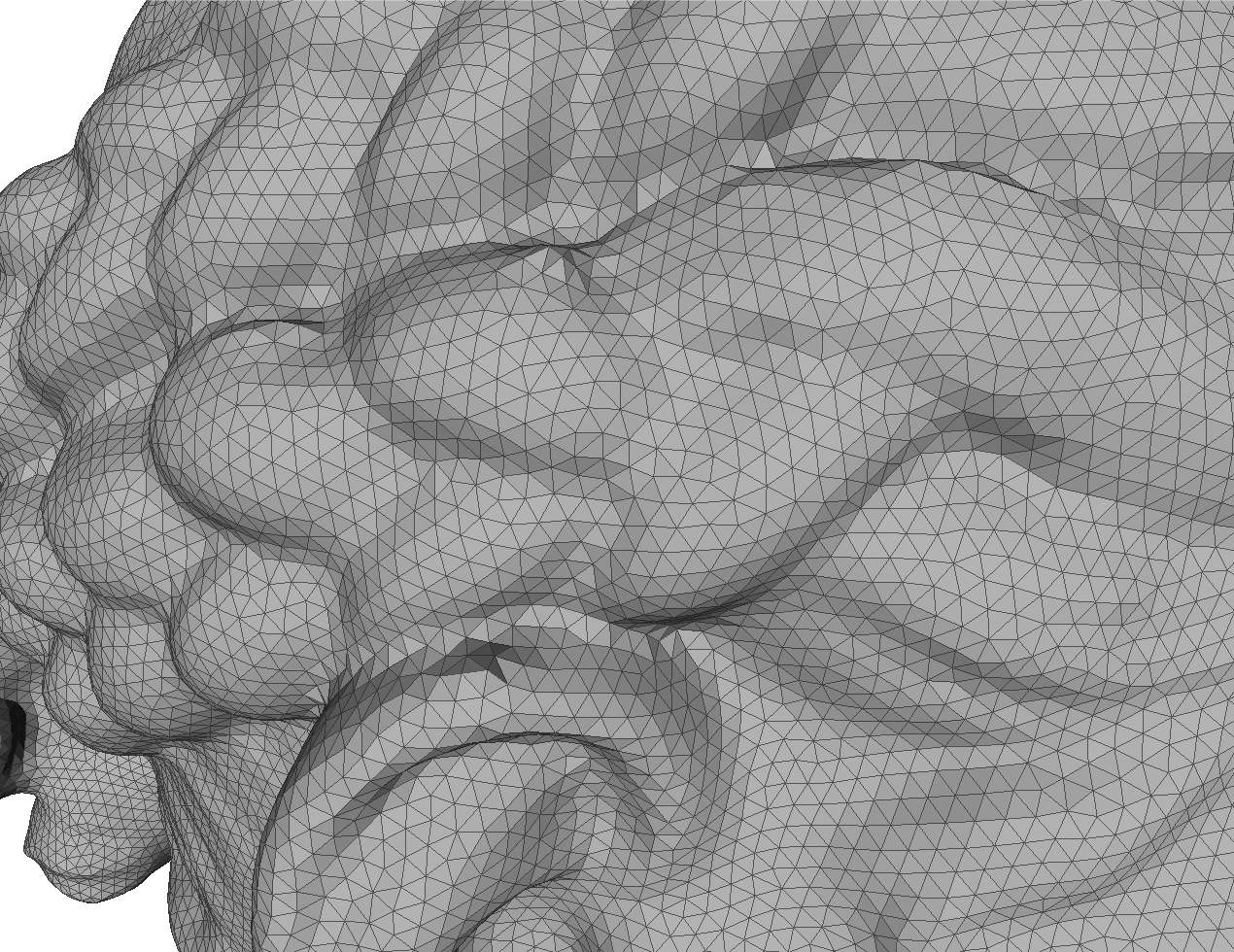}
 \includegraphics[height=3.2cm]{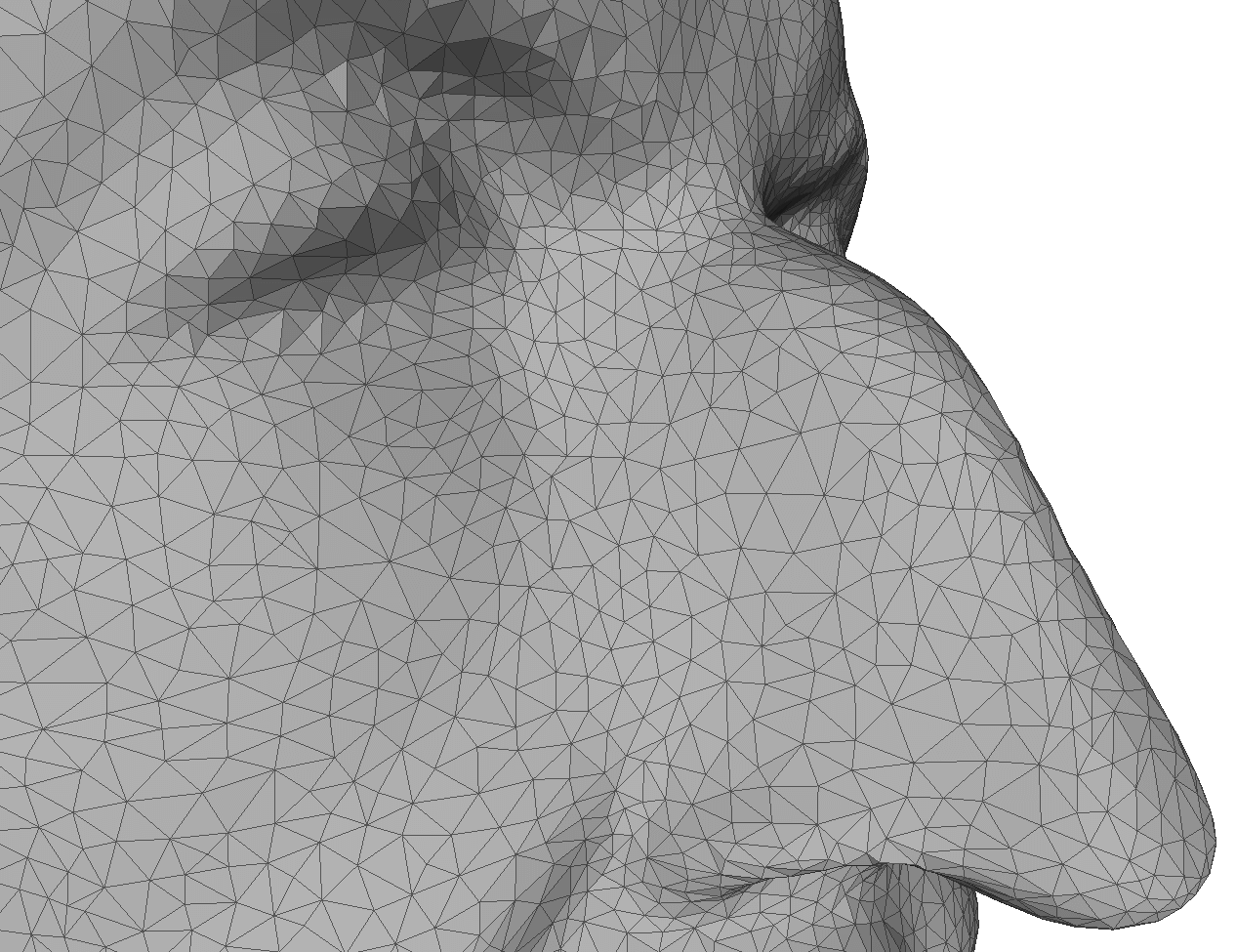}
 \includegraphics[height=3.2cm]{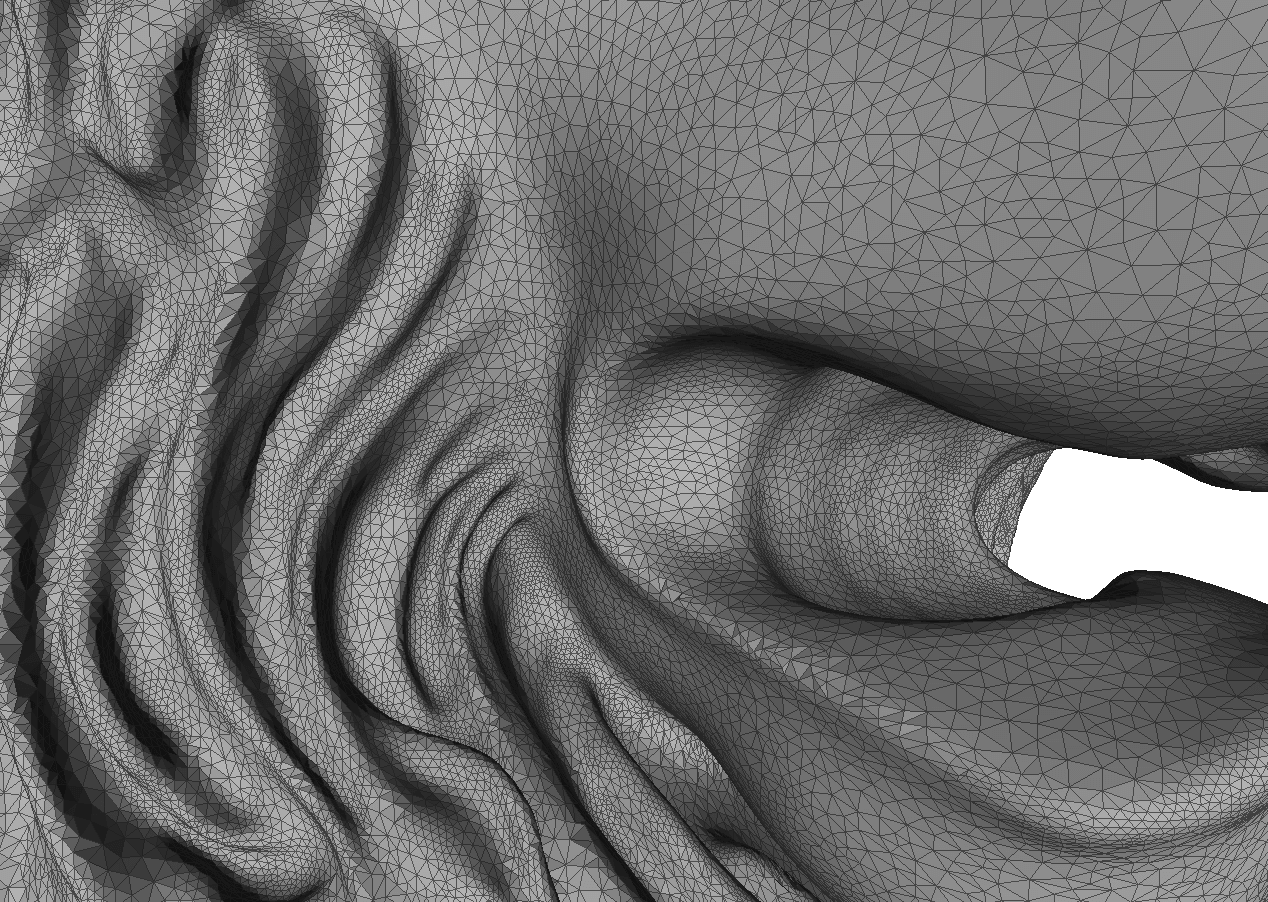}
 \caption{Meshes generated by our proposed method and a zoom-in of them. The regularity of the triangulations is attributed to our spherical conformal parameterization and the spherical Delaunay method.}
 \label{fig:mesh_zoom}
\end{figure}

As mentioned in the last subsection, we generate mesh structures on genus-0 point clouds by building Delaunay triangulations on their spherical conformal parameterizations. Our meshing scheme has two important advantages. First, the regularity of the triangulations generated is guaranteed by the preservation of the angle structures of the Delaunay triangulations computed on the spherical parameterizations. As the angle structures are well retained under the spherical conformal parameterization, a regular triangulation defined on the parameterized point clouds can effectively induce a regular and almost-Delaunay triangulation on the original point clouds. Besides Figure \ref{fig:lion}, Figure \ref{fig:bulldog} and Figure \ref{fig:convoluted}, some more examples of triangulations created by our approach are shown in Figure \ref{fig:mesh_zoom}. It can be observed that our meshing method can handle point clouds with different geometry. High quality triangulations can be created even with the presence of sharp, non-convex and convoluted regions on the input point clouds. Second, unlike most of the existing meshing methods, the meshes produced by our proposed scheme are guaranteed to be genus-0 closed meshes. No holes or unwanted boundaries will be present in our meshing result. Hence, post-processing steps are not required in our meshing scheme.

We compare our meshing method with three existing meshing approaches. As an example of parameterization-based approaches, Zwicker and Gotsman \cite{Zwicker04} generate triangulations for a genus-0 point cloud with the aid of the spherical embedding algorithm and the spherical Delaunay triangulation method. On the other hand, two typical methods for meshing without using parameterizations are the marching cubes algorithm \cite{Lorensen87} and the Tight Cocone algorithm \cite{Dey03}. Figure \ref{fig:meshing_results} provides a comparison between our method and the three mentioned approaches. It can be observed that our meshing scheme and the Tight Cocone algorithm \cite{Dey03} produce high quality triangulations, while the triangulations produced by the approaches in \cite{Zwicker04} and \cite{Lorensen87} consist of certain sharp and irregular triangles. Also, the result by the marching cubes algorithm contains holes while our method is topology preserving. Therefore, unlike the marching cubes algorithm, no further post-processing is needed in our meshing scheme.

\begin{table}
    \begin{center}
    \begin{tabular}{ |l|C{17mm}|C{17mm}|C{17mm}|C{17mm}| }
    \hline
    Point clouds & Our proposed method & Spherical embedding \cite{Zwicker04} & Marching cubes \cite{Lorensen87} & Tight Cocone \cite{Dey03}\\ \hline
      Soda Can & 0.99 & 0.93 & 0.82 & 0.98 \\ \hline
      Hippocampus & 0.99 & 0.91 & 0.82 & 0.99 \\ \hline
      Max Planck & 0.99 & 0.93 & 0.82 & 0.99 \\ \hline
      Cereal Box & 0.99 & 0.88 & 0.81 & 0.99 \\ \hline
      Spiral & 1.00 & 0.85 & 0.82 & 0.99 \\ \hline
      Brain & 0.99 & 0.82 & 0.82 & 0.99 \\ \hline
      Bulldog & 1.00 & 0.86 & 0.82 & 0.99 \\ \hline
      Chinese Lion & 0.99 & 0.84 & 0.82 & 0.99 \\ \hline
      Bimba & 1.00 & 0.88 & 0.81 & 1.00 \\ \hline
      Igea & 0.97 & 0.90 & 0.83 & 0.97 \\ \hline
      Armadillo & 0.98 & 0.80 & 0.82 & 0.98 \\ \hline
      Lion Vase & 0.99 & 0.85 & 0.83 & 0.99 \\ \hline
    \end{tabular}
    \end{center}
    \bigbreak
    \caption{The Delaunay ratios of different meshing approaches. The ratio assesses the proportion of edges in the resulting triangulations that satisfy the opposite angle sum property $\alpha+\beta \leq \pi$ in a triangulation. A Delaunay ratio exactly equals 1 indicates that the triangulation is Delaunay.}
    \label{table:delaunay}
\end{table}

To quantitatively assess the ``almost-Delaunay'' property of our meshing results induced by the spherical Delaunay triangulations, we recall that a Delaunay triangulation satisfies the opposite angle sum property: For every edge in a Delaunay triangulation, the two angles $\alpha$ and $\beta$ opposite to the edge satisfy
\begin{equation}
 \alpha + \beta \leq \pi.
\end{equation}
For the abovementioned meshing approaches, we define the {\em Delaunay ratio} of the resulting triangulation by
\begin{equation}
 \displaystyle \frac{\text{\# of edges with opposite angles } \alpha, \beta \text{ s.t. } \alpha + \beta \leq \pi}{\text{Total \# of edges in the triangulation}}.
\end{equation}
A higher Delaunay ratio indicates that the triangulation is closer to a perfect Delaunay triangulation. The Delaunay ratios by different meshing algorithms are presented in Table \ref{table:delaunay}. Because of the high accuracy of our spherical conformal parameterizations, our proposed scheme achieves significantly better triangulation results when compared with the spherical embedding algorithm \cite{Zwicker04} and the marching cubes algorithm \cite{Lorensen87}. Also, our results are comparable to or even slightly better than those of the Tight Cocone algorithm \cite{Dey03}. The comparisons demonstrate the advantages of our proposed meshing scheme. A further comparison between our method and the Tight Cocone algorithm \cite{Dey03} is given in the following subsection.

\begin{figure}[t]
 \centering
 \includegraphics[width=0.24\textwidth]{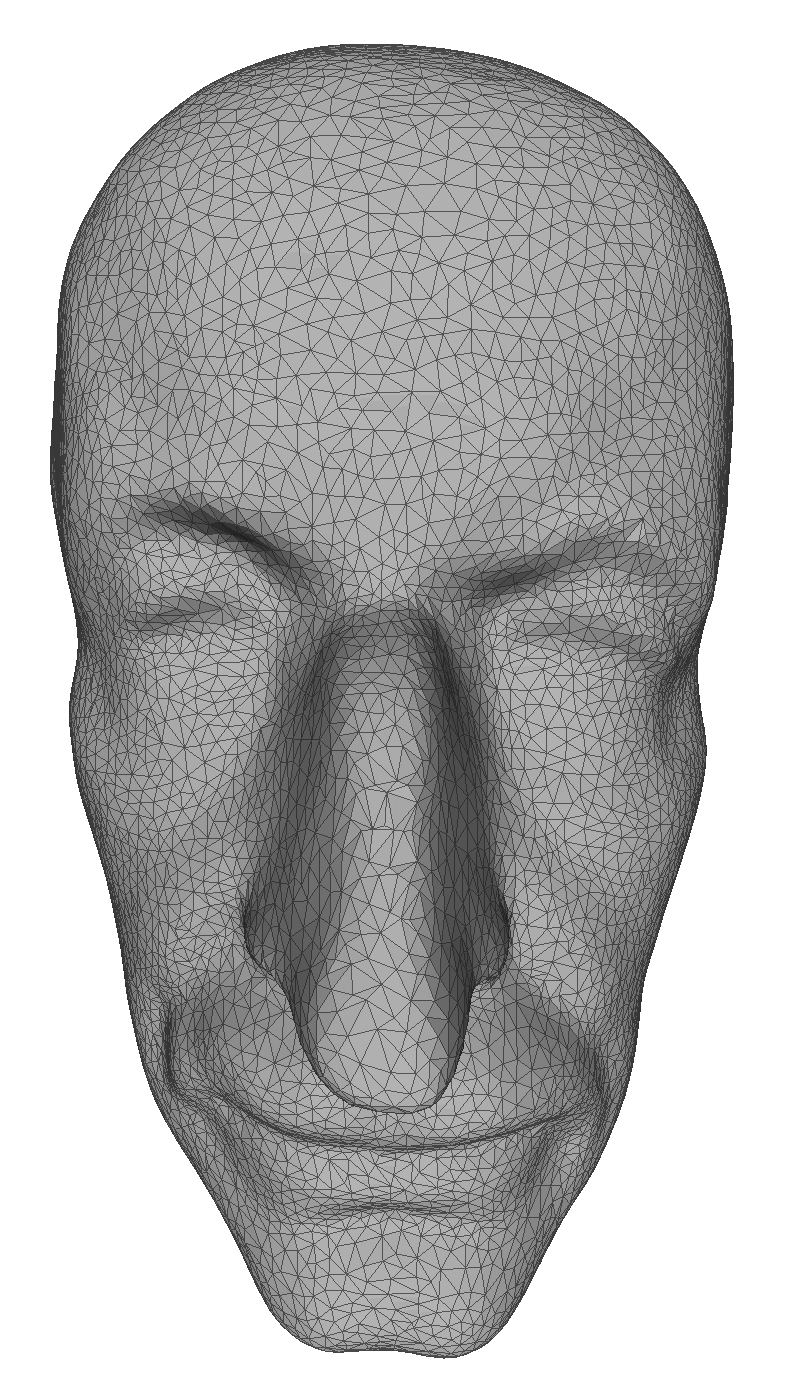}
 \includegraphics[width=0.24\textwidth]{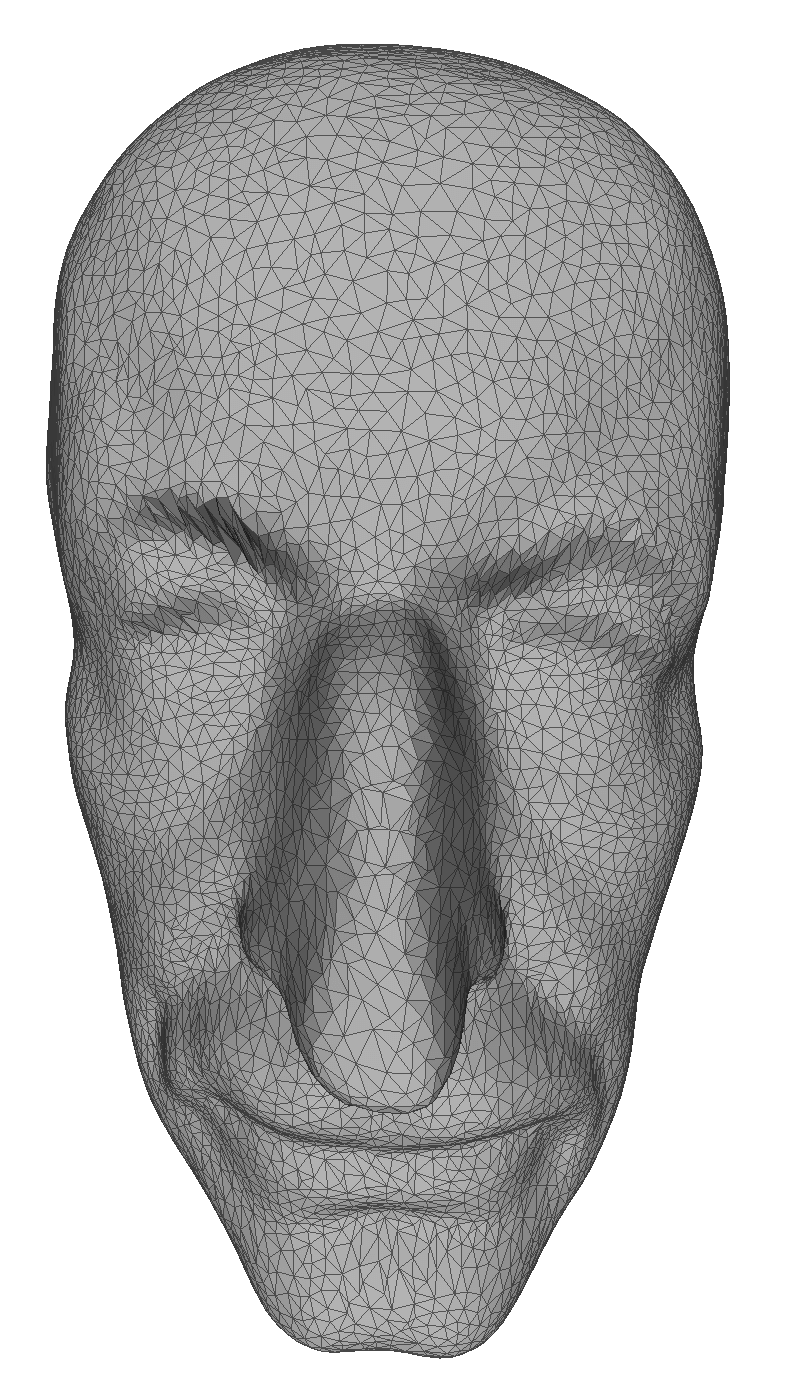}
 \includegraphics[width=0.24\textwidth]{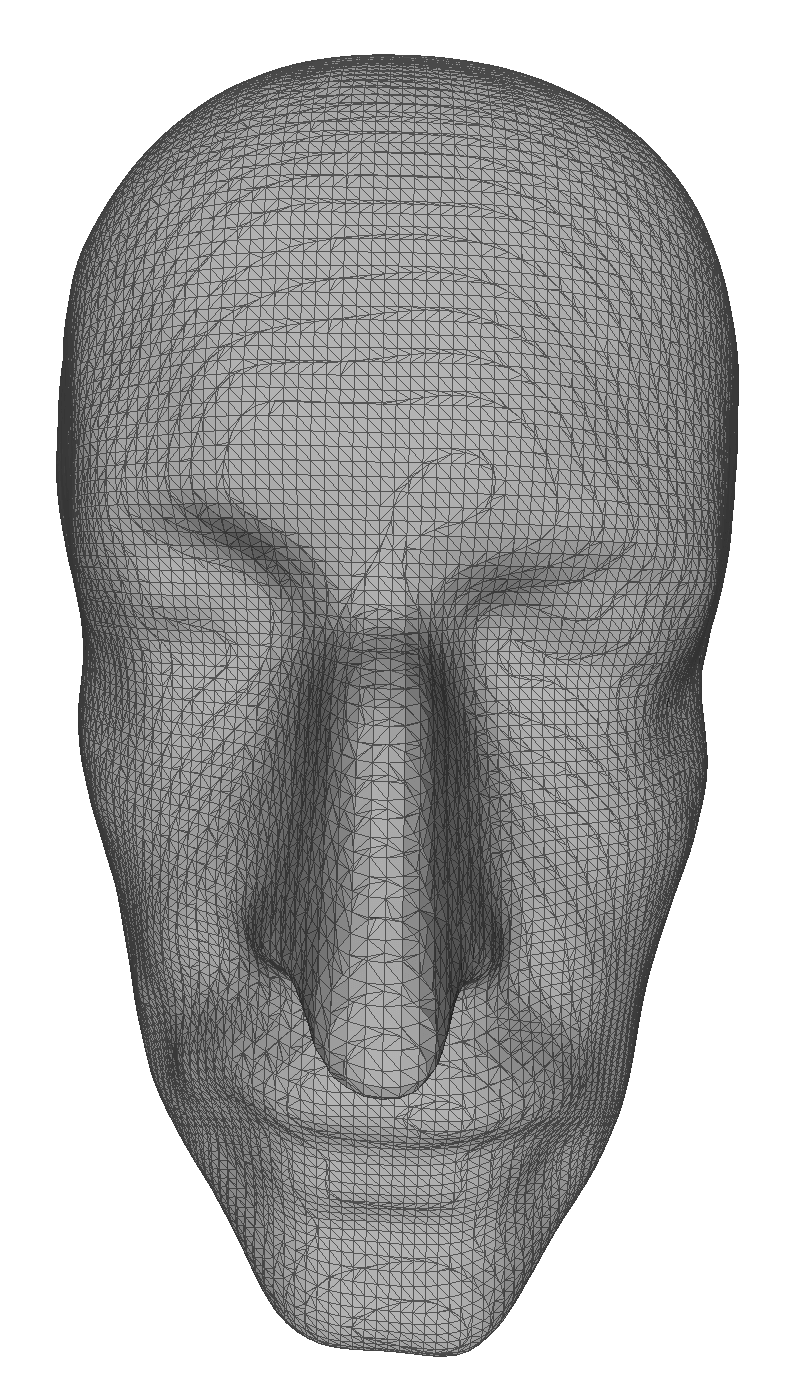}
 \includegraphics[width=0.24\textwidth]{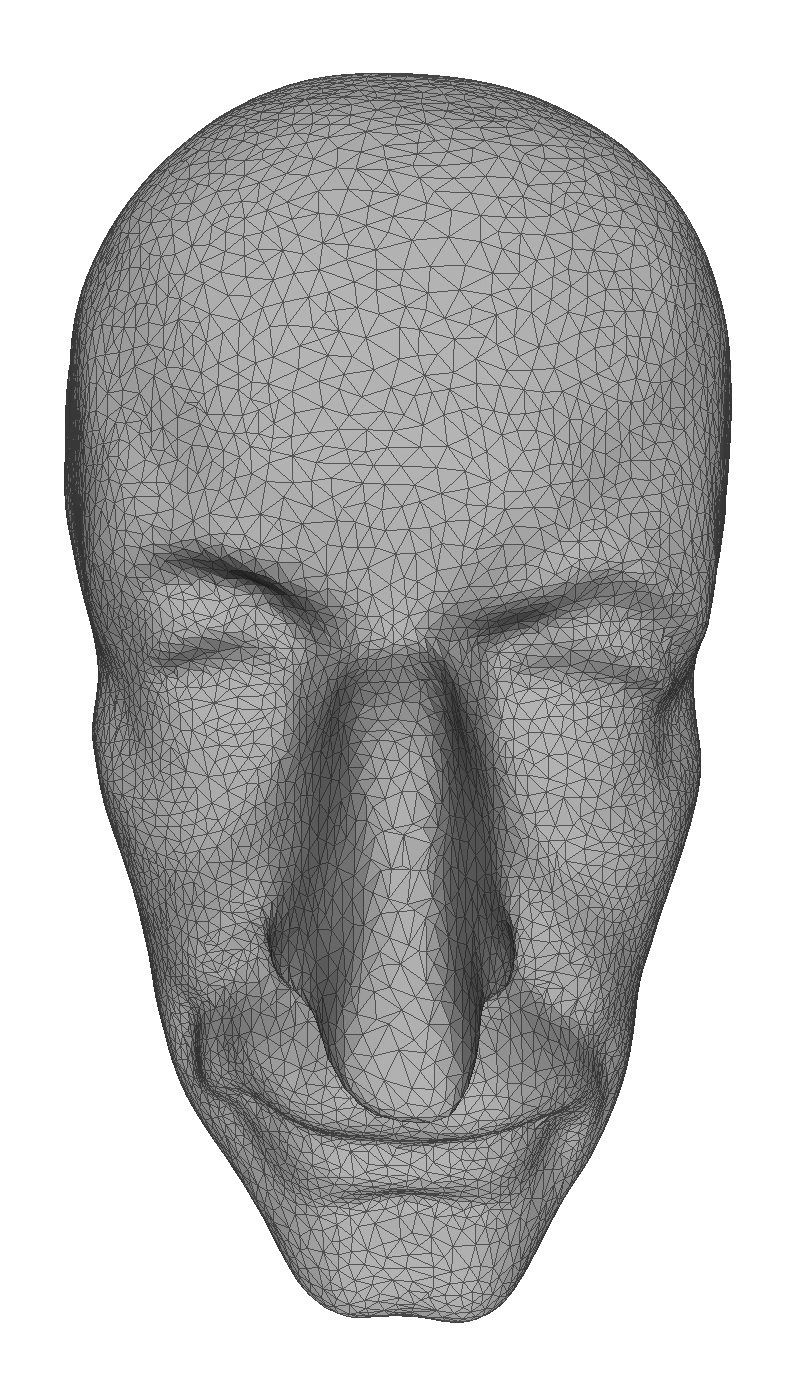}
 \includegraphics[width=0.24\textwidth]{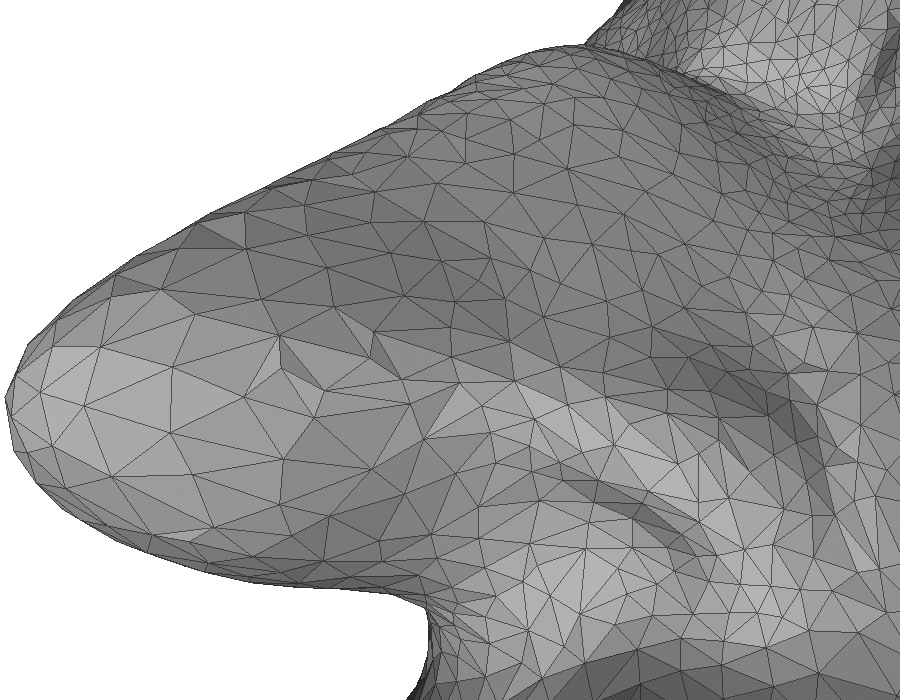}
 \includegraphics[width=0.24\textwidth]{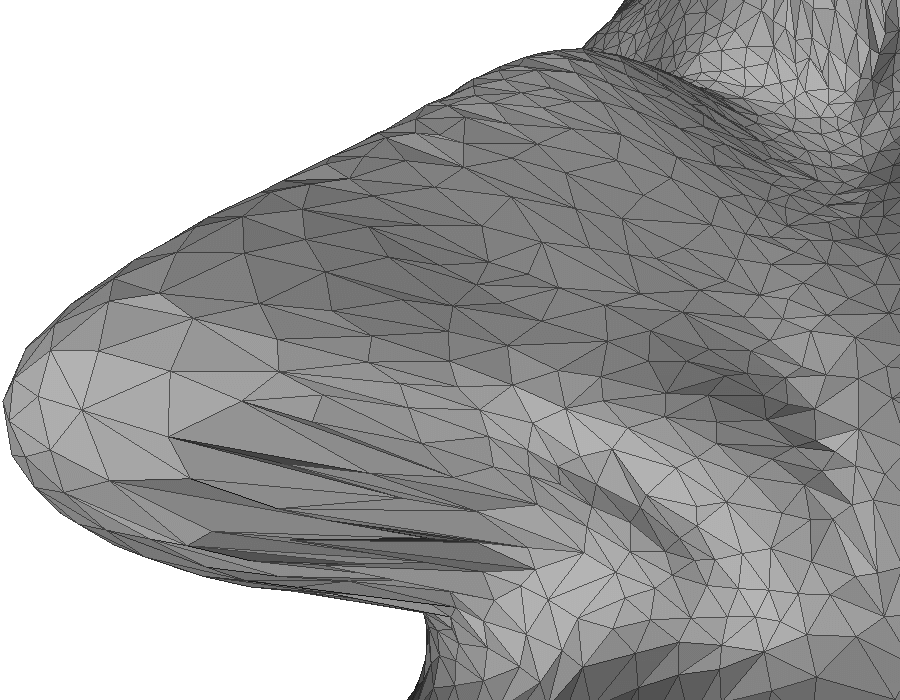}
 \includegraphics[width=0.24\textwidth]{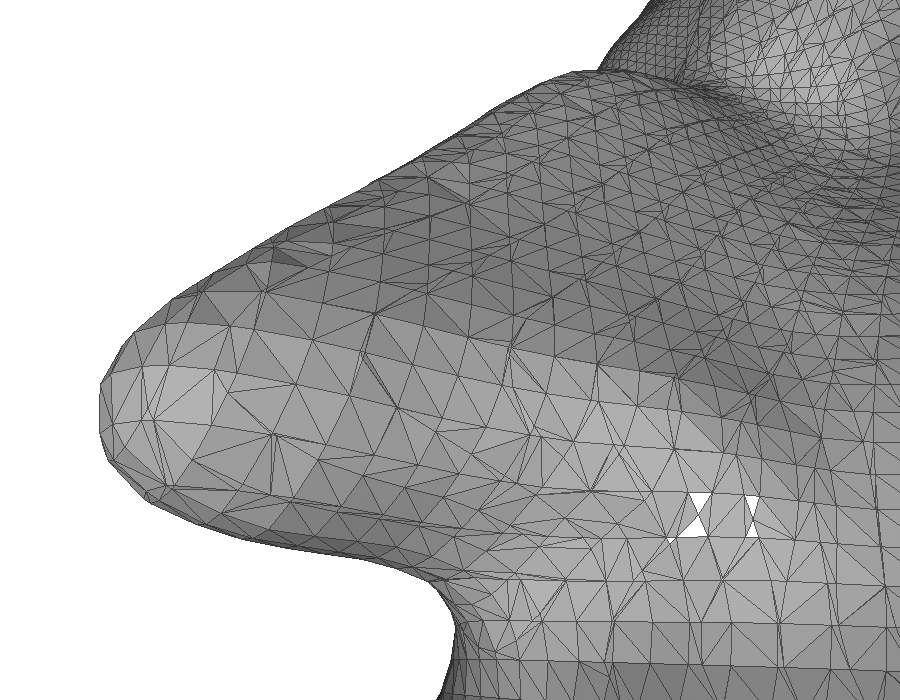}
 \includegraphics[width=0.24\textwidth]{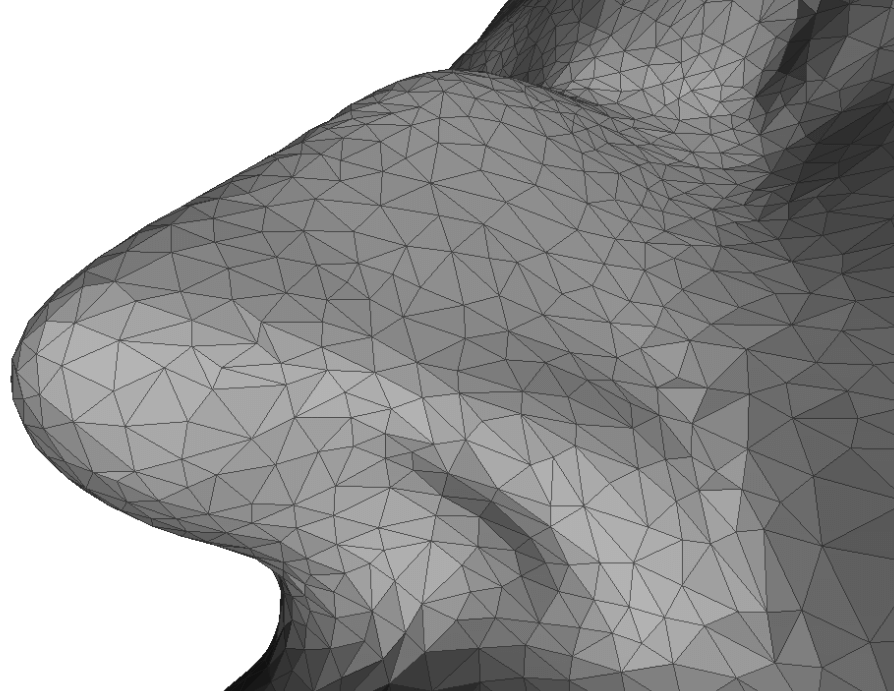}
 \caption{A comparison between our meshing scheme and other approaches. A front view of the triangulated point cloud and a zoom-in of the nose are shown for each method. Left to right: Our meshing result, the method in \cite{Zwicker04}, the marching cubes algorithm \cite{Lorensen87} and the Tight Cocone algorithm \cite{Dey03}.}
 \label{fig:meshing_results}
\end{figure}

\begin{figure}[t]
 \centering
 \includegraphics[width=0.32\textwidth]{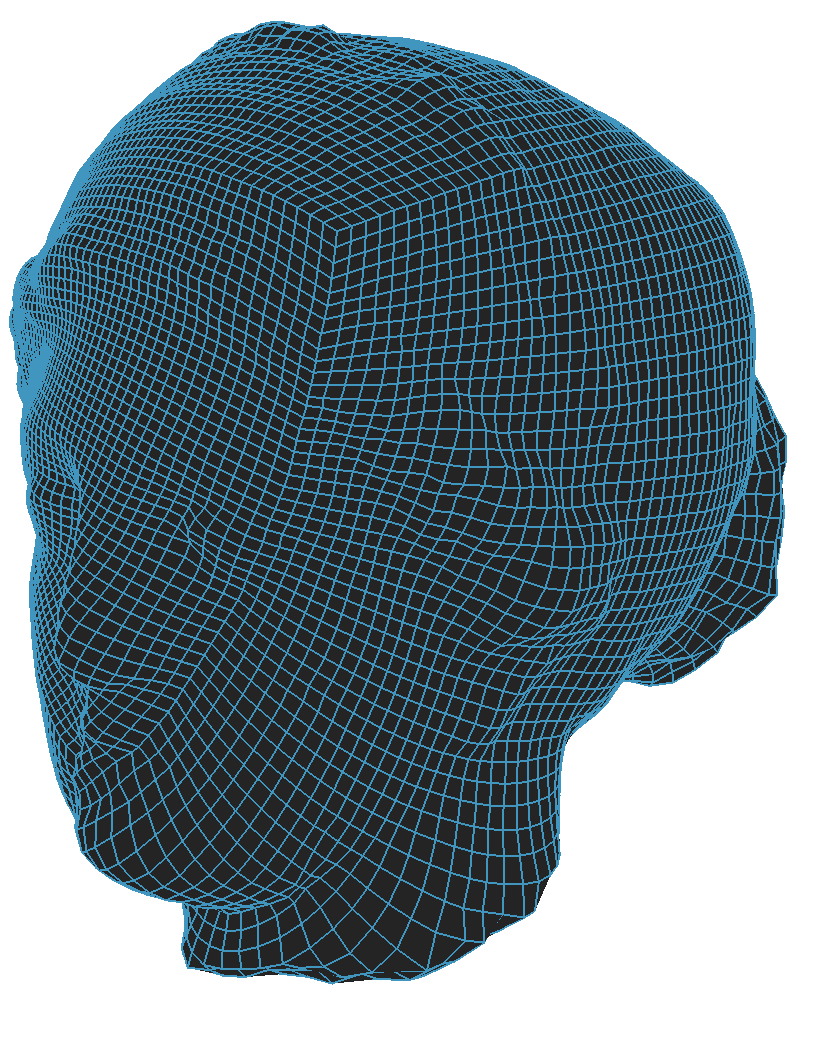}
 \includegraphics[width=0.32\textwidth]{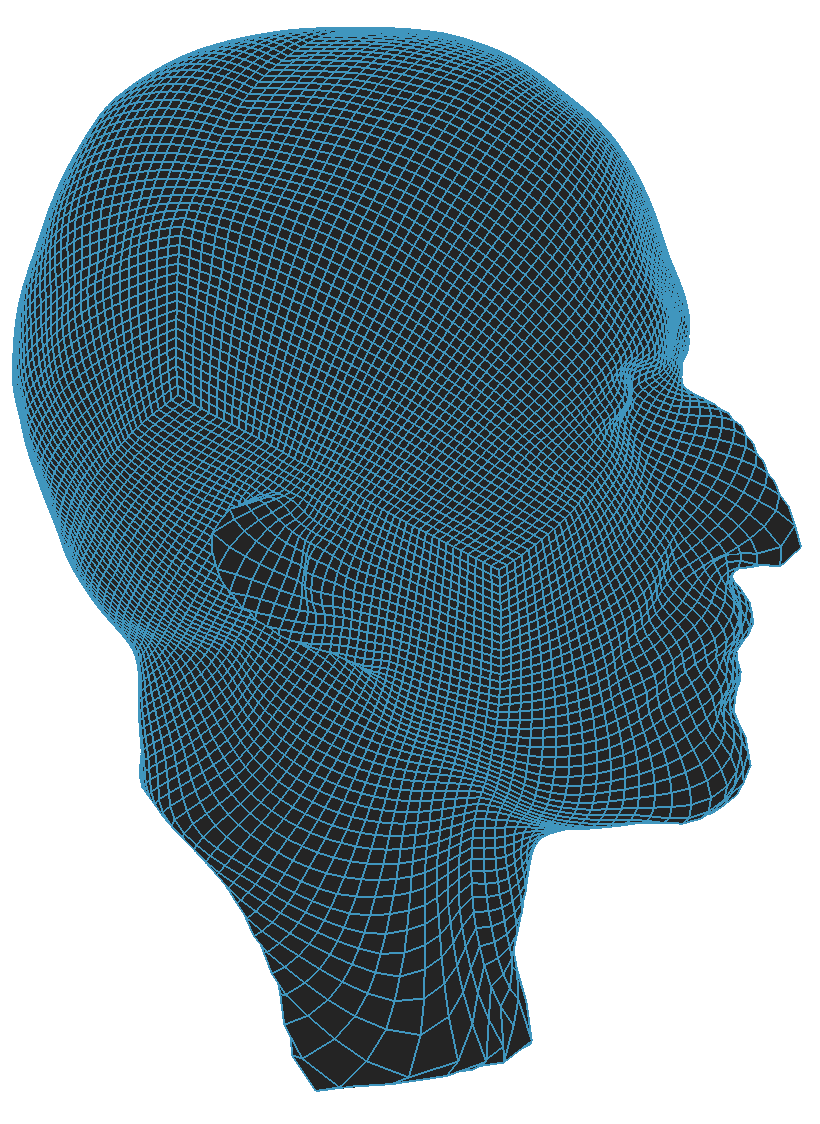}
 \caption{Quad mesh generation on point clouds using our proposed method.}
 \label{fig:quad_result}
\end{figure}
In addition, we can generate quadrangulations of point clouds with the aid of the spherical conformal parameterization. Two examples of the quad meshes generated by our method are given in Figure \ref{fig:quad_result}. To create quad meshes of point clouds, we make use of a standard spherical quad mesh and our spherical conformal parameterization results. With the aid of the spherical conformal parameterizations, we can interpolate the standard quad mesh onto the input point clouds and thus generate quad mesh representations. Because of the conformality of our parameterization scheme, the resulting quad meshes are with high quality. Also, the meshes are guaranteed to be topology preserving.

\begin{figure}[t]
 \centering
 \includegraphics[width=0.4\textwidth]{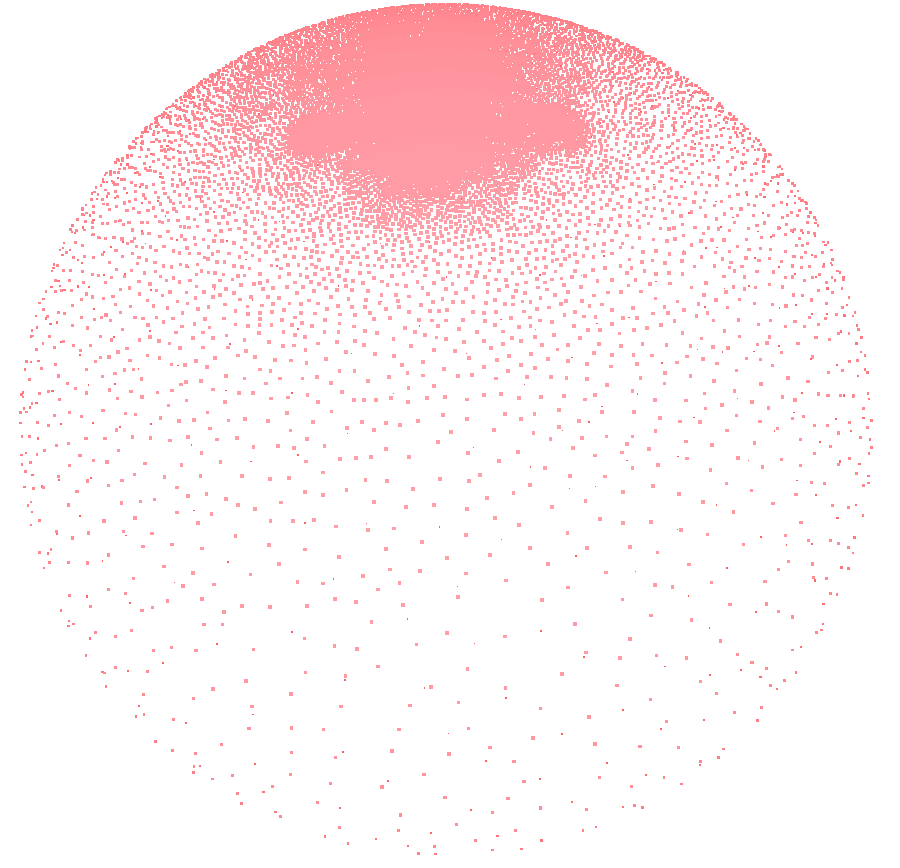}
 \includegraphics[width=0.25\textwidth]{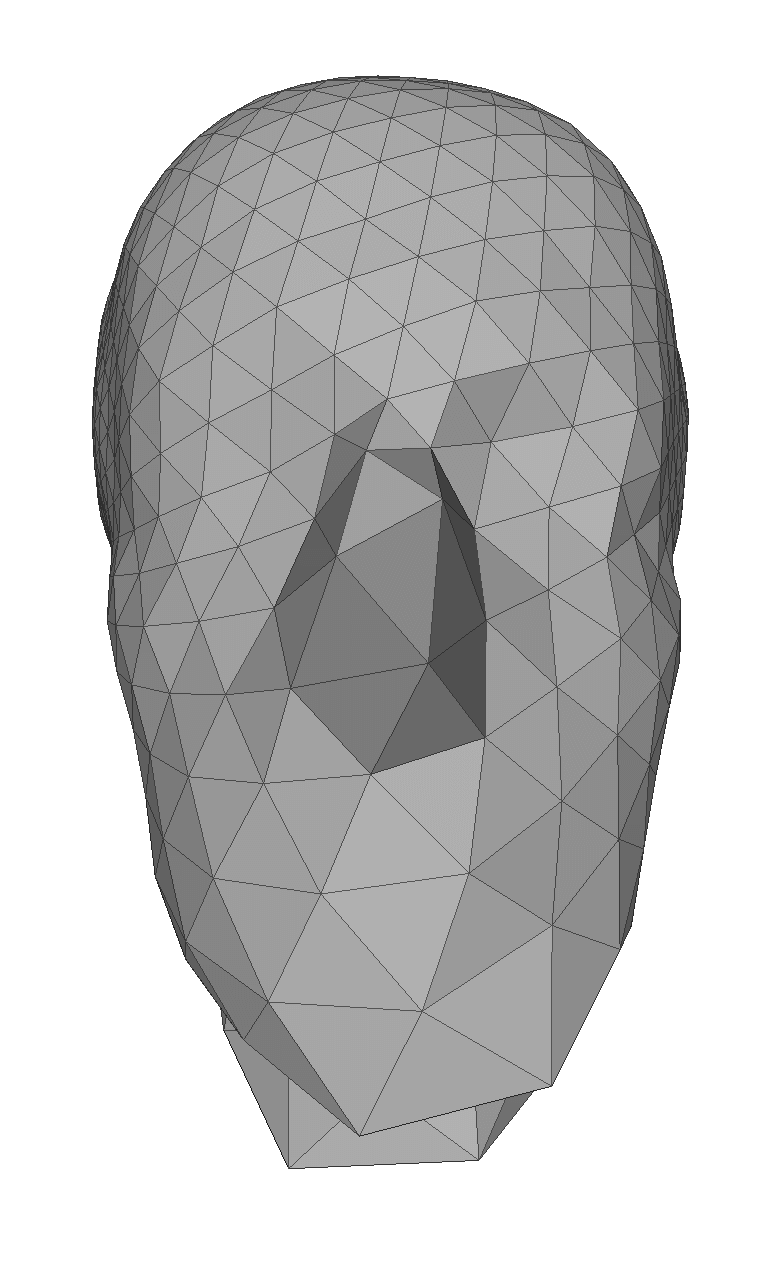}
 \includegraphics[width=0.25\textwidth]{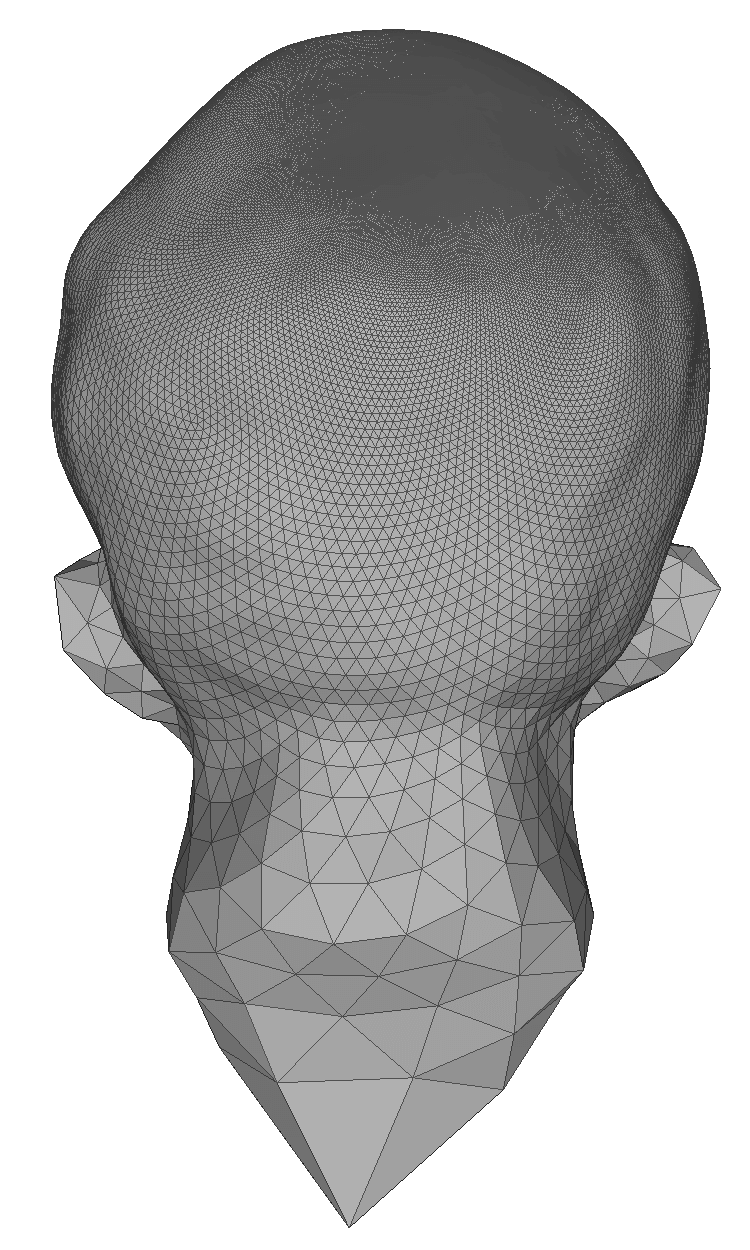}
 \includegraphics[width=0.4\textwidth]{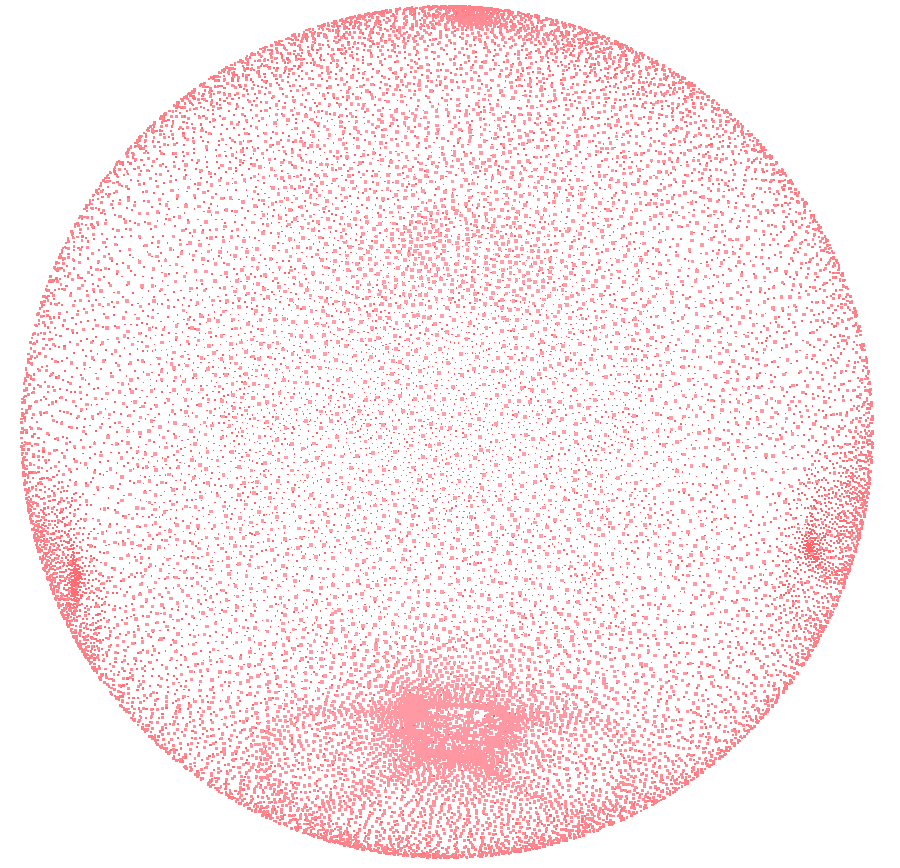}
 \includegraphics[width=0.25\textwidth]{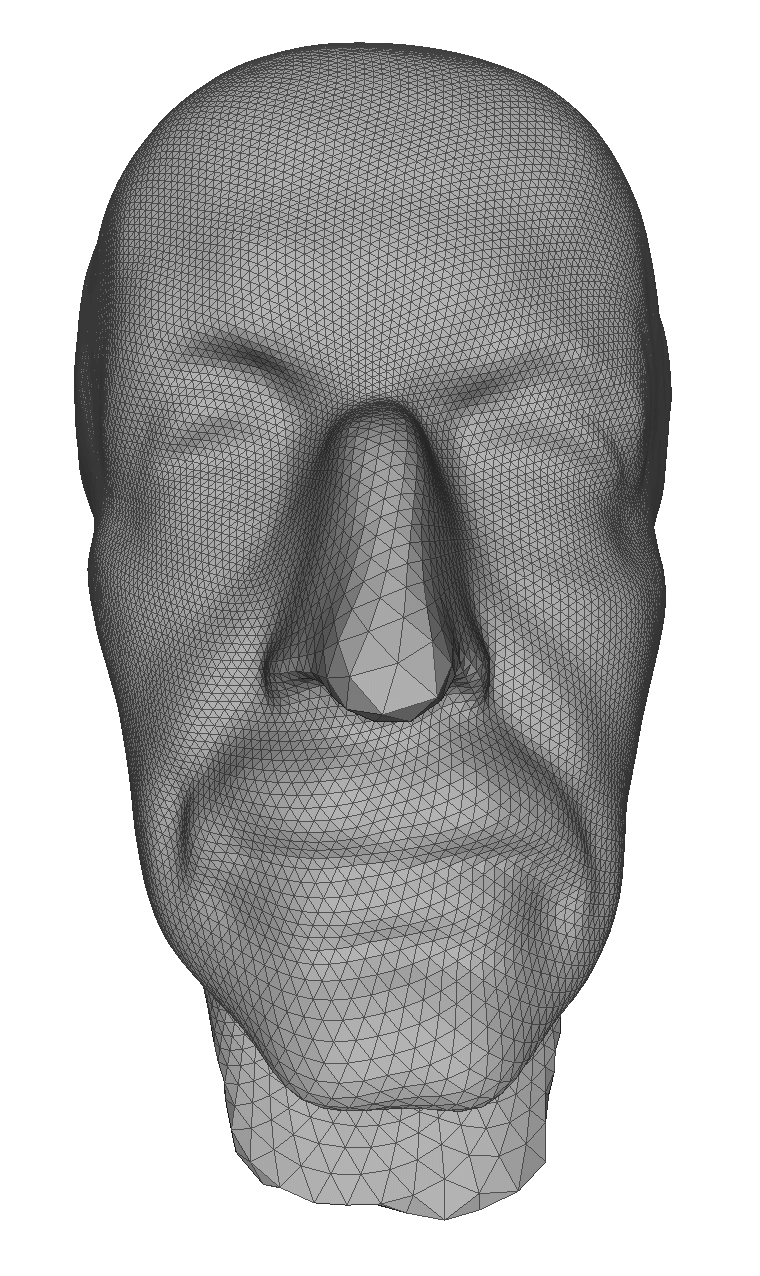}
 \includegraphics[width=0.25\textwidth]{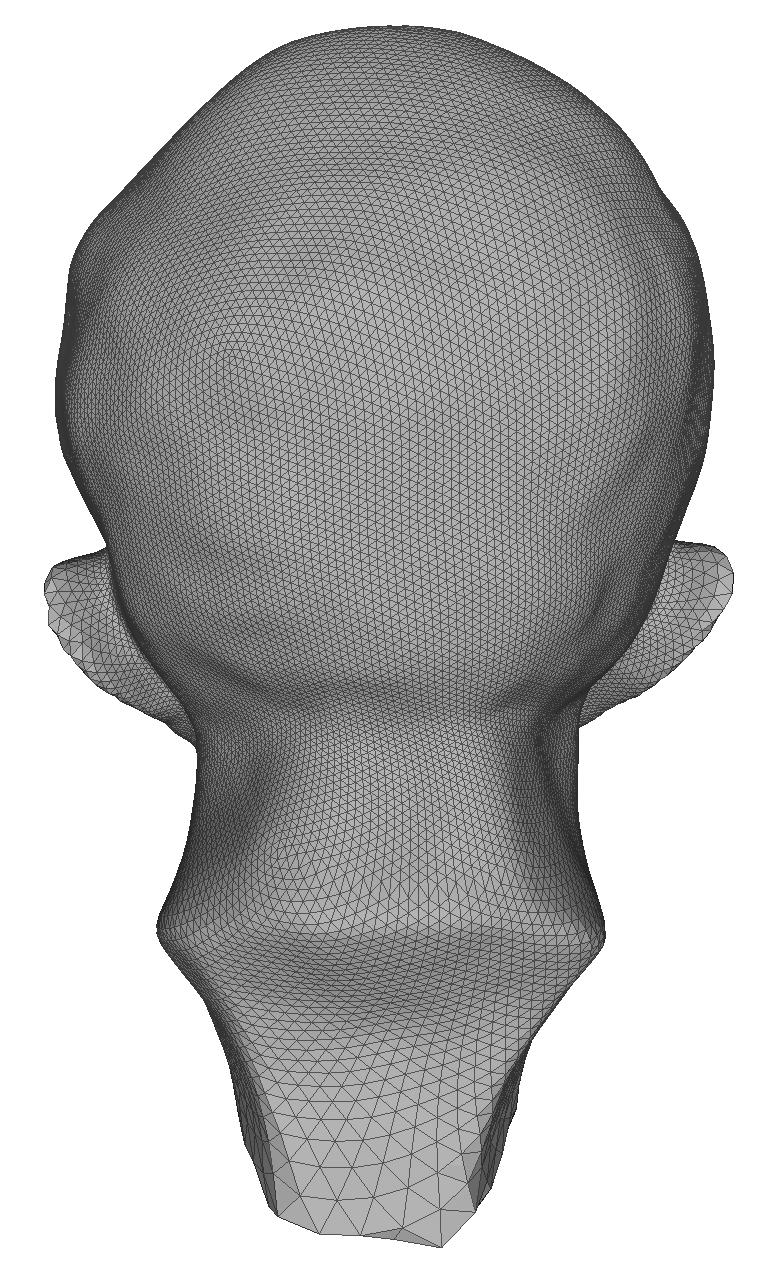}
 \caption{The effect of our balancing scheme on meshing a genus-0 point cloud of Max Planck. Top left: A spherical conformal parameterization without the balancing scheme. Bottom left: A spherical conformal parameterization with the balancing scheme. Middle: The front view of the the meshing results by interpolation with the aid of the parameterizations. Right: The back view.}
 \label{fig:distribution}
\end{figure}

Before ending this subsection, we demonstrate the significance of our proposed balancing scheme in the spherical conformal parameterization. The redistribution is vital for the meshing quality. Figure \ref{fig:distribution} shows the meshing results with and without the redistribution scheme. It can be easily observed that if the spherical parameterization of a genus-0 point cloud is unbalanced, then on the mesh generated by interpolation with the aid of the spherical conformal parameterization, most of the vertices will be concentrated at one small region of the mesh. As a result, most features of the underlying surface are lost. In contrast, with our proposed balancing scheme, a high quality mesh can be effectively generated. The above results reflect the importance of our balancing scheme in the point cloud parameterizations for meshing.

\subsection{Stability under geometrical and topological noises}
\begin{figure}[t]
 \centering
 \includegraphics[width=0.32\textwidth]{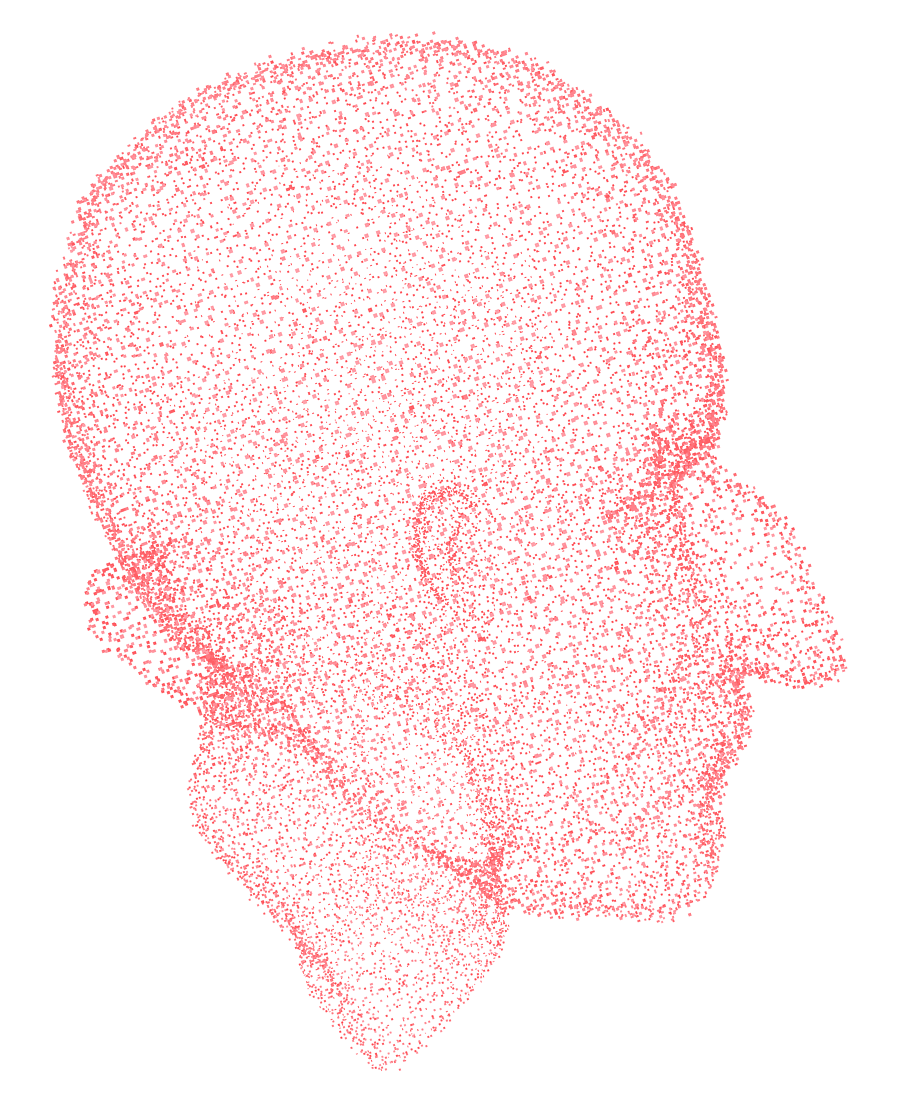}
 \includegraphics[width=0.32\textwidth]{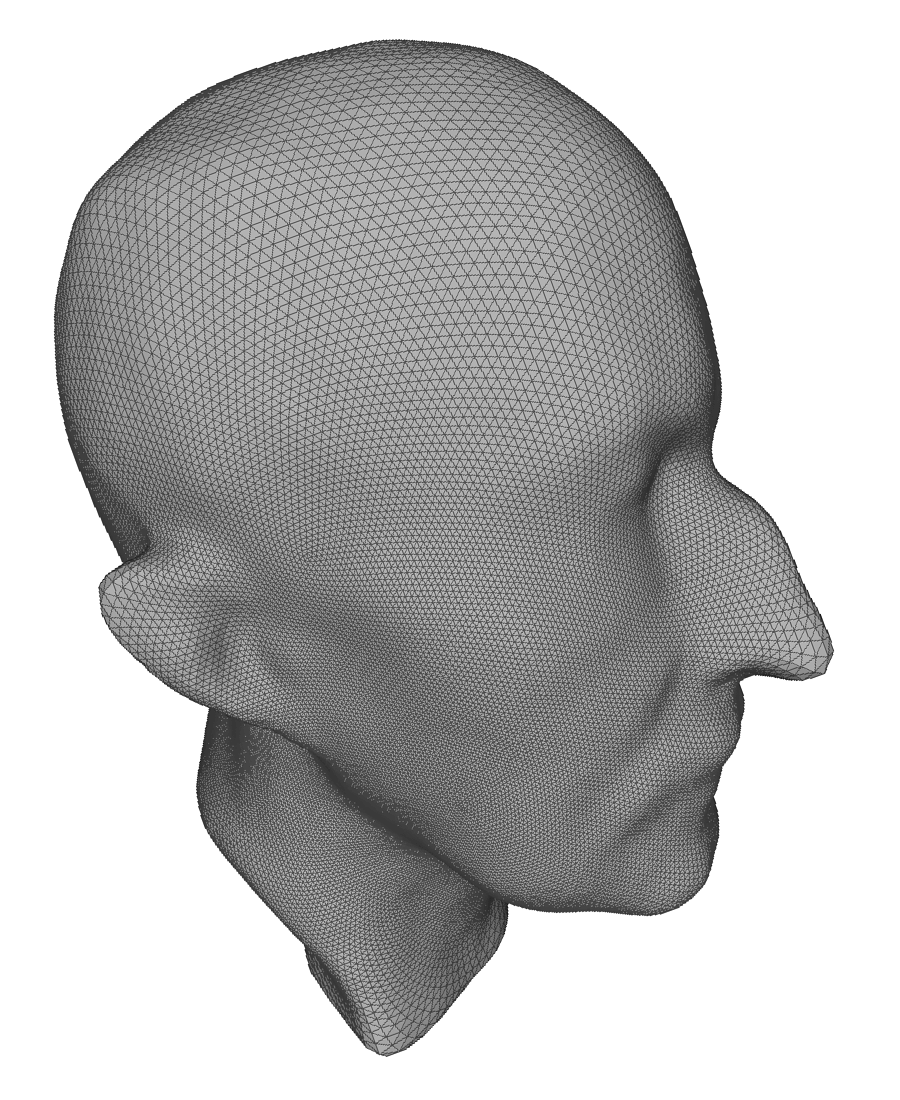}
 \caption{Meshing a geometrically noisy point cloud.}
 \label{fig:geometric_noise}
\end{figure}

\begin{figure}[t]
 \centering
 \includegraphics[width=0.19\textwidth]{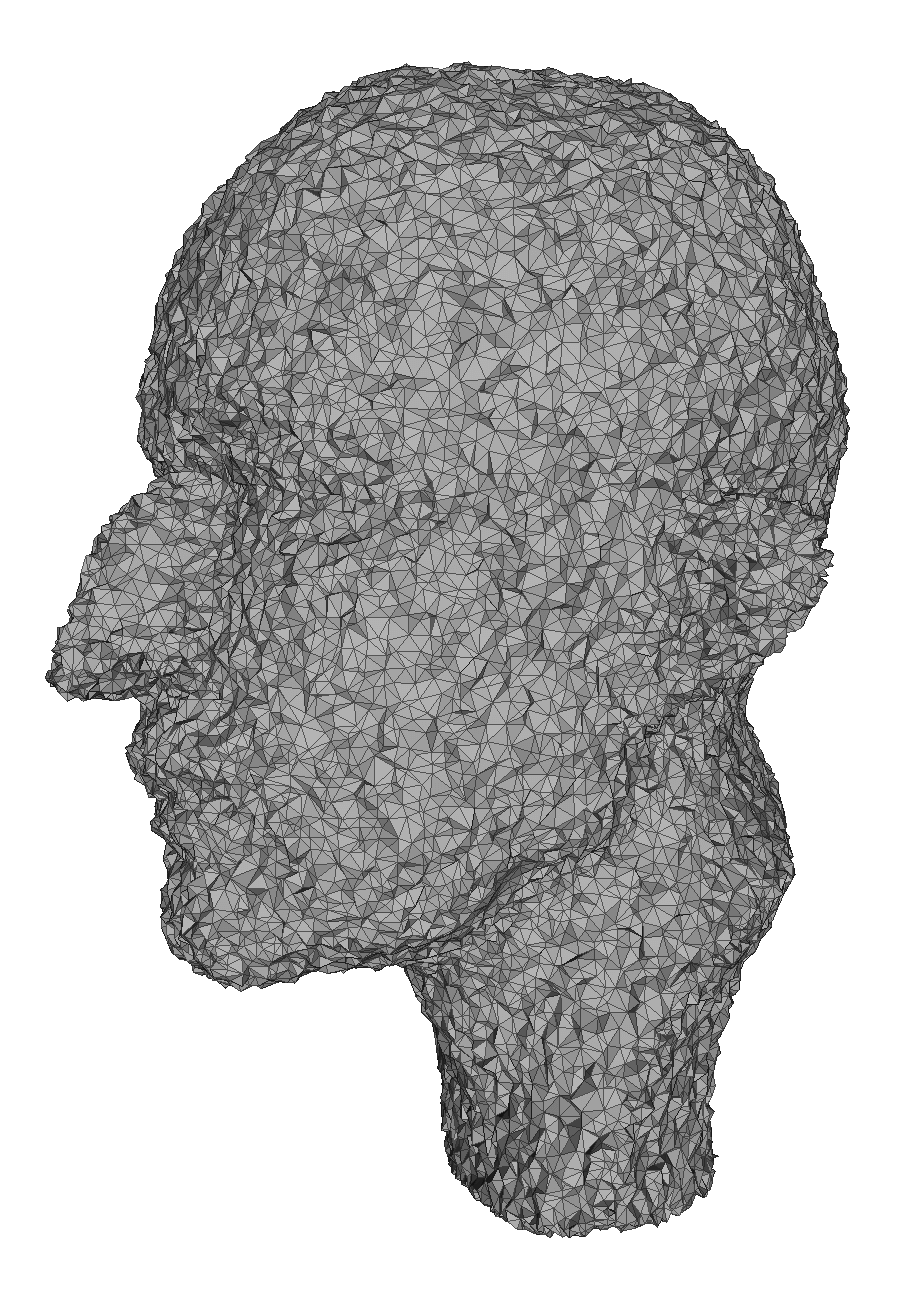}
 \includegraphics[width=0.29\textwidth]{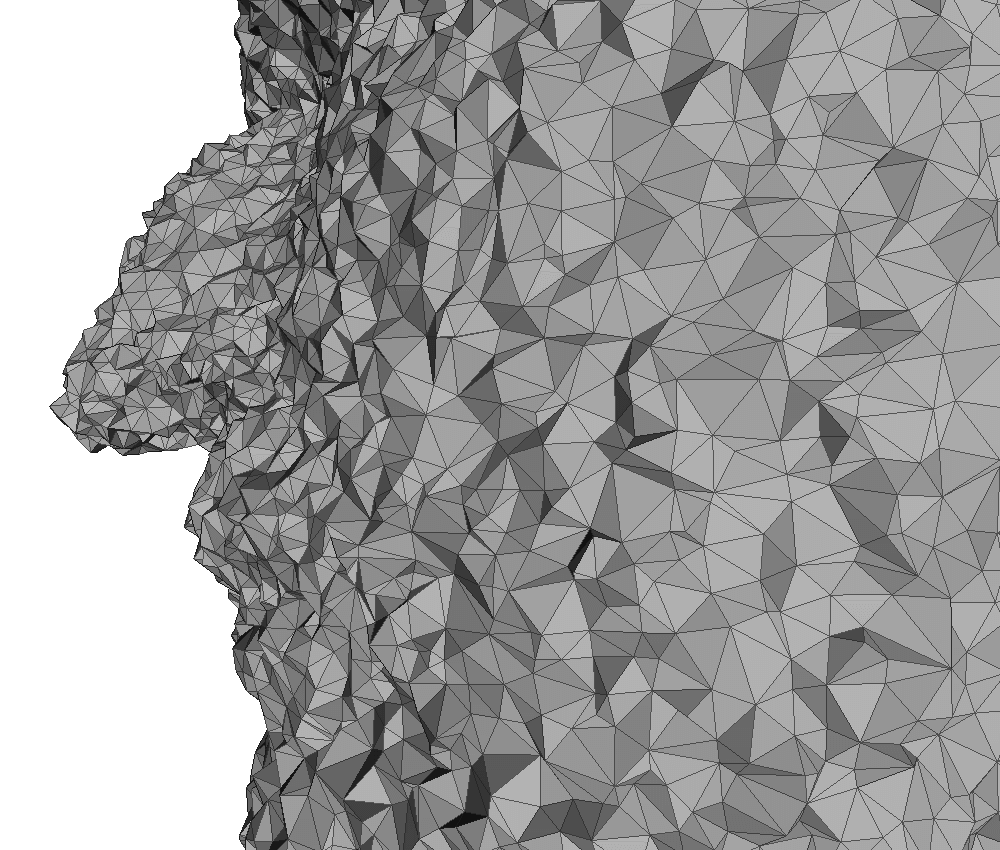}
 \includegraphics[width=0.19\textwidth]{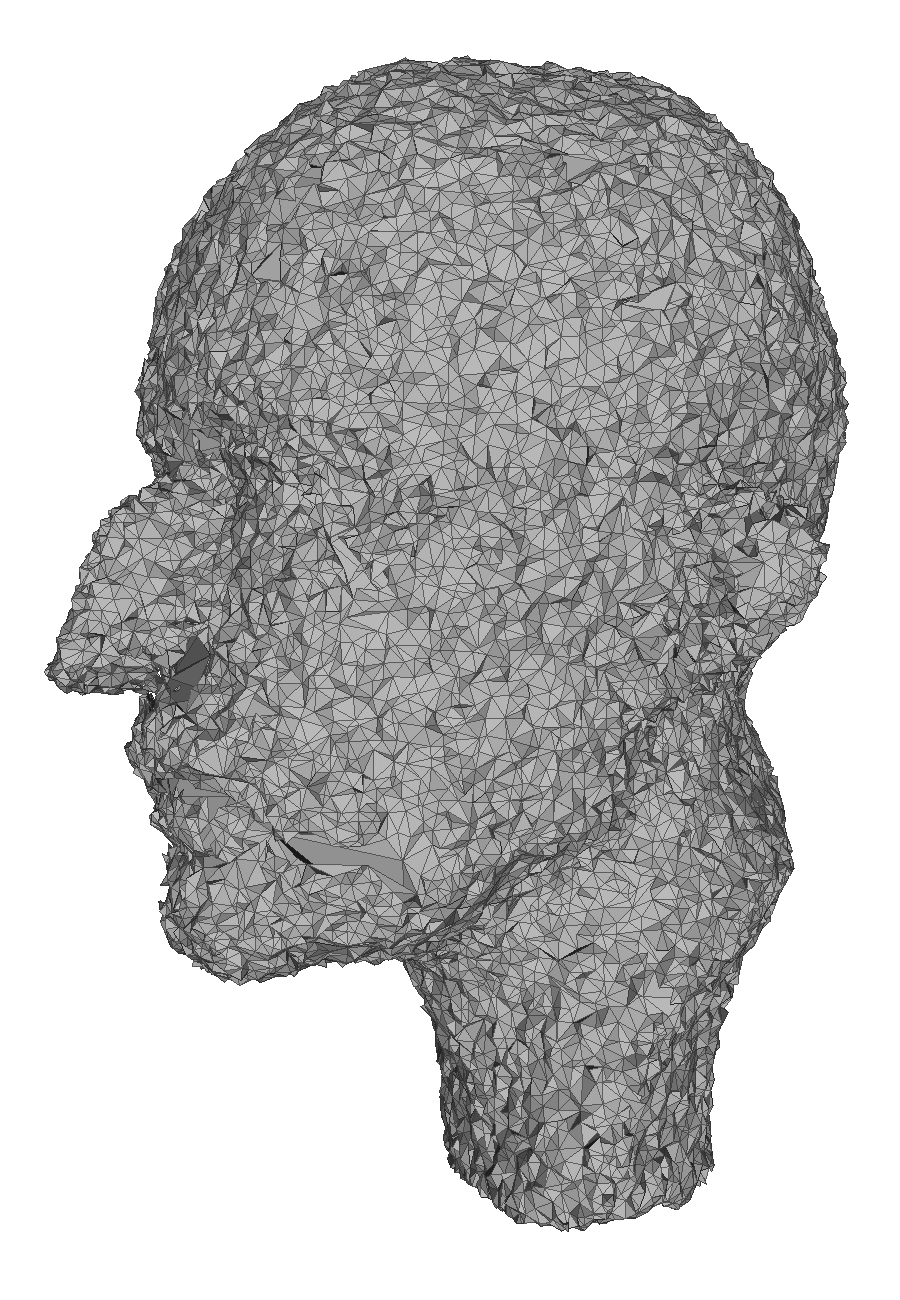}
 \includegraphics[width=0.29\textwidth]{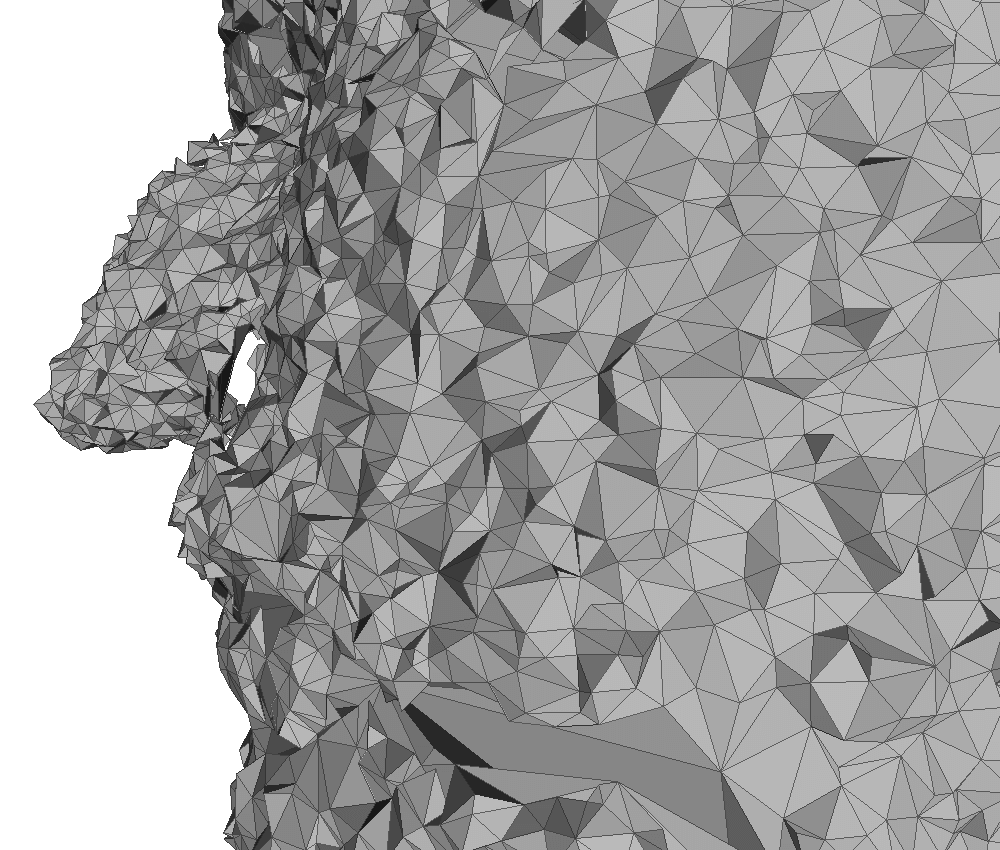}
 \caption{Comparison of our meshing scheme and the Tight Cocone algorithm \cite{Dey03} on a geometrically noisy point cloud. All points are considered in the computations. Left: Our meshing result with a zoom-in of the nose. Right: The result of the Tight Cocone algorithm with a zoom-in of the nose.}
 \label{fig:geometric_noise2}
\end{figure}

\begin{figure}[t]
 \centering
 \includegraphics[width=0.32\textwidth]{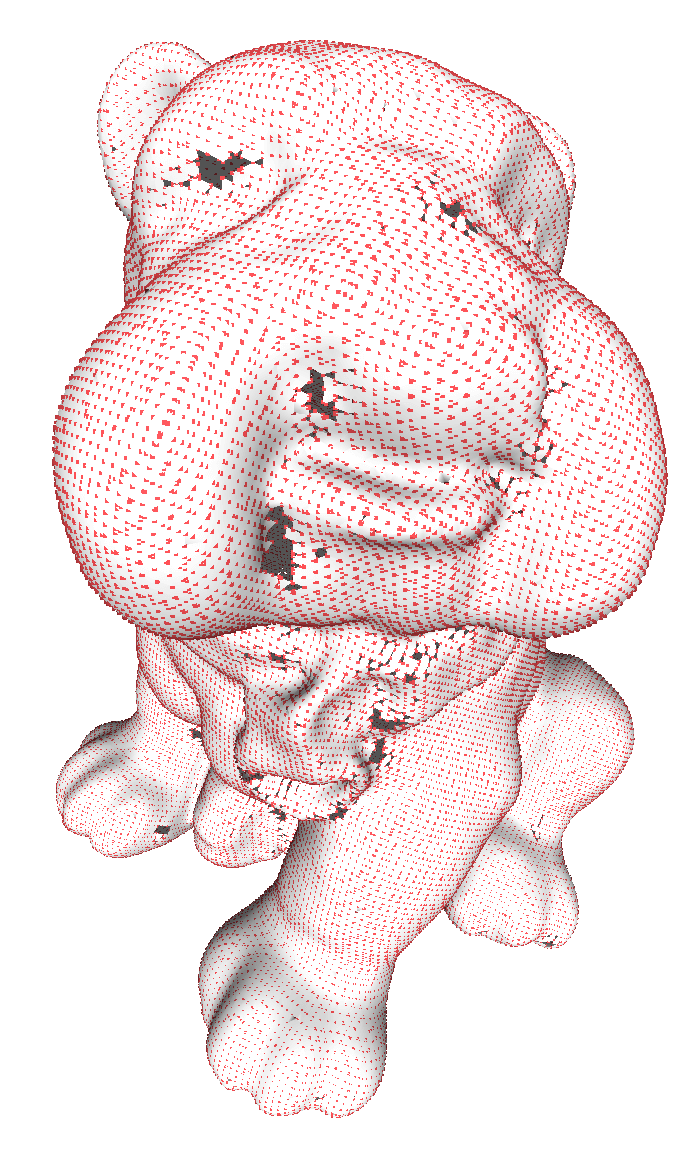}
 \includegraphics[width=0.32\textwidth]{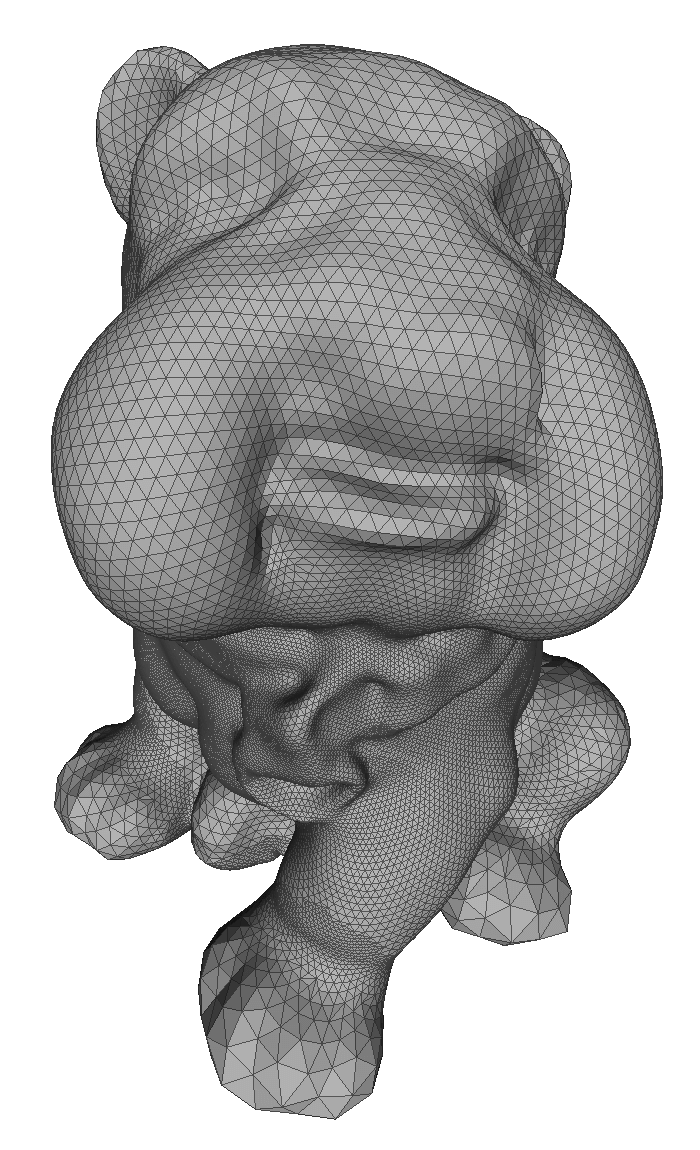}
 \caption{Meshing a topologically noisy point cloud with unwanted holes.}
 \label{fig:topological_noise}
\end{figure}

\begin{figure}[t]
 \centering
 \includegraphics[width=0.27\textwidth]{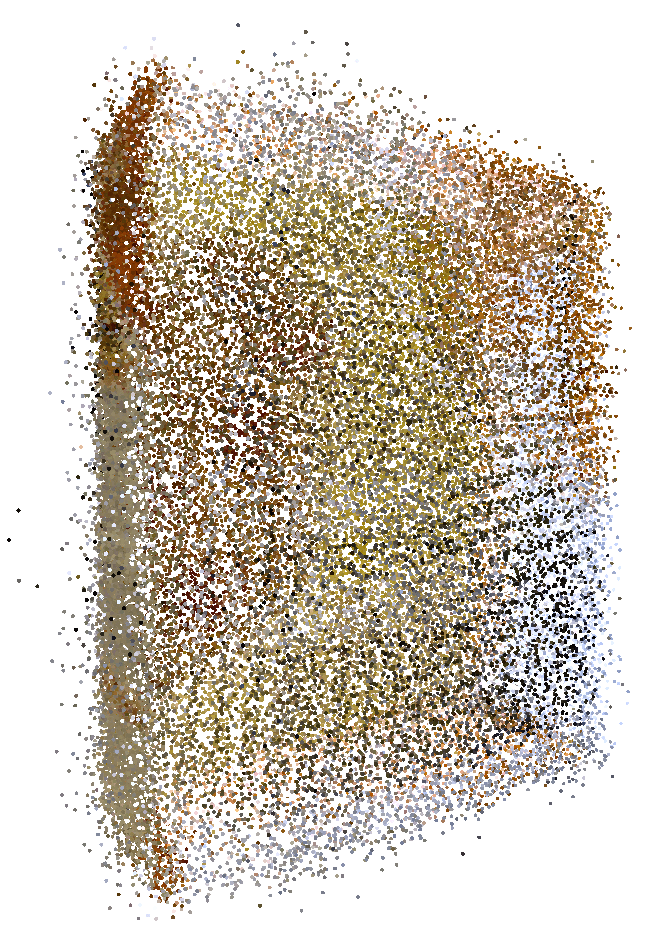}
 \includegraphics[width=0.27\textwidth]{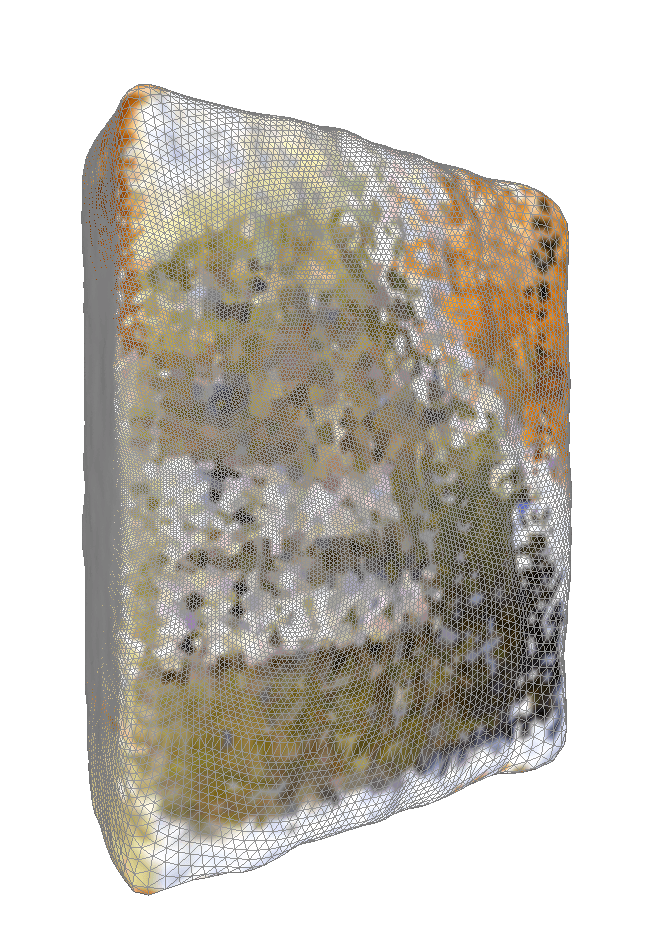}
 \includegraphics[width=0.27\textwidth]{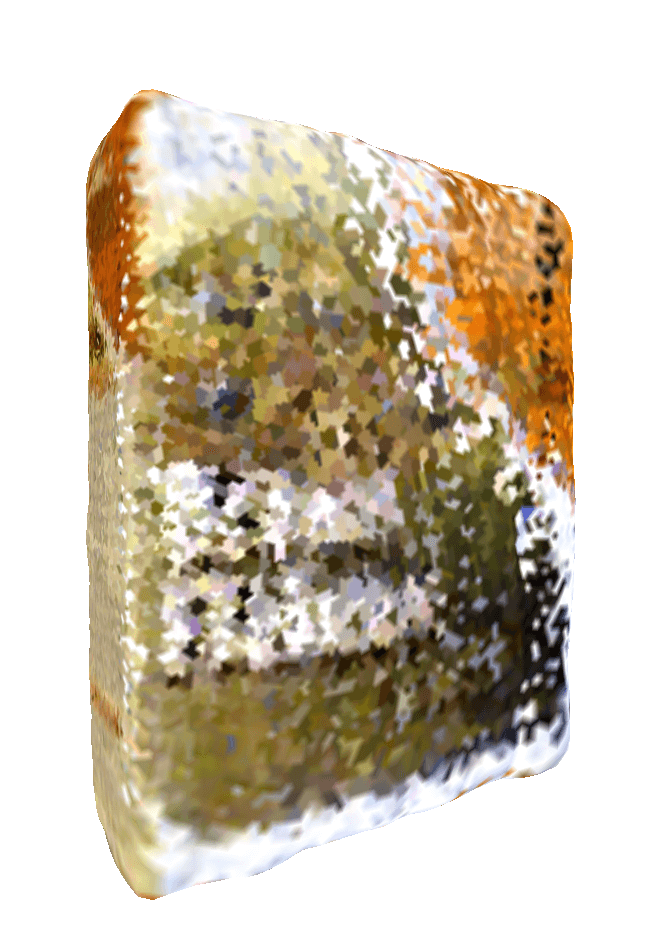}
 \includegraphics[width=0.27\textwidth]{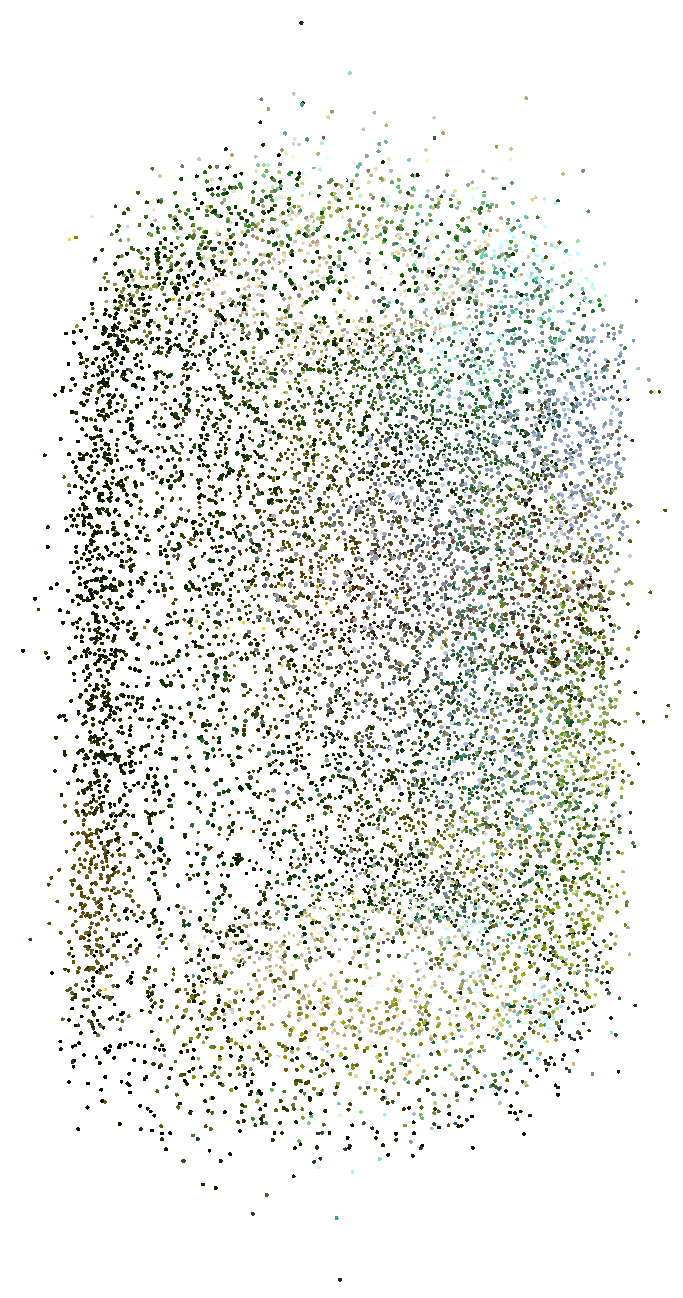}
 \includegraphics[width=0.27\textwidth]{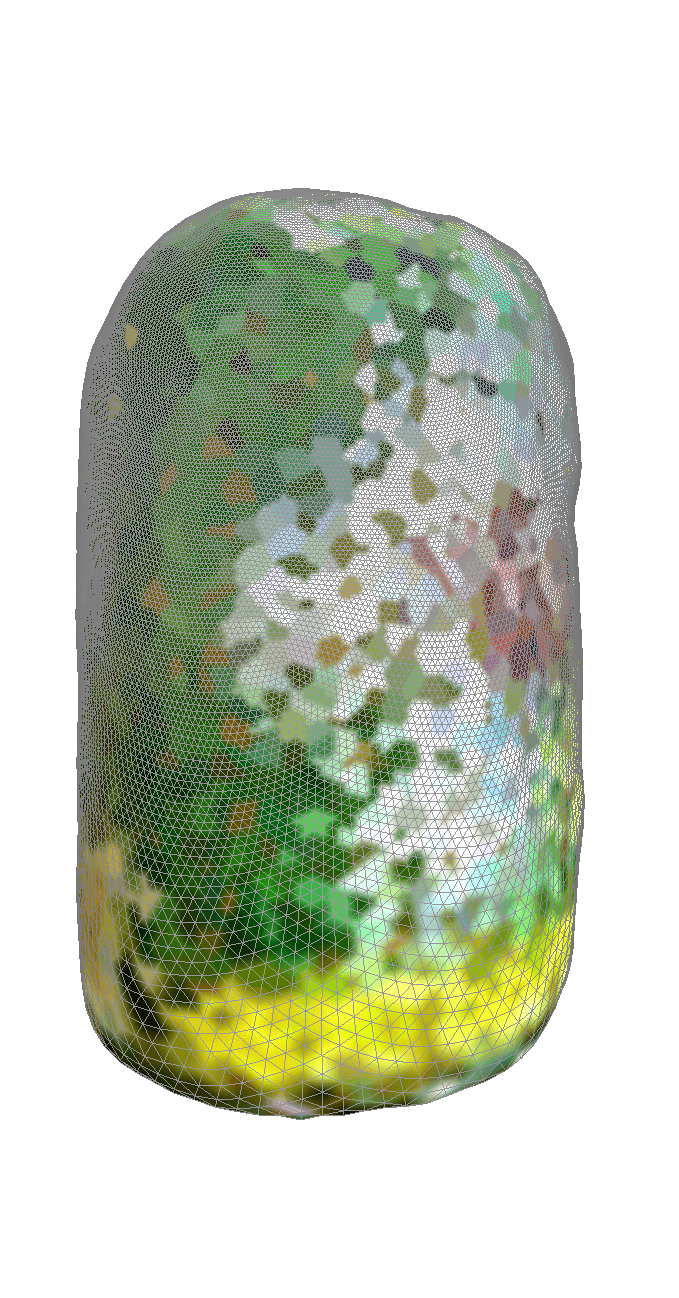}
 \includegraphics[width=0.27\textwidth]{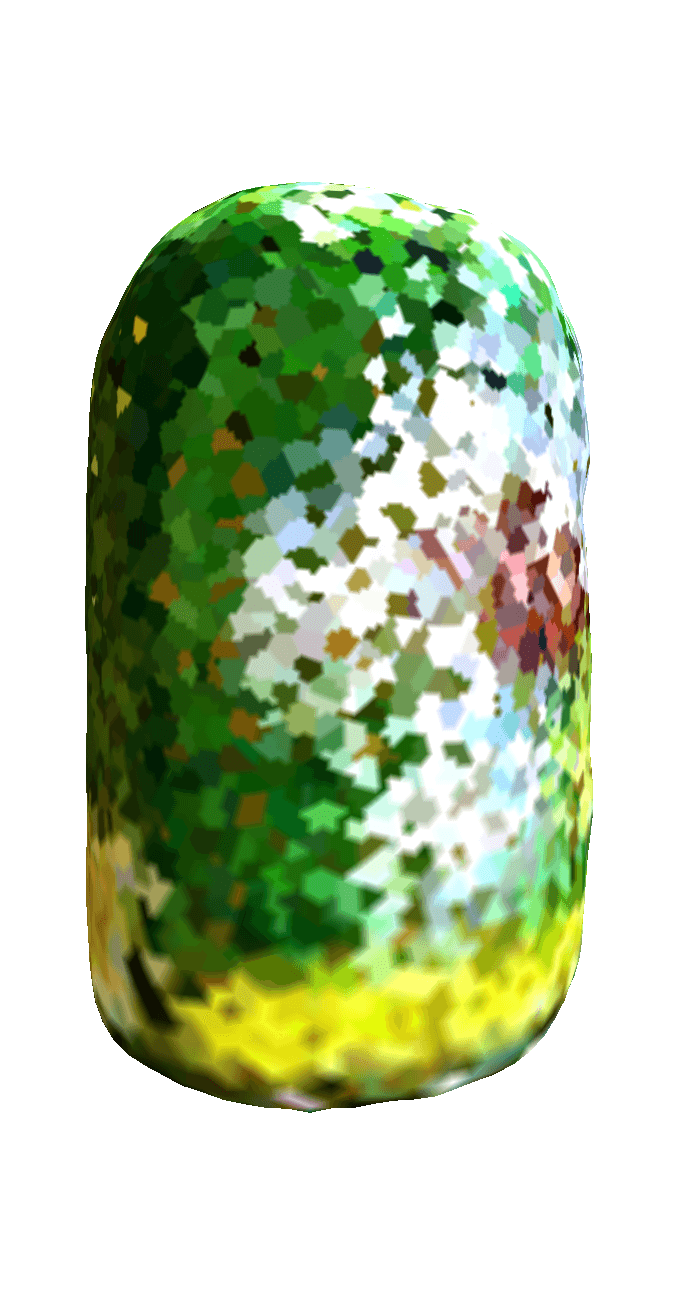}
 \caption{Meshing two real 3D scanned noisy point cloud data of a cereal box and a soda can in the RGB-D Scenes Dataset v.2 \cite{Lai14b}. Left: The raw point cloud data. Middle: The meshing results. Right: The meshing results without showing the triangulations.}
 \label{fig:real_pc}
\end{figure}

Our meshing framework is stable under geometrical and topological noises of the input genus-0 point clouds. In some situations, the point clouds obtained by 3D cameras are geometrically noisy. To compute triangulations which represent the underlying surfaces, we can first apply a Poisson filtering on the noisy point clouds. Then, with the aid of our spherical conformal parameterization, we can obtain high quality triangulations on a uniform spherical point cloud and interpolate them back onto the filtered point clouds to produce meshed surfaces. We first demonstrate the effectiveness of our algorithm by two synthetic examples.

Figure \ref{fig:geometric_noise} shows a synthetic point cloud with 3\% uniformly distributed random noise and our meshing result. We can also construct a faithful triangulated mesh on the geometrically noisy point cloud without any filtering or sampling procedure. Figure \ref{fig:geometric_noise2} shows the triangulation result of our meshing scheme and the Tight Cocone algorithm \cite{Dey03} on the noisy point cloud in Figure \ref{fig:geometric_noise}. All points of the point cloud are considered and fixed in the construction of the triangulation. It can be observed that there are irregular triangulations and topological holes on the result by \cite{Dey03}, while our meshing scheme guarantees a regular and topology preserving triangulation even for noisy input point clouds.

Besides, it is common that the sampling processes result in non-uniformly sampled point clouds. In particular, there may be large holes on certain parts of the point clouds sampled from genus-0 objects, which create topological ambiguities and hinder mesh generations. Our parameterization and meshing scheme produce satisfactory results with these topological noises. Moreover, the meshes generated are guaranteed to be genus-0 closed meshes. Figure \ref{fig:topological_noise} shows a synthetic point cloud with 1021 randomly created topological holes. It can be observed that our algorithm produces a satisfactory meshing result.

Then, we apply our algorithm for real 3D scanned noisy point cloud data. Several raw point cloud data are adapted from the RGB-D Scenes Dataset v.2 \cite{Lai14b}. Figure \ref{fig:real_pc} shows two point clouds of a soda can and a cereal box, and our meshing results. The above experiments demonstrate the stability and robustness of our proposed method for noisy point clouds.

\subsection{Multilevel representations of genus-0 point clouds}

\begin{figure}[t]
 \centering
 \includegraphics[width=0.27\textwidth]{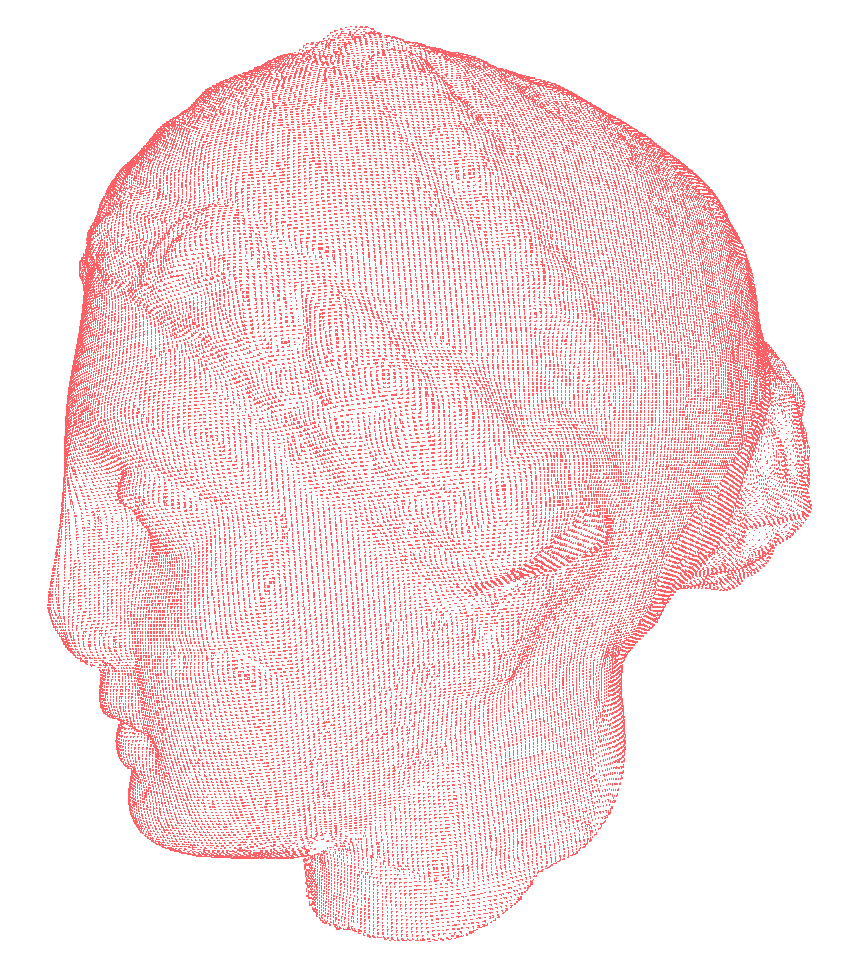}
 \includegraphics[width=0.27\textwidth]{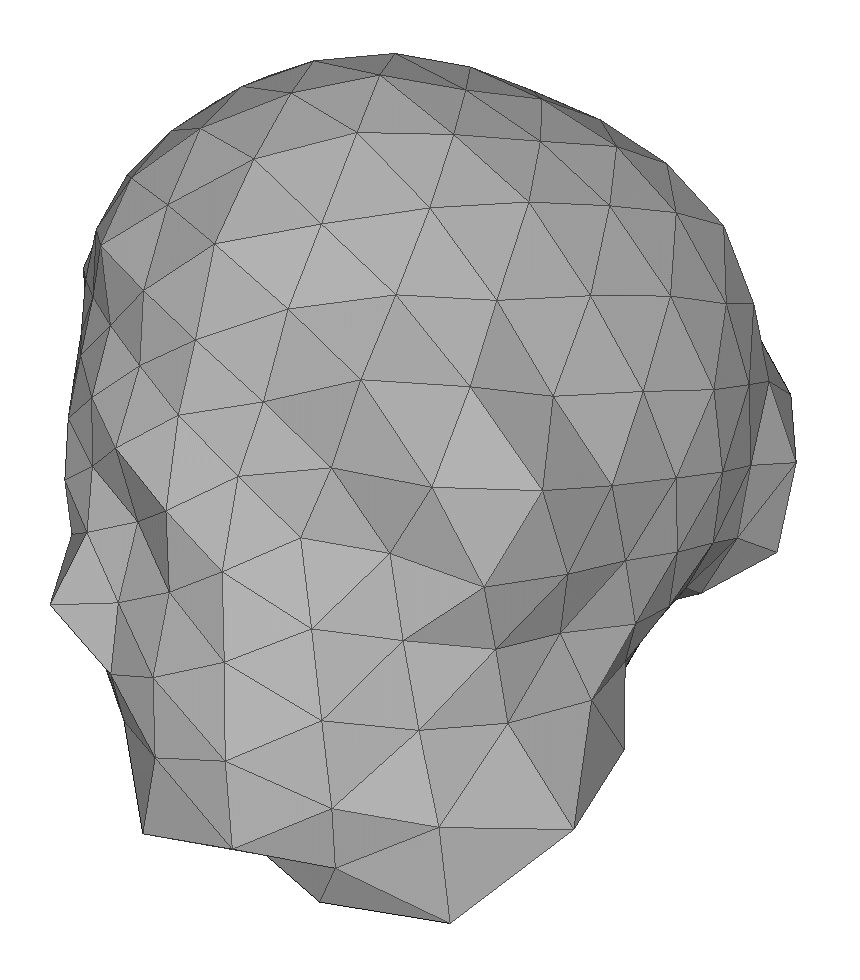}
 \includegraphics[width=0.27\textwidth]{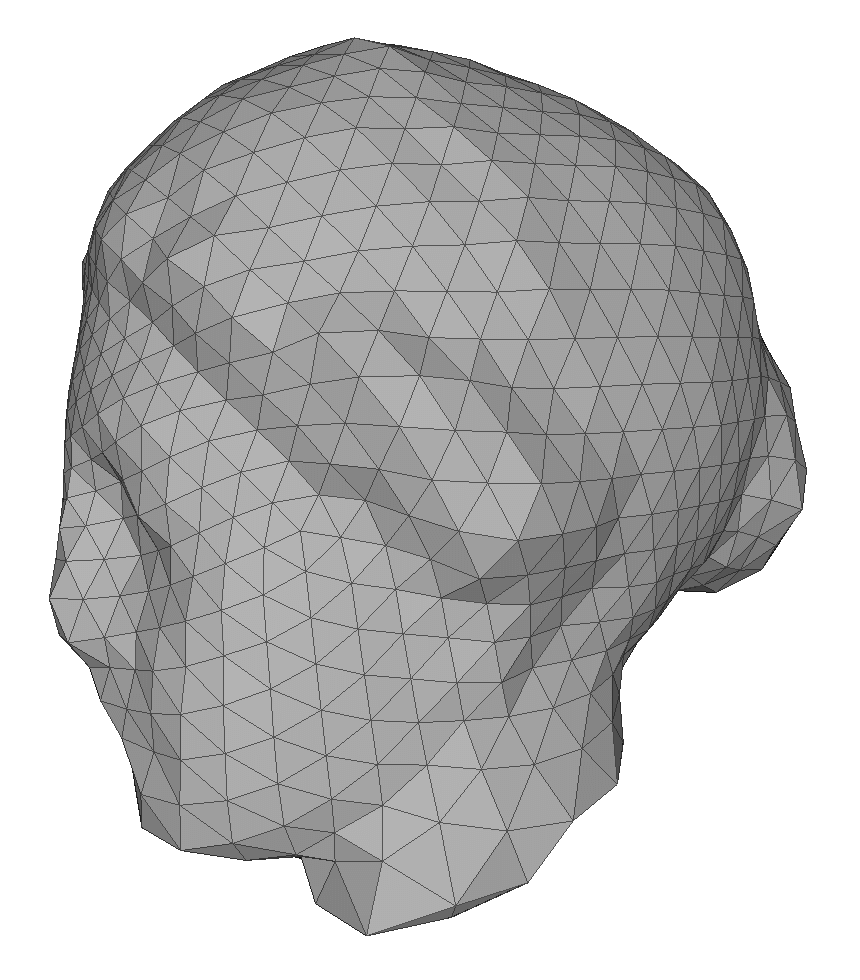}
 \includegraphics[width=0.27\textwidth]{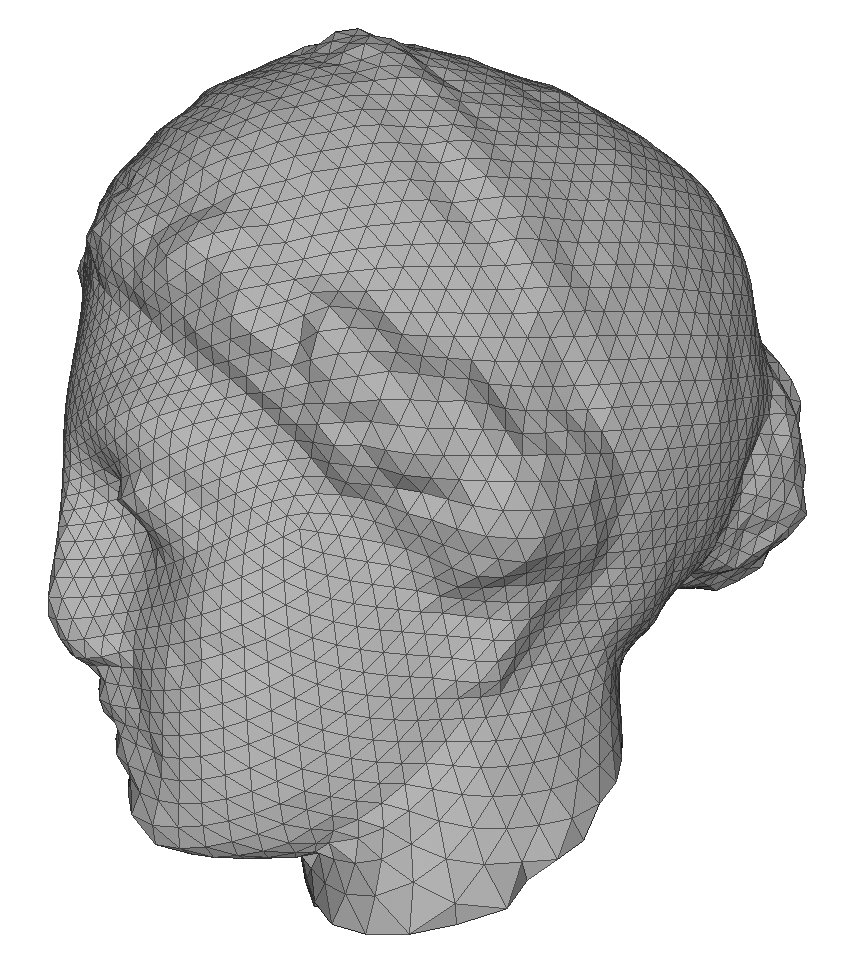}
 \includegraphics[width=0.27\textwidth]{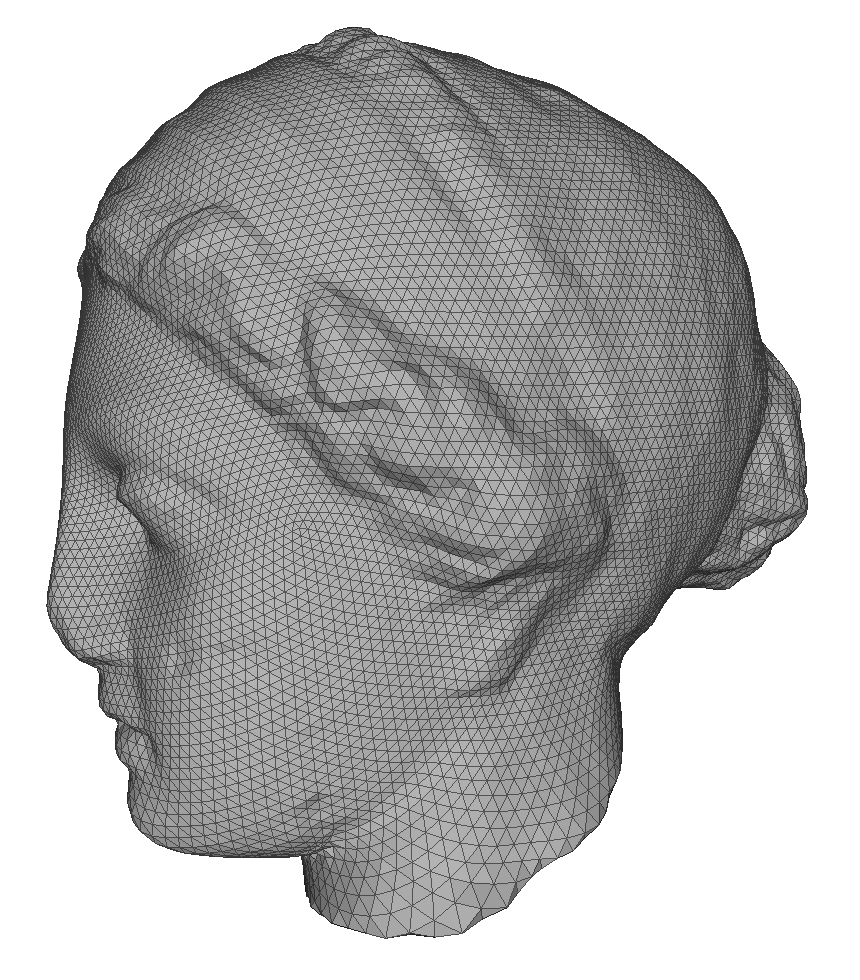}
 \includegraphics[width=0.27\textwidth]{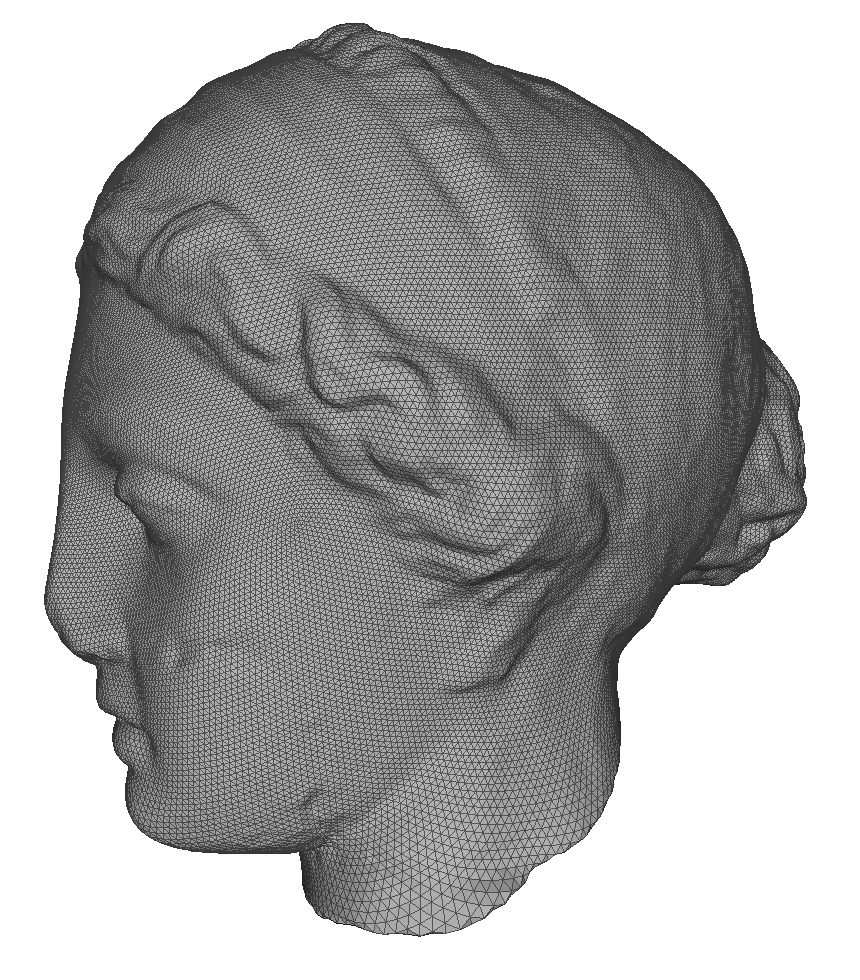}
 \caption{Multilevel representations of a genus-0 point cloud of Igea. For better visualizations, we create mesh structures on the representations. Top left: a point cloud with 134345 points of Igea. Top middle to Bottom right: the multilevel representations with 0, 1, 2, 3 and 4 subdivisions. The representations are with 642, 2562, 10242, 40962 and 163842 points respectively.}
 \label{fig:igea_multilevel}
\end{figure}

\begin{figure}[t]
 \centering
 \includegraphics[width=0.27\textwidth]{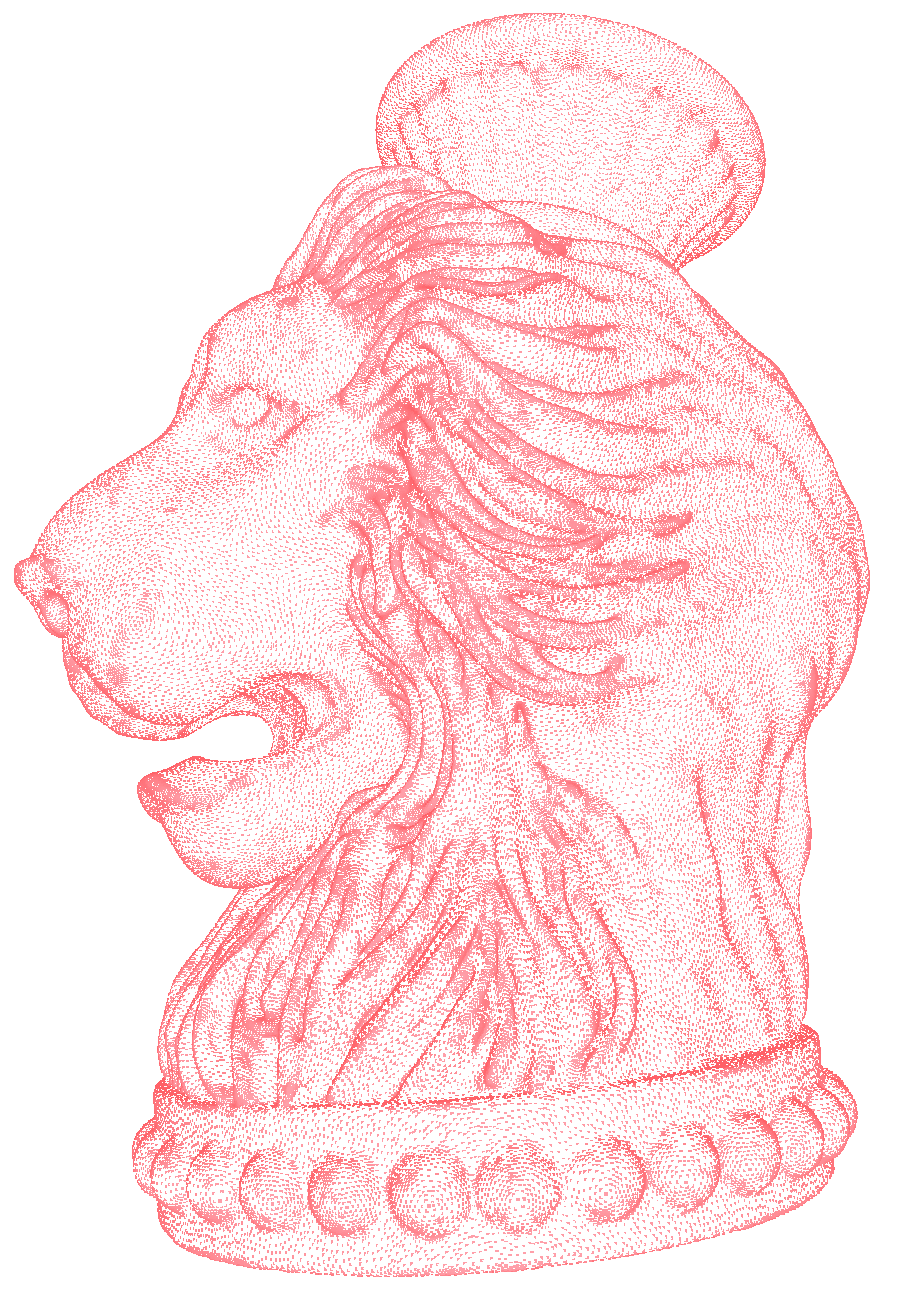}
 \includegraphics[width=0.27\textwidth]{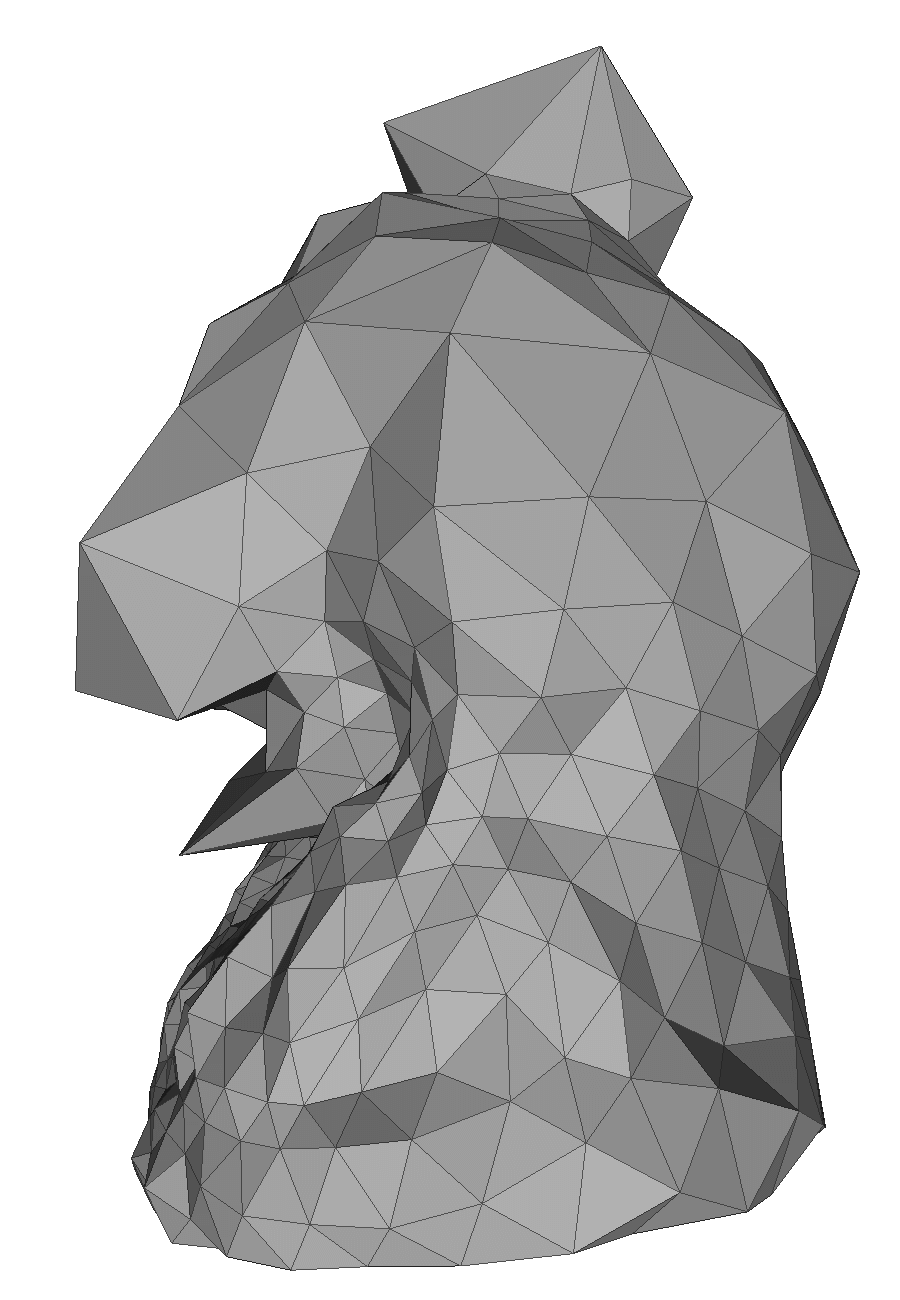}
 \includegraphics[width=0.27\textwidth]{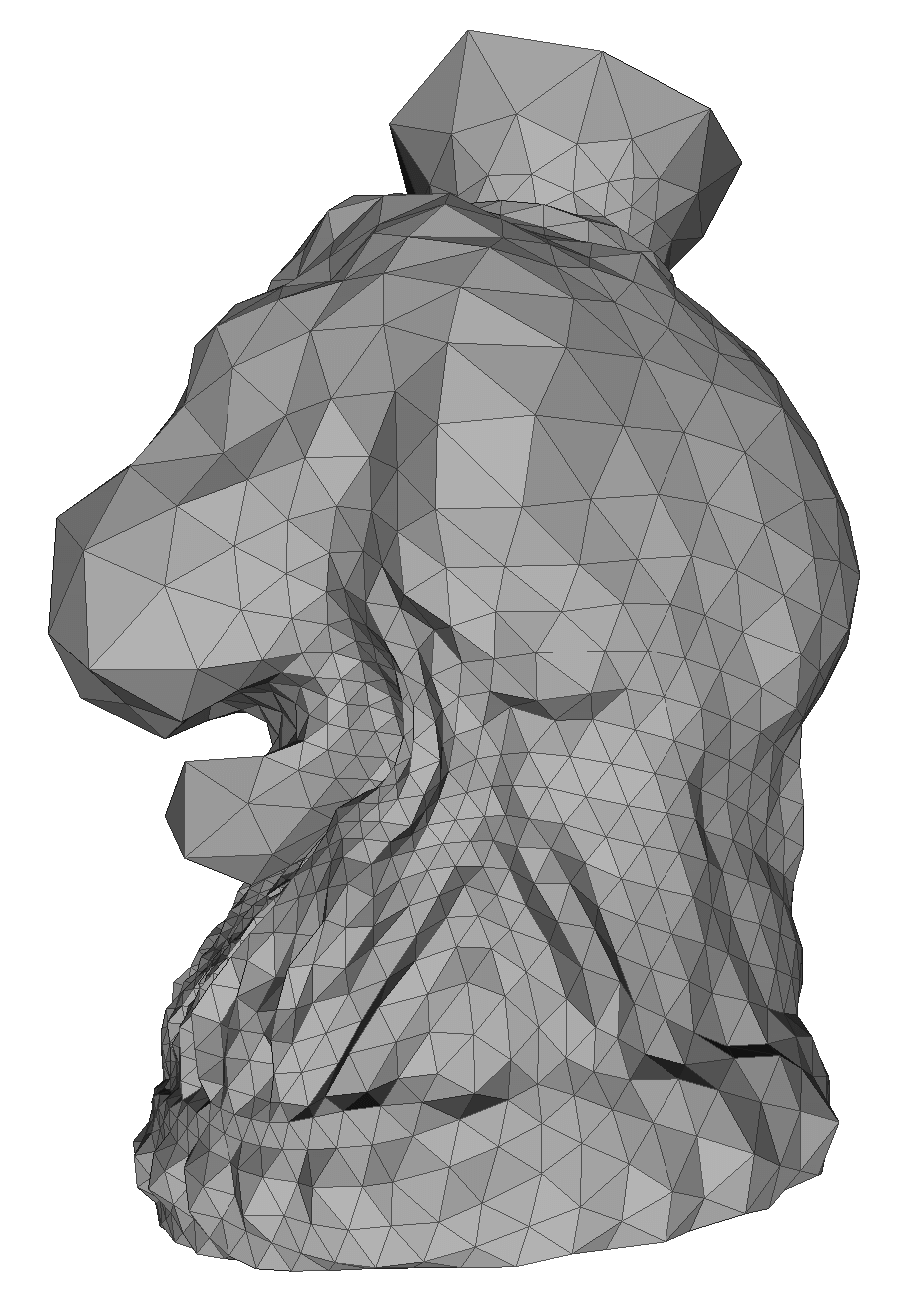}
 \includegraphics[width=0.27\textwidth]{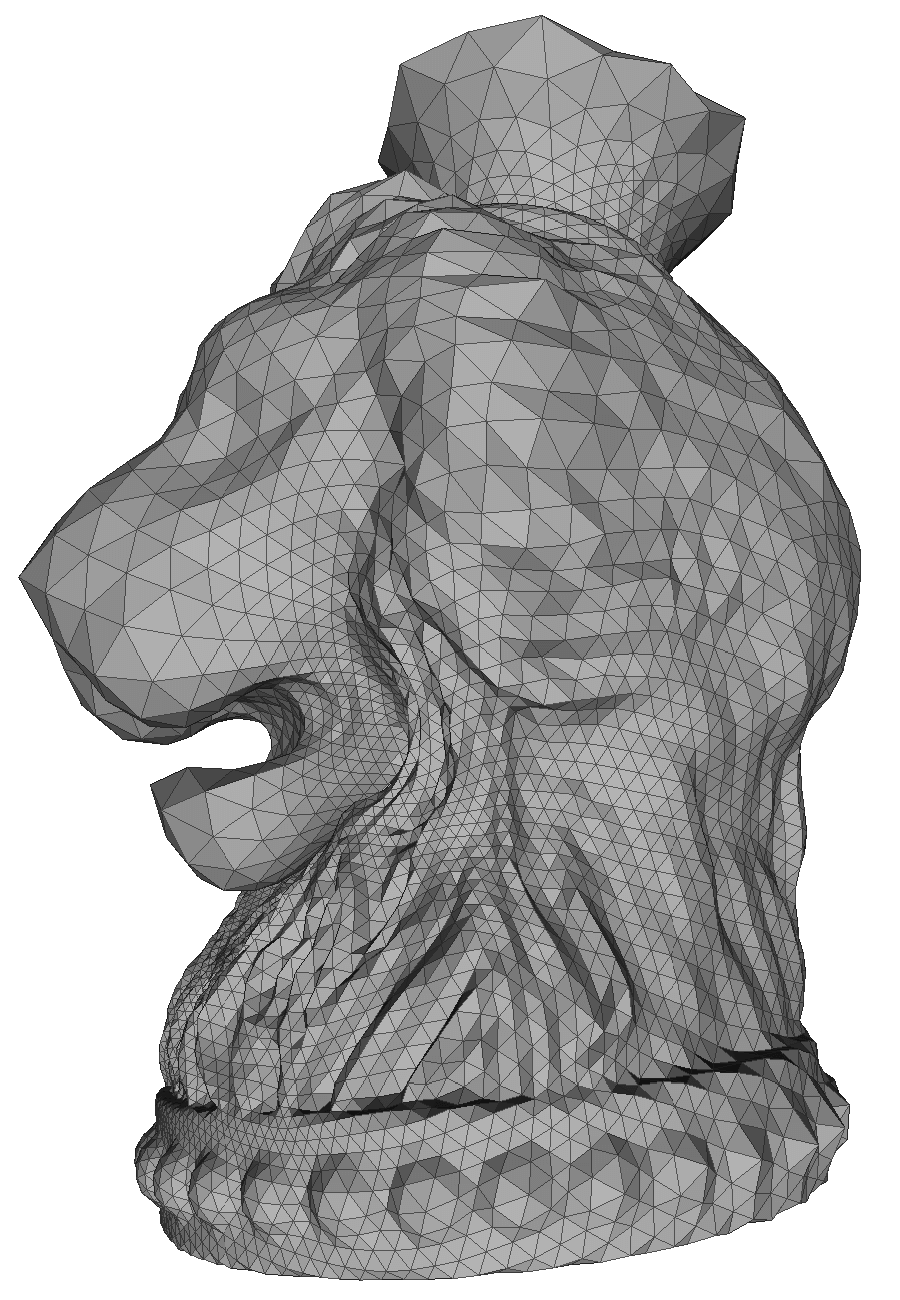}
 \includegraphics[width=0.27\textwidth]{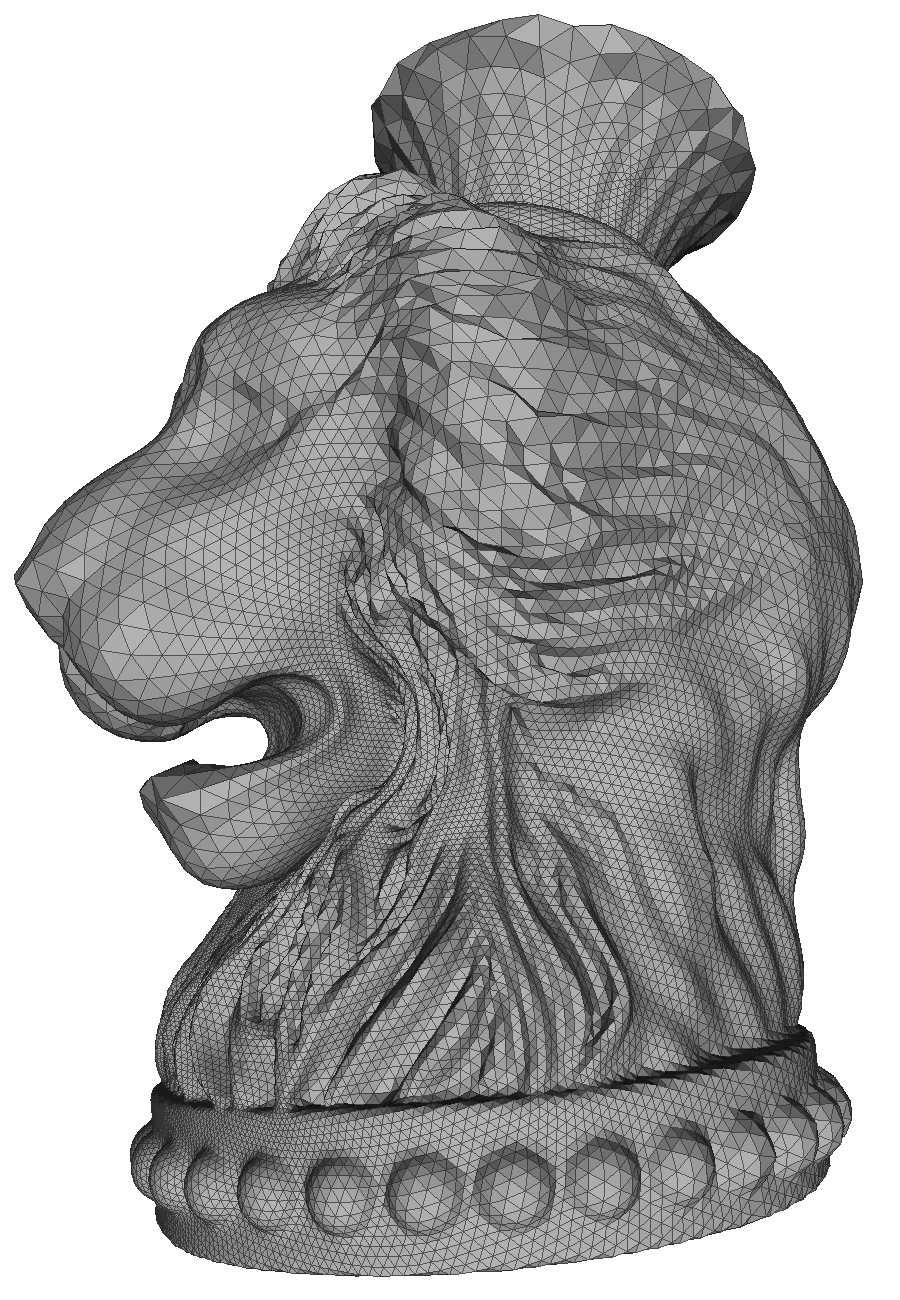}
 \includegraphics[width=0.27\textwidth]{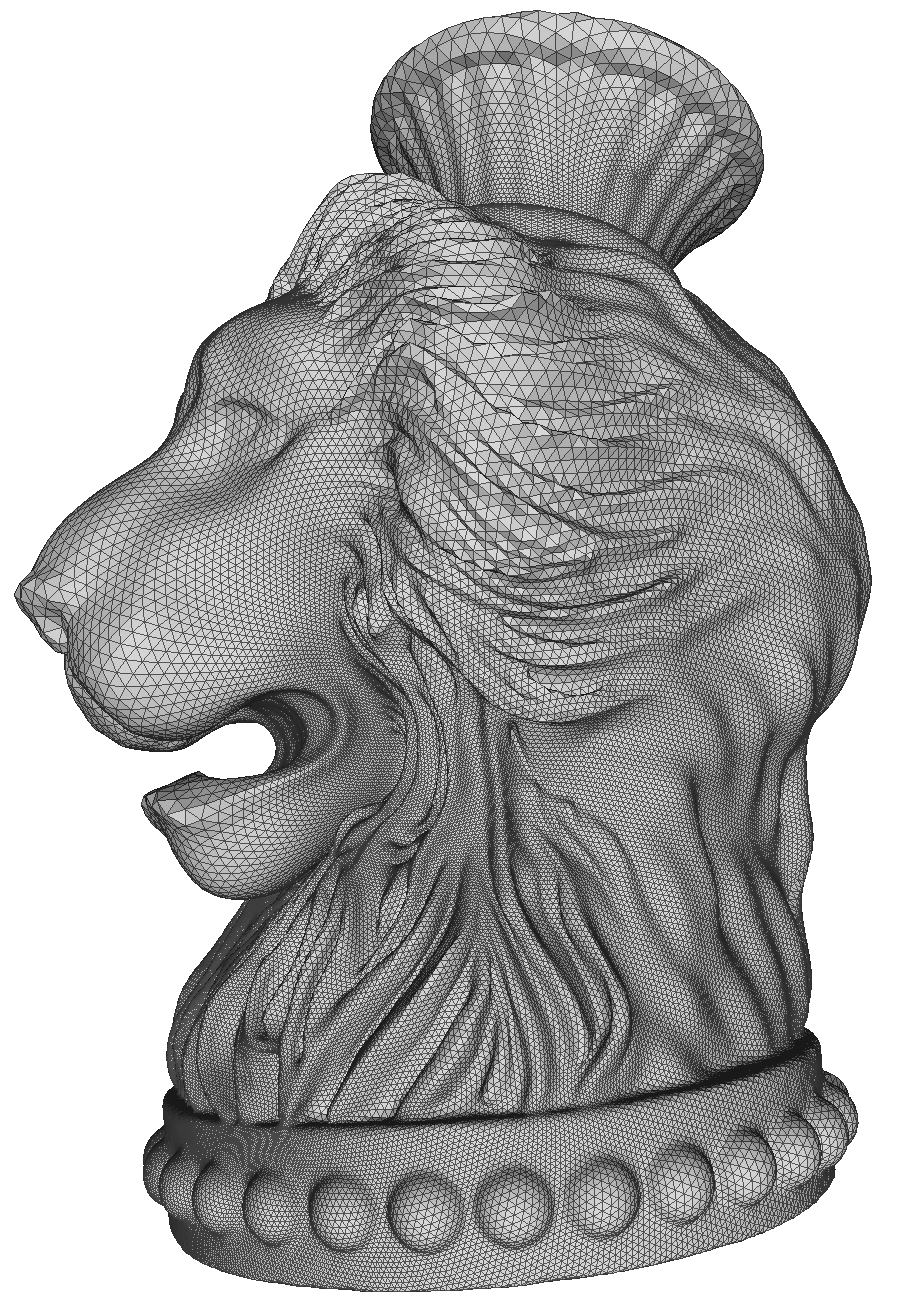}
 \caption{Multilevel representations of a genus-0 lion vase point cloud. For better visualizations, we create mesh structures on the representations. Top left: a lion vase point cloud with 256094 points. Top middle to Bottom right: the multilevel representations with 0, 1, 2, 3 and 4 subdivisions. The representations are with 696, 2778, 11106, 44418 and 177666 points respectively.}
 \label{fig:vaselion_multilevel}
\end{figure}

With our proposed spherical conformal parameterization scheme, multilevel representations of a genus-0 point cloud can be easily achieved. We start with a coarse spherical point cloud. The vertices on the sphere can be interpolated onto the genus-0 point cloud with the aid of its spherical parameterization. Then, we can progressively subdivide the sphere using existing subdivision methods, such as the butterfly subdivision method \cite{Dyn90} and the loop subdivision method \cite{Loop87}. For each subdivided sphere, we can repeat the mentioned interpolation procedure and obtain a coarse representation of the point cloud. This method results in multilevel representations of the point cloud. As the subdivision level increases, more details of the point cloud are represented. Examples of multilevel representations of genus-0 point clouds are given in Figure \ref{fig:igea_multilevel} and Figure \ref{fig:vaselion_multilevel}. In our examples, the subdivisions are generated using the loop subdivision method \cite{Loop87}. The subdivision connectivity of the results can be easily observed. The results indicate that our method can effectively generate the multilevel representations of genus-0 point clouds.

\section{Conclusion and Future Work} \label{conclusion}
In this paper, we presented a novel framework for meshing genus-0 point clouds via global spherical conformal parameterizations. We extended and improved the parameterization algorithm for triangular meshes in \cite{Choi15a}. Firstly, we enhanced the accuracy for approximating the LB operator on point clouds by using a new Gaussian-type weight function. Secondly, we proposed an iterative scheme called the N-S reiteration to replace the step of solving for a quasi-conformal map in \cite{Choi15a} for achieving better conformality. Thirdly, we introduced a balancing scheme for guaranteeing an even distribution of the spherical point cloud parameterization. Experimental results show that our proposed algorithm is highly efficient and accurate. With the aid of the spherical conformal parameterizations, almost-Delaunay triangulations and high quality quadrangulations of genus-0 point clouds can be effectively created. The meshes generated are guaranteed to be of genus-0 and no post-processing is needed. Besides, our meshing method is stable under geometrical and topological noises on point clouds. Moreover, multilevel representations of genus-0 point clouds can be easily computed. As a remark, our proposed spherical conformal parameterization algorithm also works efficiently on triangular meshes. In the future, we plan to establish a rigorous theoretical proof of the convergence of our parameterization scheme, and extend our method to handle disk-type point clouds and point clouds with arbitrary topology.

\end{document}